\documentclass[%
superscriptaddress,
twocolumn,
nofootinbib,
 aps,pra,
]{revtex4-2}

\usepackage{float}
\makeatletter
\let\newfloat\newfloat@ltx
\makeatother
\usepackage[english]{babel}
\usepackage[normalem]{ulem}

\makeatletter 
\renewcommand\onecolumngrid{
\do@columngrid{one}{\@ne}
\def\set@footnotewidth{\onecolumngrid}
\def\footnoterule{\kern-6pt\hrule width 1.5in\kern6pt}
}
\renewcommand\twocolumngrid{
        \def\footnoterule{
        \dimen@\skip\footins\divide\dimen@\thr@@
        \kern-\dimen@\hrule width.5in\kern\dimen@}
        \do@columngrid{mlt}{\tw@}
}
\makeatother

\usepackage{comment}
\usepackage{inputenc}
\usepackage{graphics}
\usepackage{selinput}
\usepackage{bbm}
\usepackage{multirow}
\usepackage{braket}
\usepackage{amsthm}
\usepackage{mathtools}
\usepackage{bm}
\usepackage{physics}
\usepackage[dvipsnames]{xcolor}
\usepackage{graphicx}
\usepackage[left=16mm,right=16mm,top=35mm,columnsep=15pt]{geometry} 
\usepackage{adjustbox}
\usepackage{placeins}
\usepackage[T1]{fontenc}
\usepackage{lipsum}
\usepackage{csquotes}
\usepackage{mathbbol} %
\usepackage[linesnumbered,ruled,vlined]{algorithm2e}
\usepackage[tight,footnotesize]{subfigure}

\SetKwInput{kwInit}{Init}

\def\P{\mathbb{P}}

\def\HC{\mathcal{H}}

\def\LC{\mathcal{L}}

\def\ad{^{\dagger}}

\usepackage[makeroom]{cancel}
\usepackage[toc,page]{appendix}
\usepackage[colorlinks=true,citecolor=blue,linkcolor=magenta]{hyperref}

\usepackage{tikz}
\tikzset{every picture/.style=remember picture}

\newcommand{\poly}{\operatorname{poly}}

\newcommand{\Ebb}{\mathbb{E}}

\newcommand{\Ubb}{\mathbb{U}}

\newcommand{\AC}{\mathcal{A}}
\newcommand{\BC}{\mathcal{B}}

\newcommand{\IC}{\mathcal{I}}
\newcommand{\MC}{\mathcal{M}}
\newcommand{\NC}{\mathcal{N}}
\newcommand{\OC}{\mathcal{O}}
\newcommand{\PC}{\mathcal{P}}
\newcommand{\QC}{\mathcal{Q}}

\newcommand{\SC}{\mathcal{S}}
\newcommand{\TC}{\mathcal{T}}
\newcommand{\UC}{\mathcal{U}}

\newcommand{\XC}{\mathcal{X}}

\newcommand{\Var}{{\rm Var}}
\newcommand{\Cov}{{\rm Cov}}

\renewcommand{\geq}{\geqslant}
\renewcommand{\leq}{\leqslant}

\newcommand{\ot}{\otimes}

\newcommand{\bs}{\textsf{BS}}

\newcommand{\sg}{\sigma }

\newcommand{\thv}{\vec{\alpha}}

\newcommand{\Xv}{\vec{X}}

\newcommand{\Yv}{\vec{Y}}
\newcommand{\sv}{\vec{s}}
\newcommand{\mv}{\vec{m}}

\newcommand{\pplus}{p_{k;t}^{(+)}}
\newcommand{\pminus}{p_{k;t}^{(-)}}

\newtheorem{innercustomthm}{Theorem}
\newenvironment{customthm}[1]
  {\renewcommand\theinnercustomthm{#1}\innercustomthm}
  {\endinnercustomthm}

\def\be{\begin{equation}}
\def\ee{\end{equation}}

\def\bs{\begin{split}}
\def\es{\end{split}}
\def\bea{\begin{eqnarray}}
\def\eea{\end{eqnarray}}

\def\P{\mathbb{P}}

\def\t{\tau}

\def\P{\mathbb{P}}

\newtheorem{theorem}{Theorem}
\newtheorem{lemma}{Lemma}
\newtheorem{corollary}{Corollary}

\newtheorem{proposition}{Proposition}
\newtheorem*{proposition*}{Proposition}
\newtheorem{supplemental_proposition}{Supplemental Proposition}
\newtheorem{supplemental_corollary}{Supplemental Corollary}
\newtheorem{supplemental_theorem}{Supplemental Theorem}

\newtheorem{definition}{Definition}

\newenvironment{specialproof}{\textit{Proof:}}{\hfill$\square$}

\usepackage{amssymb}
\usepackage{dsfont}

\usepackage[normalem]{ulem}
\usepackage{hyperref}

\DeclareRobustCommand{\tcdel}[1]{}

\newcommand{\stadd}[1]{{#1}}

\usepackage{ulem}

\newcommand{\rhot}[1]{\rho_{#1}} %
\newcommand{\rhoa}[1]{\rho_{a_{#1}}} 
\newcommand{\rhoh}[1]{\rho_{h_{#1}}} 

\newcommand{\vrhoa}[1]{\vec{\rho}_{a_{#1}}} 

\newcommand{\sga}[1]{\sigma_{a_{#1}}} 
\newcommand{\sgh}[1]{\sigma_{h_{#1}}}
\newcommand{\sgt}[1]{\sigma_{#1}}%

\usepackage{minitoc}
\usepackage[subfigure]{tocloft}
\renewcommand \partname{}
\renewcommand{\arraystretch}{1.5}

\usepackage{booktabs}
\usepackage{longtable}
\usepackage{array}

\usepackage{MnSymbol}

\renewcommand{\vec}[1]{\boldsymbol{#1}}  %
\begin{document}
\doparttoc %
\faketableofcontents %

\title{Role of scrambling and noise in temporal information processing with quantum systems}

\author{Weijie Xiong}
\affiliation{Institute of Physics, Ecole Polytechnique F\'{e}d\'{e}rale de Lausanne (EPFL),  Lausanne, Switzerland}
\affiliation{Centre for Quantum Science and Engineering, Ecole Polytechnique F\'{e}d\'{e}rale de Lausanne (EPFL),  Lausanne, Switzerland}
\email{weijie.xiong@epfl.ch}

\author{Zo\"{e} Holmes}
\affiliation{Institute of Physics, Ecole Polytechnique F\'{e}d\'{e}rale de Lausanne (EPFL),  Lausanne, Switzerland}
\affiliation{Centre for Quantum Science and Engineering, Ecole Polytechnique F\'{e}d\'{e}rale de Lausanne (EPFL),  Lausanne, Switzerland}

\author{Armando Angrisani}
\affiliation{Institute of Physics, Ecole Polytechnique F\'{e}d\'{e}rale de Lausanne (EPFL),  Lausanne, Switzerland}
\affiliation{Centre for Quantum Science and Engineering, Ecole Polytechnique F\'{e}d\'{e}rale de Lausanne (EPFL),  Lausanne, Switzerland}

\author{Yudai~Suzuki}
\affiliation{Institute of Physics, Ecole Polytechnique F\'{e}d\'{e}rale de Lausanne (EPFL),  Lausanne, Switzerland}
\affiliation{Centre for Quantum Science and Engineering, Ecole Polytechnique F\'{e}d\'{e}rale de Lausanne (EPFL),  Lausanne, Switzerland}

\author{Thiparat Chotibut}
\affiliation{Chula Intelligent and Complex Systems Lab, Department of Physics, Faculty of Science, Chulalongkorn University, Bangkok, Thailand}
\affiliation{School of Physical and Mathematical Sciences, Nanyang Technological University, Singapore}

\author{Supanut Thanasilp}
\affiliation{Institute of Physics, Ecole Polytechnique F\'{e}d\'{e}rale de Lausanne (EPFL), Lausanne, Switzerland}
\affiliation{Centre for Quantum Science and Engineering, Ecole Polytechnique F\'{e}d\'{e}rale de Lausanne (EPFL),  Lausanne, Switzerland}
\affiliation{Chula Intelligent and Complex Systems Lab, Department of Physics, Faculty of Science, Chulalongkorn University, Bangkok, Thailand}

\date{\today}

\begin{abstract}
Scrambling quantum systems have attracted attention as effective substrates for temporal information processing. 
Here we consider a quantum reservoir processing framework that captures a broad range of physical computing models with quantum systems. We examine the scalability and memory retention of the model with scrambling reservoirs modelled by high-order unitary designs in both noiseless and noisy settings. In this regime, we show that measurement readouts become exponentially concentrated with increasing reservoir size, yet strikingly do not worsen with the reservoir iterations. Thus, while repeatedly reusing a small scrambling reservoir with quantum data might be viable, scaling up the problem size deteriorates generalization unless one can afford an exponential shot overhead. In contrast, the memory of early inputs and initial states decays exponentially in \emph{both} reservoir size and reservoir iterations. %
In the noisy regime, we also prove that memory decays exponentially in time for local noisy channels. These results required us to introduce new proof techniques for bounding concentration in temporal quantum models. \stadd{Beyond this extreme scrambling regime, we numerically demonstrate that exponential concentration can still exist even with a physical reservoir such as an Ising model whenever the reservoir operates in a quantum-chaotic phase. In contrast, physical reservoirs in a many-body localized phase and at the edge of chaos appear to not suffer from such phenomena.}
\end{abstract}

\maketitle

\section{Introduction}
Physical computing is an emerging computational paradigm in which the native, continuous dynamics of a physical substrate is used directly as part of a computational process, rather than being fully abstracted into digital logic~\cite{mcmahon2023thephysics,mohseni2022ising,markovic2020physics,stern2023learning,wetzstein2020inference}. This approach can offer adaptability and, in some implementations, substantial energy savings, both of which are critical for large-scale artificial intelligence~\cite{wetzstein2020inference,seok2024beyond,hasler2021physical}. Among such substrates, quantum systems are especially appealing due to their large Hilbert spaces and the complex spatio-temporal correlations they naturally generate, making them well-suited for information processing tasks involving spatially or temporally correlated data~\cite{suzuki2022naturalq,chen2020temporal}. Recent experiments indeed show that reusing a quantum system repeatedly can enhance performance in time-series prediction, suggesting a \emph{recurrent} or \emph{iterative} mode of operation that interleaves these quantum systems with streams of quantum data~\cite{senanian2024microwave,kornjavca2024large}.

Quantum reservoir processing (QRP)~\cite{fujii2017harnessing,nakajima2019boosting, martinez2021dynamical, mujal2021opportunities,ghosh2020reconstructing,govia2021quantum,hu2023tackling} exemplifies this iterative physical computing paradigm. Building upon the input-feedback strategy in classical reservoir computing ~\cite{jaeger2001echo,maass2002real,lukosevicius2009reservoir}, this approach has recently gained considerable attention as a
compelling non-variational framework for leveraging quantum substrates for temporal information processing~\cite{fujii2017harnessing,nakajima2019boosting,mujal2021opportunities,ghosh2020reconstructing,govia2021quantum,hu2023tackling}.
At its core, QRP processes input states using a fixed quantum evolution (the \emph{reservoir}) to generate a rich feature map. The reservoir-processed states are then read out via quantum measurements and the resulting vector of expectation values are trained classically via linear regression to solve a given learning task. This approach can tackle not only standard machine learning problems with classical data~\cite{burgess2022quantum,settino2024memory,pfeffer2022hybrid,ahmed2024prediction,chen2024efficient}, but also intrinsically quantum problems such as state discrimination and tomography~\cite{angelatos2021reservoir,ghosh2020reconstructing,kobayashi2024feedback,li2024estimating,vetrano2025state,nokkala2024retrieving} and quantum dynamics prediction~\cite{sornsaeng2023quantum}. Experimental realizations across various platforms further highlight QRP's viability~\cite{yasuda2023quantum,dudas2023quantum,govia2021quantum,kornjavca2024large,garcia2023scalable,hu2024overcoming}, while a growing benchmarking studies chart its performance~\cite{fujii2017harnessing,kalfus2022hilbert,khan2021physical,nakajima2019boosting,martinez2020information,martinez2021dynamical,ghosh2020reconstructing,nokkala2021gaussian,pfeffer2022hybrid,fry2023optimizing,senanian2024microwave,martinez2023quantum,wright2020capacity,govia2022nonlinear,kubota2022quantum,suzuki2022naturalq,domingo2023taking,kubota2021unifying,garcia2023scalable} and hints towards the potential to rival large classical neural networks~\cite{fujii2016power,burgess2022quantum,pfeffer2022hybrid}.

Many QRP schemes utilize quantum scrambling under chaotic dynamics as a quantum reservoir to create a natural, high-dimensional feature map for input data~\cite{fujii2017harnessing,nakajima2019boosting,martinez2021dynamical}. Numerous studies highlight correlations between the degree of scrambling of the reservoir and the QRP performance~\cite{martinez2021dynamical,palacios2024role,vetrano2025state,martinez2020information,hu2024overcoming}. Yet existing work is confined to small problem sizes \cite{yasuda2023quantum,dudas2023quantum,govia2021quantum,kornjavca2024large,garcia2023scalable,burgess2022quantum,pfeffer2022hybrid}, thus it remains unclear whether these promising results survive when the reservoir size enters the regime where real quantum advantages might appear. In the absence of large-scale, high-quality quantum hardware, theoretical analysis is the only viable avenue to study these phenomena on larger scales.
\begin{figure*}
    \centering
    \includegraphics[width=0.99\textwidth]{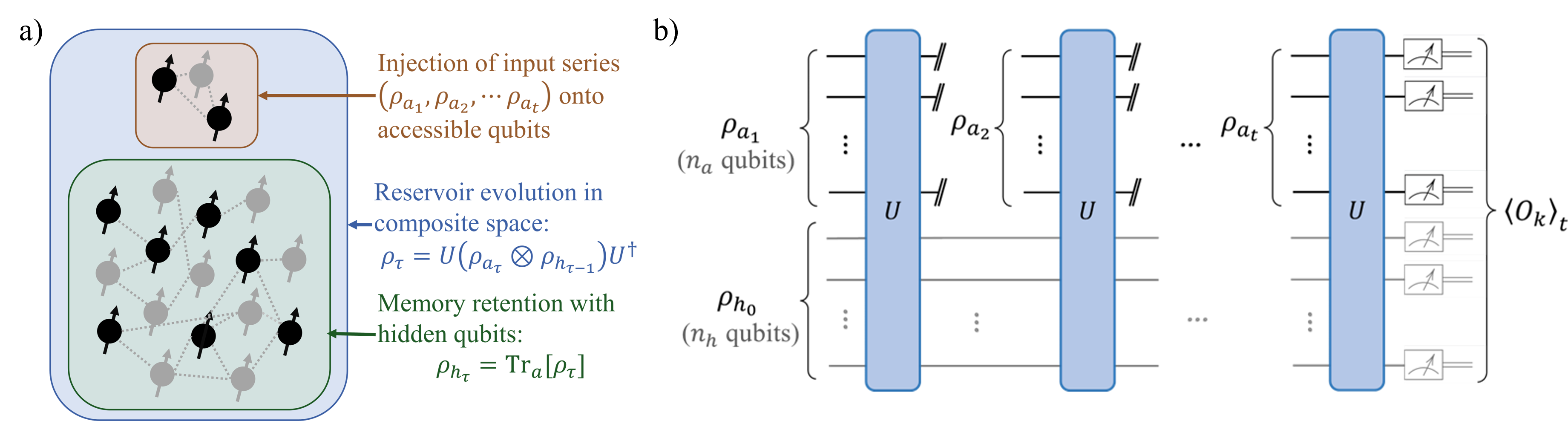}
    
    \caption{\textbf{{Framework of QRP (a) and its circuit diagram (b)}.} The quantum reservoir consists of $n=n_a+n_h$ qubits. $n_a$ denotes the number of accessible qubits, which are used for data uploading; and $n_h$ the number of hidden qubits. After injecting the input states into accessible qubits, the reservoir dynamics, described by some unitary $U$, will be applied to the composite system of accessible and hidden qubits. The key component of a RC system is its memory. In QRP, while the accessible qubits will be initialized after each iteration and subsequently be encoded with data of next time step, the hidden qubits will never be initialized and their state are obtained from the reduced state of the last reservoir evolution in hidden space. Thus, the hidden qubits play the role of quantum memory and store the information of previous inputs.
    }
    \label{fig:Framework_QRP}
\end{figure*}

In this work, we idealize scrambling reservoirs as unitaries drawn from a high-order design ensemble, to capture the extreme-scrambling limit while enable rigorous theoretical analysis with Haar integration~\cite{holmes2020barren,roberts2017chaos}. Related techniques already establish scalability barriers for variational quantum algorithms~\cite{mcclean2018barren, larocca2024review}, quantum kernels~\cite{thanasilp2022exponential}, and quantum extreme learning machines (QELMs), the single-pass variant of QRP such that the reservoir evolution is applied only once \cite{xiong2025fundamental,sannia2025ExpConcQRC}. By contrast, the full QRP protocol is much more demanding: a fixed reservoir repeatedly interleaves with a stream of input time-series. This temporal correlation between reservoir iterations hinders standard analytical techniques. To address this challenge, we apply tensor-diagram approaches to unroll multi-step QRP into a single high-moment Haar integral amenable for scalability analysis. This allows us to obtain main results summarized in Tab.~\ref{tab:summary}.

Using the `unravelled' tensor-diagram approach, we derive an asymptotic concentration bound for the QRP's output magnitudes. Our bounds reveal the exponential concentration of the output with increasing reservoir size (i.e., the number of qubits), yet incurs no further concentration over successive reservoir iterations. On the one hand, this exponential concentration presents a scalability barrier- their outputs for distinct inputs become effectively data independent for large scale problems.
On the other hand, as reservoir iterations cost no further concentration,  repeatedly reusing a modest-sized scrambling reservoir remains effective with a feasible number of measurement shots. This potentially leaves room for practical data processing advantages for quantum data in the spirit of physical computing. 

We also examine how a scrambling reservoir retains memory of its past. For this, we introduce a \emph{memory indicator} that quantifies how strongly outputs depend on earlier inputs or initial reservoir states. Previous studies have indicated that the memory of QRP generally decays exponentially over reservoir iterations~\cite{garcia2023scalable, sannia2025non}. Here we compute the exact rate of memory decay for a fully scrambling reservoir. {Viewed positively, this temporal decay represents Fading Memory Property (FMP)-like forgetting of early inputs and Echo State Property (ESP)-like washout of initial conditions at the level of our memory indicators. However, the memory indicators also provably  decay exponentially with reservoir size, ultimately rendering the model input-insensitive at large sizes.} 
Thus, while modest scale QRP with fully scrambling reservoirs could remain viable, QRP with large highly scrambling reservoirs will be unusable in practice as it becomes input-agnostic and its memory fades exponentially fast. {That is, while modest FMP/ESP is good, but too much could lead to detrimental effects - not due to non-unitarity itself, but that contraction becomes detrimental when its washout timescale is shorter than the memory timescale required by the task.}

Realizing QRP on near-term quantum devices unavoidably involves hardware noise. We further extend our analysis to scrambling reservoirs with both unital and non-unital noises. We derive upper bounds showing that the reservoir loses the memory of an initial state exponentially with increasing temporal iterations. While recent evidence suggests certain forms of dissipation or decoherence can be used as a computational resource for improved performance on small QRP~\cite{martinez2023quantum,martinez2024input,sannia2024dissipation, cheamsawat2025dissipation, monzani2024leveraging, yosifov2025dissipation}, our analysis explicitly shows how those generic noise channels can compound concentration effects over multiple iteration steps,  leading to vanishing memory and scalability barrier for QRP in the presence of noise.

\stadd{Finally, we ask to what extent these scalability barriers, derived in the extreme-scrambling limit, carry over to other more physical reservoirs. To this end, we study output concentration in two families of physical reservoirs, including (i)~brickwork-circuit reservoirs with tunable depth, and (ii)~dynamical reservoirs generated by experimentally motivated many-body Hamiltonians, covering the disordered transverse-field Ising model~\cite{fujii2017harnessing,negoro2018machine,nakajima2019boosting,kutvonen2020optimizing,martinez2021dynamical,li2026quantum,hou2026high,ivaki2024quantum,settino2025topology,mccaul2025minimal,cindrak2024enhancing,delorenzis2025harnessing,fuchs2024quantum,xia2023configured}, %
the Sachdev--Ye--Kitaev (SYK) model~\cite{kobayashi2026edge} and kicked-Ising dynamics~\cite{dao2026breaking,zhang2026robust} as summarized in Table~\ref{tab:QRP_reservoirs_main}.
Using spectral statistics to identify the reservoir's dynamical phase diagram, we find that exponential concentration emerges whenever the reservoir operates deep in the ergodic quantum-chaotic phase and is evolved sufficiently long between input injections. The level of exponential concentration tracks the degree of chaos. By contrast, reservoirs in localized, near-critical, or shallow-circuit regimes exhibit markedly milder concentration, albeit at the potential cost of reduced expressive power and potentially classical simulability. 
Notably, the edge-of-chaos regime appears not to suffer from exponential concentration, identifying it as a potential operating regime for QRP in line with Ref.~\cite{kobayashi2026edge,martinez2021dynamical}.}

\section{Preliminaries}
\subsection{Framework of QRP}\label{sec:QRP_framework}

To systematically examine quantum reservoir processing (QRP) across diverse experimental platforms, we recast it in a unified framework that naturally accommodates both quantum and classical inputs. This abstraction captures the essential structure encompassing most current proposals for quantum reservoir computing, allowing for the applicability of our results across quantum substrates.

Concretely, we will consider a dataset $\SC$ containing $D$ time series, each of  length~$t$, denoted
by $\bigl\{\vrhoa{t}^{(l)}\bigr\}_{l=1}^{D}$.
Every sequence of input quantum states
$
  \vrhoa{t}^{(l)}
  =\bigl(\rhoa{1}^{(l)},\rhoa{2}^{(l)},\dots,\rhoa{t}^{(l)}\bigr)
$
is associated with a target output sequence
$\vec{y}_{t}^{(l)}=(y_{1}^{(l)},\dots,y_{t}^{(l)})$.
The input sequences can be intrinsic quantum data such as time-dependent density matrices generated from quantum dynamics ~\cite{cerezo2022challenges}, or they can
arise from classical data $s_\tau$  (e.g. Mackey-Glass chaotic sequences) that have been encoded through a unitary
$U_{\mathrm{enc}}(s_\tau)$ acting on a fixed reference state.

Given a training subset
\(
  \mathcal{S}_{\rm tr}=\bigl\{(\vrhoa{t}^{(l)},\vec{y}_{t}^{(l)})\bigr\}_{l=1}^{D_{\mathrm{tr}}}
  \subseteq\SC,
\)
our goal is to predict an output $y_{\tau}$ associated with a previously unseen input $\rhoa{\tau} \notin \mathcal{S}_{\rm tr}$.
For simplicity, we assume the output is scalar, though it may generally be vector-valued.

Concretely, the quantum reservoir architecture consists of
\begin{itemize}
  \item $n_a$ \emph{accessible qubits}, onto which each new input $\rhoa{\tau}$ is loaded at time $\tau$ (replacing the previous accessible-qubit state),
  \item $n_h$ \emph{hidden qubits}, which accumulate information from past inputs and hidden-qubit states. We denote their initial state by $\rhoh{0}$, and throughout iterations the reservoir is never reinitialized. 
\end{itemize}

At each time step $\tau$, the QRP model recursively processes the data stream as follows (see  Fig.~\ref{fig:Framework_QRP}):
\begin{enumerate}
\item[\textbf{(i)}] \textbf{Data injection:} Replace the state of the $n_a$ accessible qubits by the new input $\rhoa{\tau}$. The hidden qubits remain in the state carried over from the previous iteration.
\item[\textbf{(ii)}] \textbf{Reservoir evolution:} 
   Apply a \emph{fixed} reservoir evolution $U$ acting on all $n_a + n_h$ qubits (here we assume a unitary channel for simplicity, although we will consider more general noisy channels in subsequent sections.):
   \begin{equation}\label{eq:qrp-total-state}
      \rhot{\tau}
        =
         U\bigl(\rhoa{\tau} \,\otimes\, \rhoh{\tau-1}\bigr)\,
         U^\dagger,
   \end{equation}
   where $\rhoh{\tau-1} = \mathrm{Tr}_a\bigl[\rhot{\tau-1}\bigr]$ is the hidden-qubit state from the previous time step, and $\rhoh{0}$ denotes the initial hidden-qubit state.
\end{enumerate}
After processing the input sequence by repeating (i) and (ii) up to time step $\tau \leq t$, one can choose to perform measurement readout to obtain QRP outputs. 
\begin{enumerate}
\item[\textbf{(iii)}] \textbf{Measurement readout:} 
   Measures a fixed set of observables $\{O_k\}_{k=1}^M$ on the state $\rhot{\tau}$. The QRP output associated with the $k^{\rm th}$ component of the total readout is then
   \begin{align}\label{obs_def_QRP}
       \langle O_k \rangle_\tau = \Tr\left[ O_k \rhot{\tau} \right] \;,
   \end{align} 
   leading to the ideal (shot noise-free) readout of the form
   \begin{equation}
     \bm{r}_\tau 
       =
         \Bigl(\mathrm{Tr}[O_1\,\rhot{\tau}],\,
                \dots,\,
                \mathrm{Tr}[O_M\,\rhot{\tau}]
         \Bigr).
   \end{equation}
   In practice, however, these outputs are estimated with finite measurement shots leading to limited precision due to shot noise~\cite{hu2023tackling}.
\item[\textbf{(iv)}] \textbf{Classical post-processing:}
To approximate the label $y_\tau$, one typically forms a linear combination of these measured QRP observables, namely
   \begin{equation}\label{eq:qrp-model-prediction}
     f_{\t}\bigl(\vrhoa{\t};\bm{\eta}\bigr)
       =
         \sum_{k=1}^M\,\eta_k\,
         \mathrm{Tr}\bigl[\,O_k\,\rhot{\tau}\bigr]\;,
   \end{equation}
   where $\vrhoa{\t}=(\rhoa{1},\rhoa{2},\dots,\rhoa{\t})$ and $\t\leq t$ specifies the number of iterations of QRP, $\bm{\eta} = (\eta_1,..., \eta_M)\in \mathbb{R}^M$ are trainable weights and we recall that $\rhot{\tau}$ implicitly depends on $\vrhoa{\t}$. 

The optimal weights are then obtained by classically minimising an empirical loss,  $\bm{\eta}^{*}
=\arg\min_{\bm{\eta}}[\mathcal{L}(\bm{\eta})],$
which, for suitable choices of the loss $\mathcal{L}$ (e.g.\ regularised
least-squares), reduces to standard linear regression and therefore admits a
closed-form solution. For example, one could use the mean-squared error loss
\begin{equation}\label{eq:loss}
  \mathcal{L}(\bm{\eta})
  =
  \frac{1}{|\mathcal{T}_{\mathrm{tr}}|}
  \sum_{\tau\in\mathcal{T}_{\mathrm{tr}}}
  \mathcal{L}_{\tau}(\bm{\eta}),
\end{equation}
where $\mathcal{T}_{\mathrm{tr}}\subseteq[t]$ specifies the time steps used
for training -- typically those after a wash-out period,
$\mathcal{T}_{\mathrm{tr}}=\{t_{\mathrm{w}},t_{\mathrm{w}}+1,\dots\}$ with
$t_{\mathrm{w}}\gg1$~\cite{suzuki2022naturalq} -- and
   \begin{equation}
    \mathcal{L}_\t(\bm{\eta})
       =
       \frac{1}{D_{\rm tr}}\sum_{(\vrhoa{t}^{(l)},\vec{y}_t^{(l)})\in \mathcal{S}_{\rm tr}}
         \Bigl(
           f_{\t}\bigl(\vrhoa{\t}^{(l)};\bm{\eta}\bigr)
           -
           y_{\tau}^{(l)}
         \Bigr)^2 ,
   \end{equation}
where  $\vrhoa{\t}^{(l)}$ is a subseries of element $\vrhoa{t}^{(l)}$ of the training set $\SC_{\rm tr}$ . With that loss form, one obtains closed-form optimal weights $\bm{\eta}^*$ ~\cite{fujii2017harnessing}. Then, the trained model $f_{\boldsymbol{\eta}^*}(\cdot)$ can be used to predict or classify  unseen time-series data.
\end{enumerate}

The special case $t=1$, involving a single reservoir evolution and a single input datum, corresponds to the \emph{quantum extreme learning machine} (QELM) framework, of which Fourier expressivity and scalability barriers were analyzed in Ref.~\cite{xiong2025fundamental}. In contrast, the multi-step QRP protocol boosts expressivity by repeatedly interleaving fresh inputs into the reservoir dynamics, yet this recurrent coupling introduces scalability analysis challenges.

\subsection{Memory properties of reservoirs}\label{subsec:mem_prop}
Here we review relevant memory properties of reservoirs, which have been widely studied in the context of theory of dynamical systems~\cite{mujal2021opportunities, grigoryeva2018echo, fujii2017harnessing}. Two fundamental memory properties, that are necessary for reservoir computing to be operational (though by themselves insufficient to guarantee strong predictive performance) are:

\begin{enumerate}

    \item \emph{Fading Memory Property (FMP):} A reservoir has the FMP if the influence of inputs fed \emph{far in the past} on the present output \emph{fades away} in time. In other words, perturbations in the input $\rhoa{t-\tau}$ for large $\tau$ cause only a negligible difference in the current output $y_t$. FMP ensures that the reservoir focuses on the relevant recent history of the input and is not overwhelmed by far past information.

    \item \emph{Echo State Property (ESP):} A reservoir has the ESP if the reservoir state $\rhoh{t}$ at large $t$ depends only on the driving inputs $\{\rhoa{\tau}\}_{\tau \leq t}$ and \emph{not} on the reservoir's initial condition, i.e., the initial state of hidden qubits $\rhoh{0}$. The ESP thus guarantees that, after a long washout (burn-in) period, the reservoir state becomes a function of only input sequences, and of only recent inputs if the reservoir also has FMP, but \emph{independent} of how the reservoir was initially prepared. Without ESP, the output could spuriously depend on arbitrary past state preparations, making the system ill-defined for consistent time series processing.
\end{enumerate}
Both properties are necessary conditions for the QRP model to be operative~\cite{mujal2021opportunities,grigoryeva2018echo}. Moreover, by properly choosing a reservoir, such as a sufficiently large Echo State Network (ESN), certain QRP models with such properties can achieve a \emph{universal approximation} of fading memory maps and possess short-term sequence-modelling capabilities \cite{grigoryeva2018echo, gonon2021universal}.

\begin{table*}[t!]
\centering
\begin{tikzpicture}
\definecolor{zred}{HTML}{BD2327}
\definecolor{zgreen1}{HTML}{0F483E}
\draw[thick] (-8.0, -4.6) rectangle (7.0, -8.6);
\draw[thick] (-8.0, -5.4) -- (7.0, -5.4);
\draw[thick] (-8.0, -7.0) -- (7.0, -7.0);
\draw[thick] (-5.0, -4.6) -- (-5.0, -8.6);
\draw[thick] (1.0, -4.6) -- (1.0, -8.6);
\node[align=center,text width=4.0cm] at (-6.5,-4.98) 
    {Reservoir Dynamics};
\node[align=center,text width=4.0cm] at (-6.5,-6.2) 
    {Scrambling};
\node[align=center,text width=4.0cm] at (-6.5,-7.8) 
    {Local Noise};
\node[align=center,text width=4.0cm] at (-2.0,-4.98) 
    {Output};
\node[align=center,text width=4.0cm] at (4.0,-4.98) 
    {Memory};
\node[align=center,text width=6cm] at (-2.0,-5.8) 
    { \textbf{Theorem}~\ref{Thm:HaaResVar}};
\node[align=center,text width=6cm] at (-2.0,-6.35) 
    {\textcolor{zred}{$\bullet$ Exp. concentration in system size} };
\node[align=center,text width=6cm] at (-2.0,-6.7) 
    {\textcolor{zgreen1}{$\bullet$ No concentration in iterations} };
\node[align=center,text width=6cm] at (-2.0,-7.4) 
    { Refs.~\cite{thanasilp2022exponential,xiong2025fundamental,wang2020noise}};
\node[align=center,text width=6cm] at (-2.0,-8.1) 
    {\textcolor{zred}{$\bullet$ Exp. concentration in system size \\ (Pauli noise)} };
\node[align=center,text width=6cm] at (4.0,-5.8) 
    { \textbf{Theorem}~\ref{thm_main:mem_ind_upp_bound}};
\node[align=center,text width=6cm] at (4.0,-6.35) 
    {\textcolor{black}{$\bullet$ Exp. decay in system size} };
\node[align=center,text width=6cm] at (4.0,-6.7) 
    {\textcolor{black}{$\bullet$ Exp. decay in iterations} };
\node[align=center,text width=6cm] at (4.0,-7.4) 
    { \textbf{Theorem}~\ref{thm:unital_noise} / \textbf{Theorem}~\ref{thm:non_unital_noise_temp}};
\node[align=center,text width=6cm] at (4.0,-8.1) 
    {\textcolor{black}{$\bullet$ Exp. decay in iterations (and also in system size with non-unital noise)} };
\end{tikzpicture}
\caption{\textbf{Summary of \stadd{extreme scrambling results}.} The analytical results of the paper are summarized in this table. We analyse the scaling of the variance of output observables and indicators of the memory of the reservoir for two types of reservoir dynamics -- scrambling and noise. }
\label{tab:summary}
\end{table*}

\subsection{Exponential concentration}
The exponentially large Hilbert space, if not properly handled,  can hinder the scalability of QML models. In variational quantum models, this curse of dimensionality manifests itself as the Barren Plateau (BP) phenomena~\cite{mcclean2018barren, larocca2024review}, where loss landscapes become exponentially flat with system size, prohibiting training. In contrast, the curse of dimensionality appears differently in \emph{non-variational} QML approaches such as quantum kernel methods~\cite{thanasilp2022exponential,suzuki2024quantum}, and QELMs~\cite{xiong2025fundamental}. In these classically trainable, non-variational QML models, the exponential concentration results in data-insensitive models that generalize poorly. 

We now review the definition of exponential concentration

\begin{definition}[Probabilistic exponential concentration]\label{def_prob_exp_concentration}
Consider a quantity $Q(\vec{\alpha})$ that depends on some variable $\vec{\alpha}$ which can be estimated from an $n$-qubit quantum computer as an expectation value of some observable. $Q(\vec{\alpha})$ is said to probabilistically exponentially concentrate around an input-independent value $\mu$ if
\be\label{eq:def_exp_concentration}
\mathrm{Pr}_{\vec{\alpha}}[|Q(\vec{\alpha})-\mu|\geq\delta] \leq \dfrac{\beta}{\delta^2}\hspace{2pt},\;\; \beta \in \OC(1/b^{n}) \;,
\ee
for some $b > 1$. That is, the probability that $Q(\vec{\alpha})$ deviates from $\mu$ by a small amount $\delta$ is exponentially small for all $\thv$.
\end{definition}

The above definition is rather general and can be applied to many instances. In the context of QRP, we will consider $Q(\thv)=\expval{O}_{t}$, where $O$ can be any one of the observables out of the set used for training. In this case, the probability in Eq.~\eqref{eq:def_exp_concentration}, as well as variance and mean of $\expval{O}_{t}$, would be taken over the set of all inputs $\XC$ or over an ensemble from which the reservoir $U$ is drawn. We note that in Eq.~\eqref{eq:def_exp_concentration}, by applying Chebyschev's inequality, $\beta$ can readily be associated to  $\Var_{\vec{\alpha}}[Q(\vec{\alpha})]$ with $\mu=\Ebb_{\vec{\alpha}}[Q(\vec{\alpha})]$. 

\section{Results}

\subsection{Extreme scrambling reservoir dynamics}
Here we study the effect of using scramblers as the quantum reservoir.
A highly scrambling process at long times can be modelled by typical (pseudo) random unitaries~\cite{roberts2017chaos, tangpanitanon2023signatures, holmes2020barren}.  
This approach for modelling the scrambling reservoir has the added benefit of allowing us to adopt tools from Haar integration for analytical treatments~\cite{mele2023introduction}.
Concretely, for $t$ time steps, we suppose that the scrambling unitary is drawn from a distribution that agrees with the (Haar) random distribution of unitaries up to the 2t-th moment. That is, we consider reservoirs to form $2t$-design over the unitary group of $d_a d_h$ dimensions $\Ubb(d_a d_h)$ with $d_a = 2^{n_a}$ and $d_h = 2^{n_h}$. 

\subsubsection{Concentration in output's magnitude}\label{main:conc_output}
We start with the impact of extreme scrambling reservoirs on the scalability. The scrambling dynamics from the reservoir spreads out local information encoded from the input injection steps onto the exponentially large Hilbert space, making the information inaccessible in the large system size limit. Consequently, this effect manifests itself as the exponential concentration of the QRP output with respect to the reservoir size and is formalized as the vanishing variance in the following theorem.

\begin{theorem}[Variance of QRP output with extreme scrambling reservoir]\label{Thm:HaaResVar} Given that a reservoir is sampled from a $2t$-design on $\Ubb(d_a d_h)$, the variance of a QRP's outputs vanish exponentially with the system size. Formally, for large $t$, the variance is given by
\begin{equation}
\Var_{U\sim\Ubb(d_a d_h)}[\expval{O}_t]=\frac{\Tr[O^2]\left(1+ \beta\right)}{ d_a^2 d_h^3}\;,
\end{equation}
with $\beta \in \OC(\frac{1}{d_a d_h})$.
\end{theorem}
By using Chebyshev's inequality, we straightforwardly obtain the probabilistic exponential concentration as in Definition~\ref{def_prob_exp_concentration} with the average $\mathbb{E}_{U \sim\Ubb(d_a d_h) }\left[\expval{O}_t\right] = \frac{\Tr[O](1+\beta)}{d_a d_h}$ as the concentration value, which is also independent of either reservoir or input states. A formal version of Theorem~\ref{Thm:HaaResVar} and its proof is presented in Appendix~\ref{app:conc_output}.

Proving Theorem~\ref{Thm:HaaResVar} is non-trivial and involves going beyond the standard proof technique commonly used to show exponential concentration. In particular, while the same instance of the quantum reservoir is employed for every time step, the entire composite system becomes larger with every time step due to the  input state injection process. This prompts the question of whether the randomness generated by the same scrambling instance is sufficient to cause concentration. The proof of Theorem~\ref{Thm:HaaResVar} relies on expressing the whole QRP dynamics as a tensor diagram. One then carefully works with the recursive expression of expectation values in Eq.~\eqref{obs_def_QRP} by writing in the effectively extended space such that the first and second moments of the expectation transformed to a standard Haar integral with $2t$-design. This alternative form of the QRP output is presented in Supplemental Proposition~\ref{sup-prop:qrp-output}.

We remark that our assumption that the reservoir is drawn from a unitary $2t$-design is stronger than typically required for similar concentration bounds in QELM or quantum kernel methods~\cite{xiong2025fundamental,thanasilp2022exponential}, where a $2$-design ensemble is sufficient. Whether this is also a necessary condition, or just a relic of our proof techniques, remains an open question.

\begin{figure}[ht]
    \centering
    \includegraphics[width=0.48\textwidth]{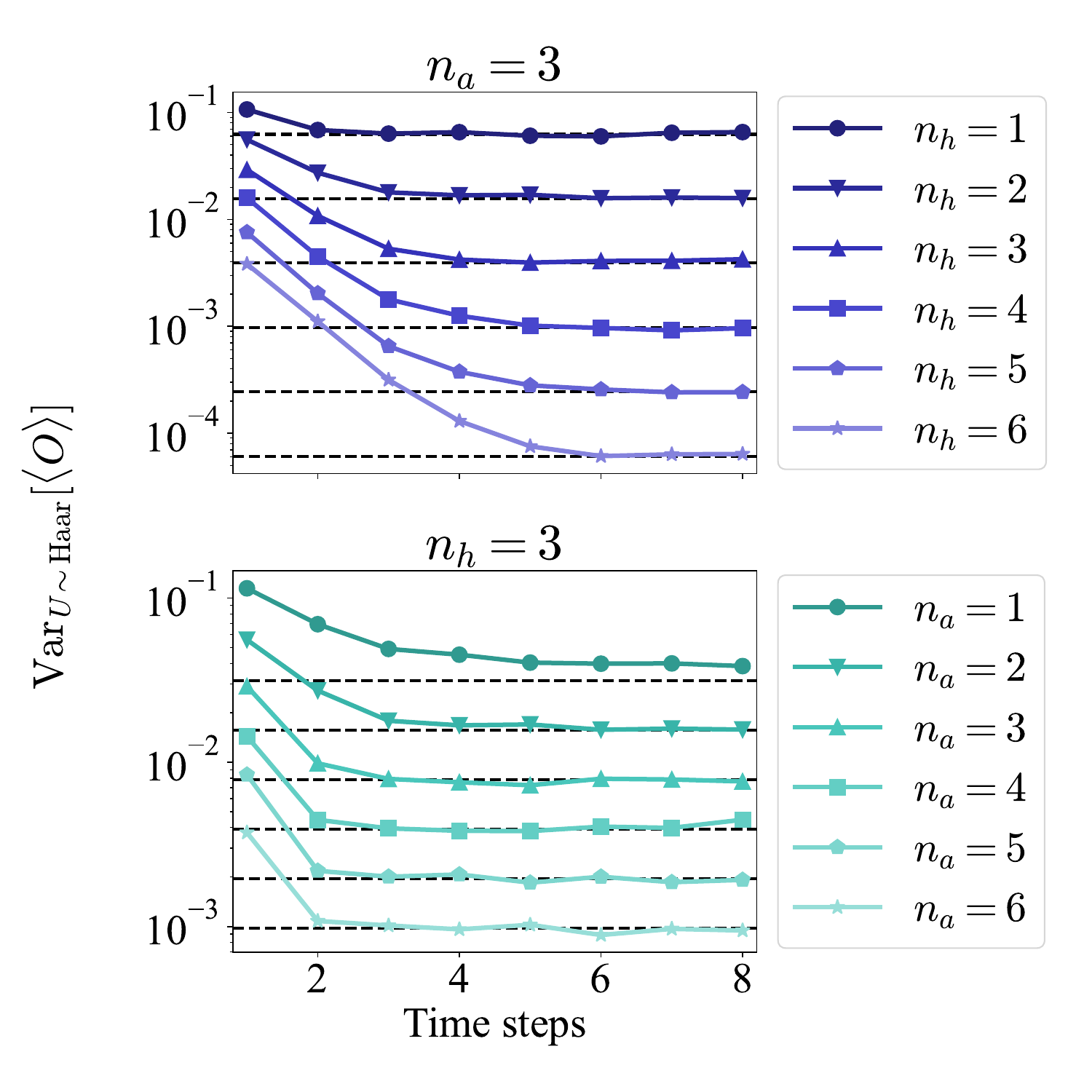}
    
    \caption{\textbf{Concentration induced by reservoir dynamics.} The variance of reservoir output over Haar random reservoirs are plotted against number of time steps for different numbers of hidden qubits (top) and accessible qubits (bottom). The observables are Pauli strings. The dashed horizontal lines represent the leading terms $\frac{1}{d_a d_h^2}$ obtained analytically. Note that the classical input are sequences of independent random numbers uniformly drawn from $[0,1]$ and the exponential encoding strategy is employed to encode the data. 
    }
    \label{fig:conc_U}
\end{figure}

A numerical simulation of the variance scaling with the scrambling reservoir is illustrated in Fig.~\ref{fig:conc_U}, as a function of $n_a$ and $n_h$. While the dominant term of the variance (for large $d$) in Theorem~\ref{Thm:HaaResVar} does not depend on the time steps, one indeed observes some discrepancy at the early time steps. %
This temporally vanishing contribution of variance is dominated by the term $\frac{\Tr[O^2]}{d_a^{t+1}d_h^2}$ as presented in the formal version of the theorem and proven/further discussed in Appendix~\ref{app:var_small_t}. %
However, after some reservoir iterations, the variance quickly saturates to the predicted value which depends only on the number of qubits. The convergence rate is empirically observed to depend on the ratio $n_h/n_a$, which reflects the interplay of the injection of new data and the accumulated scrambling due to reservoir dynamics.

\paragraph*{Practical consequences.} Exponential concentration does not prevent the training of a QRP model, as the optimisation over readout weights is a classical convex optimization problem. However, due to the intrinsically probabilistic nature of quantum hardware, the observables read-out can only be estimated with measurement shots. For large system sizes, one typically can only afford a polynomial number of measurements to estimate these observables. As a consequence of the concentration phenomena, these estimates become statistically indistinguishable from input-independent random variables. Training on these estimates therefore yields an \emph{input-insensitive}
QRP with poor generalisation. 
Thus, while trainability is guaranteed, the scalability of generalisation is not.

The following proposition formalizes this notion for QRP output.
\begin{proposition}\label{prop:prac-cons}
    Consider the same assumptions as Theorem~\ref{Thm:HaaResVar} and further suppose that $O$ is a Pauli string. After polynomial measurement shots $N \in \OC(\poly(n))$, the estimate of $\langle O\rangle_t$  is statistically indistinguishable, with probability exponentially close to $1$,
    from an input-independent random variable of the form
    \begin{align}\label{eq:rand-output-main}
        \widehat{O}_{\text{rand}} = \dfrac{1}{N}\sum_{i=1}^N z_i \;,
    \end{align}
    where each $z_i$ takes a value `$+1$' and `$-1$' with an equal probability.
\end{proposition}
\noindent 
The formal version of this statement is presented in Supplemental Theorem~\ref{sup-thm:stat-indis-qrp-output} using hypothesis testing tools. In addition, Supplemental Corollary~\ref{sup-coro:qrp-performance} also formally presents an explicit connection to trainability and unseen model predictions. 
We refer the readers to Appendix~\ref{app:concentration-consequence-summary} for further technical discussions and the proof.

While the above implications arise in the regime of large system size, for \emph{small}
reservoirs one can afford sufficient shots to obtain accurate estimates of observables.
Together with the fact that the prediction variance in Theorem~\ref{Thm:HaaResVar} does not further decrease with time, this leaves a
 window of opportunity in which \emph{an extreme-scrambling QRP on modest system sizes might
offer a polynomial advantage over standard tomography approaches when
processing quantum data.}

\stadd{\paragraph*{QRP variants with mid-circuit measurements and feed-forward.} We further analyze a variant of QRP with mid-circuit measurements and feed-forward strategies, which have been proposed and analysed in Ref.~\cite{mujal2023time,franceschetto2024harnessing,kobayashi2024feedback,murauer2025feedback,monomi2025feedback,morguisancho2026general}. We show in Appendix~\ref{app:adaptive_mid_meas_feedforward} similar analytical concentration results in the extremely scrambling reservoir limit, assuming each mid-circuit outcomes can be efficiently sampled with polynomially small probability. This indicates that extremely scrambling reservoirs remain unfavourable regardless of how they are processed. However, such strategies does not introduce additional concentration.}

\paragraph*{Other concentration sources.} With the detrimental implications of concentration being unfolded, it is crucial to identify other potential factors. While the concentration due to hardware noise will be discussed in Section~\ref{sec:noise}, other sources of exponential concentration of QRP --such as entanglement and global measurement-- can be shown directly using the same proof strategies as for QELM~\cite{xiong2025fundamental} - we can simply view the initial state of QELM as the state of QRP $\big(\rhoa{t}\otimes \rhoh{t-1}\big)$ before the reservoir evolution of the last step $t$.

\begin{figure}[ht]
    \centering
    \includegraphics[width=0.48\textwidth]{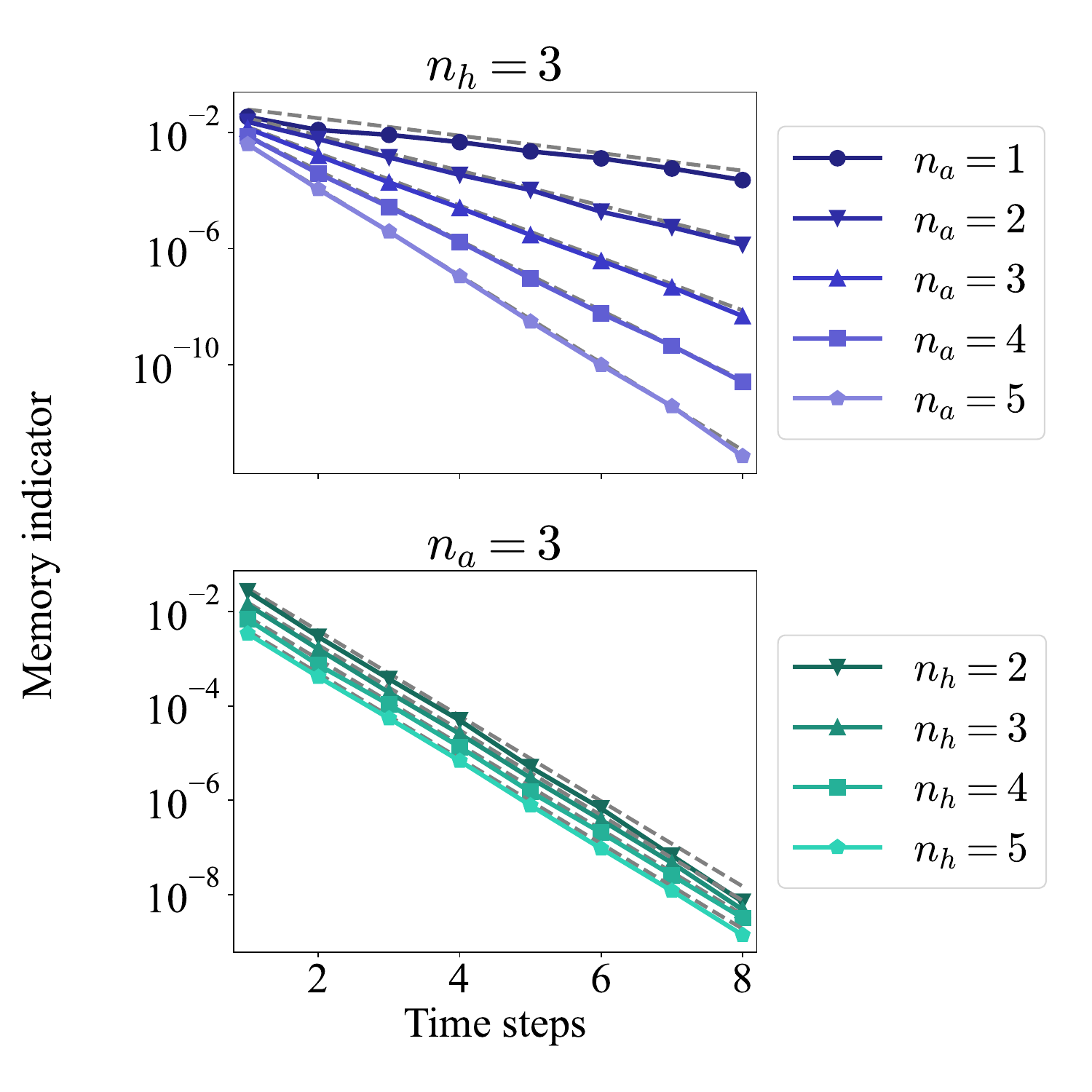}
    
    \caption{\textbf{Memory indicator of the early input.} The averaged dependence of outputs on inputs of $t$ steps early $\Ebb_{U\sim\Ubb_{d_a d_h}}[\MC_a(t;\SC_a,U)]$ are plotted for different qubit numbers. Here both the reservoir and the first input state are Haar random. The grey lines indicate the bounds given in Theorem~\ref{thm:mem_ind_upp_bound}.
    }
    \label{fig:inp_dep}
\end{figure}

\subsubsection{Exponential decay of memory indicators}\label{sec:exp_mem_indicators}
A viable QRP must not only remember useful information about past inputs, but also should eventually forget older inputs and older reservoir states, so that the current outputs are not dominated by the far past history. In classical reservoir computing (RC), these two requirements are commonly encapsulated by the \emph{Fading Memory Property} (FMP) and the \emph{Echo State Property} (ESP)~\cite{grigoryeva2018echo,mujal2021opportunities}. As reviewed in Section~\ref{subsec:mem_prop}, FMP requires that the effect of inputs fed far in the past fades away, while ESP demands that the reservoir’s internal state eventually becomes independent of its initial state. 

To capture these forgetting characteristics of QRP, we define two \emph{memory indicators} that quantify how much the output depends on an early state: one for the \emph{accessible qubits} (input channel) and one for the \emph{hidden qubits} (reservoir’s internal state). %

\begin{definition}[Memory of input $\t$] Let us denote the distribution of input states at step $\t$, which we assume is independent of inputs at other time steps, by $\rhoa{\t}\sim \SC(d_a)$. Given a reservoir unitary $U$, we define the memory of the input at step $\tau$ after $t$ time steps have passed as\label{def:input_dep}
\be
    \MC_a(t;\SC_a,U):=\Var_{\rhoa{\t}\sim \SC_a}\left[\expval{O}_{(t+\t)}\right]\,.
\ee
\end{definition}

\begin{definition}[Memory of initial condition] Given the distribution of initial state of the hidden qubits $\rhoh{0}\sim \SC(d_h)$, we define the memory of the initial state at step $t$ for a given reservoir $U$ as\label{def:init_state_dep}
\be
    \MC_h(t;\SC_h,U):=\Var_{\rhoh{0}\sim \SC_h}\left[\expval{O}_{t}\right]\,.
\ee
\end{definition}

{These memory indicators provide operational, quantitative measures of memory retention in QRP. For a fixed reservoir and readout observable, a large $\MC_a(t)$ means that the readout remains sensitive to the input supplied $t$ steps earlier, whereas a large $\MC_h(t)$ means that it remains sensitive to the initial hidden state. Accordingly, we call decay of $\MC_a(t)$ with temporal delay \emph{FMP-like}, because it represents loss of readout sensitivity to increasingly distant inputs, and decay of $\MC_h(t)$ \emph{ESP-like}, because it represents washout of the readout's dependence on the initial hidden state. The qualifier ``-like'' is essential because the indicators concern a chosen readout and prescribed input or initial-state \emph{ensembles}, and the bounds below additionally average them over the reservoir ensemble. These indicators therefore do not certify formal FMP or ESP, but capture their essential operational notion of forgetting in a form amenable to analytical bounds. Specifically, the bounds below establish \emph{ensemble-averaged}, \emph{observable-level} forgetting and resolve both its temporal decay and the reservoir-size dependence of the surviving, experimentally accessible input signal. We elaborate on this distinction in Appendix~\ref{appendix:memory}.}

\smallskip

We now state the contrentation bounds of memory indicators, with the formal propositions and proofs presented in Appendix~\ref{app:mem_conc}.
\begin{theorem}[Upper bound of memory of arbitrary early input and initial state ensemble (informal)]\label{thm_main:mem_ind_upp_bound}
    For extreme scrambling reservoirs and an arbitrary ensemble of inputs $\SC_a$, the memory indicators defined in Def.~\ref{def:input_dep} and \ref{def:init_state_dep} exponentially vanish in system size $\OC\left(\frac{1}{d_h d_a^t}\right)$, with high probability exponentially close to $1$.
\end{theorem}

{Theorem~\ref{thm_main:mem_ind_upp_bound} bounds the ensemble-averaged indicators along two directions with different interpretations. By applying Chebychev tail bound, the probabilistic concentration statement (see Corollary~\ref{cor:var_mem} for the formal version) further implies that this exponential decay happens with probability close to $1$. At fixed reservoir size, their decay with temporal distance describes the potentially useful forgetting of far-past inputs through $\MC_a(t)$ and washout of the initial hidden state through $\MC_h(t)$. In the extreme-scrambling regime, however, the indicators are also exponentially suppressed with reservoir size. In particular, the bound on $\Ebb_U[\MC_a(t)]$ can become exponentially small even at a fixed, task-relevant delay, so the remaining input-dependent readout signal may become unresolvable with polynomially many measurement shots as the reservoir is enlarged. It is this system-size-dependent loss of a task-relevant signal, rather than fading memory itself, that constitutes the scalability barrier.}

{The exponential temporal memory decay established above applies to the original QRP architecture, in which past information remains encoded in the quantum reservoir. Recent measurement-based and feed-forward architectures require a separate treatment~\cite{mujal2023time,franceschetto2024harnessing,kobayashi2024feedback,murauer2025feedback,monomi2025feedback,morguisancho2026general}. We discuss in Appendix~\ref{app:adaptive_mid_meas_feedforward} that information transferred to a persistent classical record need not decay under subsequent reservoir iterations, whereas the final quantum readout of the adaptive protocol remains subject to system-size-dependent concentration under the stated efficient-sampling assumptions.}

\subsection{Noisy reservoir dynamics}\label{sec:noise}
Here we study the impact of hardware noise in two different settings: (i) unital noise where local Pauli noise acts between different noiseless reservoir dynamics, and (ii) noisy scrambling reservoir where the reservoir itself consists of local scrambling unitaries intervened with arbitrary local non-unitary noise.

\subsubsection{Unital noise}
To study how the noise between each time steps of noiseless reservoir evolutions affects the  memory, we consider the state after reservoir evolution as
\begin{equation}
    \rhot{\t+1}=\NC\left(U\left(\rhoa{\t+1}\otimes\rhoh{\t}\right)U^\dag\right)\,,\label{Eq:noisy_res_evo}
\end{equation}
where $\NC(\cdot)$ is a noise channel which implements a tensor product of local Pauli noise on each individual qubits i.e., $\NC = \NC_1 \otimes...\otimes \NC_{n_a+n_h}$ with 
\begin{align}\label{Eq:NOise_factor}
    \NC_j(\sg)=q_\sg \sg \;,
\end{align}
acting on the $j^{\rm th}$ qubit. Here, $\sg \in \{X, Y,Z\}$ is one of the non-trivial single-qubit Pauli operators and $q_\sg$ is an associated noise coefficient. We further define the characteristic noise parameter $q := \max \{|q_X|, |q_Y|, |q_Z|\}$. 

With the above setting, the following theorem holds.
\begin{figure}[ht]
    \centering
    \includegraphics[width=0.45\textwidth]{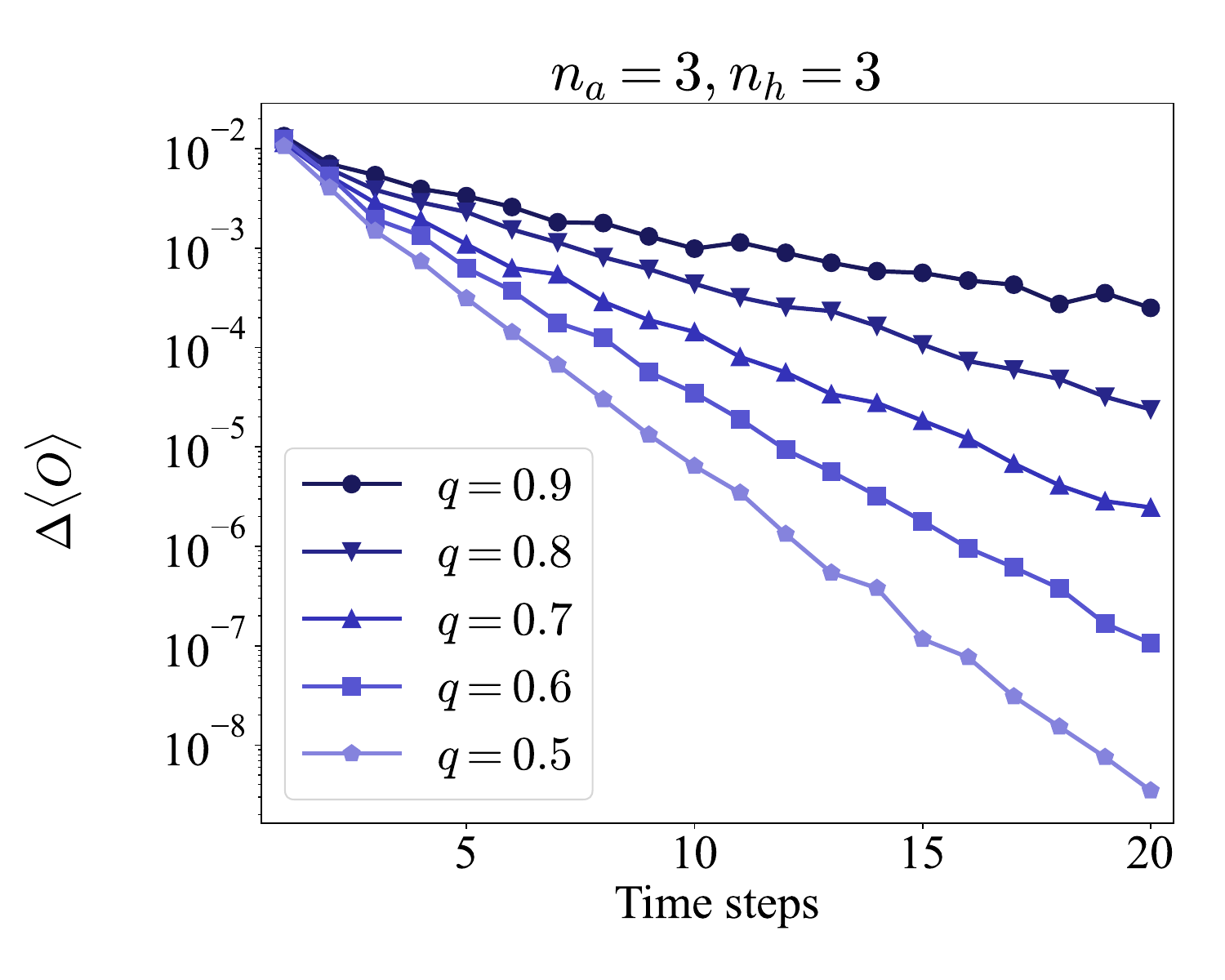}
    
    \caption{\textbf{Noisy reservoir induced concentration.} The averaged deviation of two outputs of QRP with randomly sampled pairs of initial states and first inputs are plotted for different noise factors. 
    }
    \label{fig:conc_noisy_reservoir}
\end{figure}

\begin{theorem}[{Memory erasure induced by unital noise}]\label{thm:unital_noise}
Consider a QRP with noisy layers between noiseless reservoir evolutions as defined in Eq.~\eqref{Eq:NOise_factor}, with $q<1$. Then, for time step $t$, the expectation value of the observable will become independent of the initial condition and early inputs, exponentially in number of steps $t$. Formally, given two initial states of hidden qubits $\rhoh{0}$ and $\sgh{0}$, given two inputs $\rhoa{1}$ and $\sga{1}$ at the first time step, given the same series of input states $[\rhoa{\tau}]_{\tau=2}^t=[\sga{\tau}]_{\tau=2}^t$ from the second step, then the deviation of the corresponding outputs
    \begin{align}
       \Delta O\leq & \norm{O}_\infty \sqrt{2\ln{2}q^{b(t-1)}} \nonumber \\
        &\qquad\cdot\bigg(S_2\left(\rhoa{1}\lVert\sga{1}\right)+S_2\left(\rhoh{0}\lVert\sgh{0}\right)\bigg)^{1/2}\\
        \in&\OC(\rm{exp(-t)})\,,
    \end{align}
    where $\Delta O=|\expval{O}^{(\rhoh{0},\rhoa{1})}_t - \expval{O}^{(\sgh{0},\sga{1})}_t|$, $b= 1/\ln{2}$ and $S_2(\cdot\lVert\cdot)$ denotes the sandwiched 2-Rényi relative entropy.
\end{theorem}
This upper bound indicates that the difference between outputs from different initial states shrinks exponentially with the number of time steps. Thus at long times the QRP is effectively insensitive to the initial state. Numerics in Fig.~\ref{fig:conc_noisy_reservoir} demonstrates these effects. The proof of this Theorem is provided in Appendix~\ref{sec:conc_noisy_res_proof}.

On another note, when input data are classical, noise can act also during the input-data encoding stage. We proved in Appendix~\ref{sec:conc_noisy_enc} that noisy encoding can solely drive exponential concentration in QRP's outputs, extending the noisy-encoding bounds previously obtained for QELMs~\cite{xiong2025fundamental}.

\subsubsection{Non-unital noise}
In this section, we investigate how \emph{arbitrary} local noise -- possibly significantly deviating from the Pauli noise model considered above -- {erase memory of early inputs and reservoir's initial state}. In particular, we extend our analysis to include non-unital noise, which exhibits markedly different behaviors. Unlike unital noise, non-unital noise does not leave the maximally mixed state invariant. Instead, it can reduce the entropy of a quantum system, thereby providing a potential resource for computation. This property has been exploited in various contexts, including fault-tolerant quantum computing~\cite{ben2013quantum} and QRP~\cite{sannia2024dissipation}

In the case of non-unital noise it is not possible to derive bounds for generic reservoirs because with careful engineering (essentially sneaking in error correction) arbitrary circuits can be realized despite non-unital noise~\cite{ben2013quantum}. Instead, we here focus on \emph{typical} noisy reservoirs. Specifically, following Refs.~\cite{mele2024noise, angrisani2025simulating}, we model the reservoir $\mathcal{U}$ as sequence of $L$ random unitary layers interleaved by arbitrary local noise $\mathcal{N}^{\otimes n}$:
\begin{align}\label{Eq:non-unital-noisy reservoir}
    \mathcal{U} =  \mathcal{U}_L \circ \mathcal{N}^{\otimes n} \circ \mathcal{U}_{L-1}\circ \dots \circ \mathcal{N}^{\otimes n}  \circ \mathcal{U}_1 \circ \mathcal{N}^{\otimes n} .
\end{align}
We assume that each $\mathcal{U}_j$ is sampled from a ``locally scrambling'' distribution~\cite{caro2022outofdistribution, huang2023learning, angrisani2024classically}, e.g. $\mathcal{U}_j$ can be Brickwork layer with local Haar-random gates, and we let the depth $L$ be at least linear in system size.

One intuitively expects noisy QRP, even with non-unital noise, to be insensitive to its early states due to the noise destroying their information. The following theorem formalizes this notion.
\begin{theorem}[{Memory erasure induced by non-unital noise}]\label{thm:non_unital_noise_temp}
Given a noisy reservoir $\mathcal{U}$ as in Eq.~\eqref{Eq:non-unital-noisy reservoir}, let $\mathcal{N}$ be an arbitrary non-unitary channel with constant noise rate and let $L\geq c \cdot (n_a + n_h)$ for a sufficiently large constant $c>0$.
Given two initial states of hidden qubits $\rhoh{0}$ and $\sgh{0}$, given two inputs $\rhoa{1}$ and $\sga{1}$ at the first time step, given the same series of input states $[\rhoa{\tau}]_{\tau=2}^t=[\sga{\tau}]_{\tau=2}^t$ from the second step, the corresponding output states satisfies
\begin{align}
    \norm{\rhot{t} - \sgt{t}}_1 \in \exp\left(-\Omega((n_a+n_h)t)\right),
\end{align}
with probability at least $1 - \exp(-\Omega(n_a+n_h))$.
In particular, for any observable $O$ it holds that
\begin{align}
    \Delta O \in \exp\left(-\Omega((n_a+n_h)t)\right) \norm{O}_\infty,
\end{align}
with probability at least $1 - \exp(-\Omega(n_a+n_h))$, where we define $\Delta O=|\expval{O}^{(\rhoh{0},\rhoa{1})}_t - \expval{O}^{(\sgh{0},\sga{1})}_t|$.
\end{theorem}
Our bounds imply that for reservoirs with arbitrary local noise, possibly non-unital, the dependence of reservoir states on early input states and the initial state of hidden qubits vanishes exponentially in number of time steps. {This gives FMP- and ESP-like forgetting under the stated assumptions (see Appendix~\ref{appendix:memory}). This interpretation is compatible with dissipation being useful: engineered dissipation can establish an appropriate washout timescale by suppressing unwanted dependence on far-past inputs and initial conditions~\cite{sannia2024dissipation, chen2020temporal}. Our bounds do not imply that non-unitarity is itself an obstacle. Rather, under the specific noise models considered here, they identify regimes in which contraction can be too rapid. %
Contraction becomes detrimental when its washout timescale is much shorter than the memory timescale required by the task. Note also that while} previous results on concentration (and lack thereof) for non-unital dynamics rely on i.i.d.\ assumptions, our setting involves repeated application of the same noisy reservoir. To establish our result, we extend the analytical framework developed in Refs.\ \cite{mele2024noise, angrisani2025simulating}.

{Finally, we remark that the assumptions required for the unital and non-unital noise results are essentially
different. For the unital noise considered here, memory contraction can be established without imposing randomness on the reservoir dynamics. Such a general statement is not possible for non-unital noise: suitably engineered circuits can exploit non-unitality as a cooling resource to implement effective RESET operations and sustain arbitrarily deep fault-tolerant computation~\cite{ben2013quantum,shtanko2025complexity}. This motivates our restriction to typical locally scrambling reservoirs in the non-unital case. For such random dynamics, contraction remains possible, consistent with average-case results for noisy random circuits
~\cite{mele2024noise,angrisani2025simulating,shtanko2025complexity}.}

{
\subsection{Physical scrambling reservoir dynamics} %
The extreme scrambling limit analyzed above, while analytically tractable, is an idealization. Reservoirs realized in practice are generated by structured quantum gates or by the finite-time dynamics of physical Hamiltonians, and need not form a $2t$-design. A natural question, therefore, is whether the exponential concentration phenomena can emerge in other physical reservoirs that do not fully reach the extreme scrambling limit. In this section we address this question numerically, studying output concentration for the main reservoir families that have been proposed or employed in the QRP literature. 
We find that exponential concentration is not an artifact of the extreme scrambling assumption and persists across physical reservoirs whenever the dynamics are sufficiently scrambling, with a concentration rate controlled by the underlying dynamical phase. 
{
\begingroup
\scriptsize
\setlength{\tabcolsep}{4pt}
\renewcommand{\arraystretch}{1.15}

\begin{longtable}{
p{0.40\linewidth}|
p{0.18\linewidth}
>{\centering\arraybackslash}p{0.07\linewidth}
p{0.25\linewidth}
}
\caption{Summary of physical reservoir concentration results.}
\label{tab:QRP_reservoirs_main}
\\
\hline
\textbf{\centering Reservoirs}
&
\multicolumn{3}{c}{\textbf{Exp.\ Conc.}}
\\
\hline
\endfirsthead

Haar (global)~\cite{kobayashi2024feedback,murauer2025feedback}
&
&
$\checkmark$
&
Sec.~\ref{main:conc_output}
\\[0.35em]
\hline

\multirow{2}{=}{Local random~
\cite{chen2020temporal,suzuki2022natural,kubota2022quantum,domingo2023taking,fry2023optimizing,pfeffer2022hybrid,pfeffer2023reduced,wudarski2023hybrid,monzani2024non,ricci2026quantum,yasuda2023quantum,hu2024overcoming,malik2026evaluating}}
&
Deep
&
$\checkmark$
&
\multirow{2}{*}{Sec.~\ref{sec:digi_rese_conc}}
\\
&
Shallow
&
$\times$
&
\\[0.35em]
\hline

\multirow{3}{=}{Ising models 
\cite{fujii2017harnessing,negoro2018machine,nakajima2019boosting,kutvonen2020optimizing,martinez2021dynamical,li2026quantum,hou2026high,ivaki2024quantum,settino2025topology,mccaul2025minimal,cindrak2024enhancing,delorenzis2025harnessing,fuchs2024quantum,xia2023configured,hayashi2025effective,delorenzis2025harnessing,gili2024generating, dao2026breaking,zhang2026robust}}
&
Ergodic
&
$\checkmark$
&
\multirow{2}{*}{Sec.~\ref{sec:conc_ising}}
\\
&
Integrable
&
$\times$
&
\\
&
Kicked
&
$\checkmark$
&
App.~\ref{app:kicked_ising}
\\[0.35em]
\hline

\multirow{2}{=}{SYK fermions~
\cite{kobayashi2026edge}}
&
Ergodic
&
$\checkmark$
&
\multirow{2}{*}{Sec.~\ref{app:syk}}
\\
&
Integrable
&
$\times$
\\[0.35em]
\hline

\end{longtable}
\endgroup
}

We categorize physical reservoirs considered in this work into two classes
\begin{enumerate}
    \item \emph{Brickwork-circuit reservoirs}~\cite{kobayashi2024feedback,murauer2025feedback,malik2026evaluating,chen2020temporal,suzuki2022natural,kubota2022quantum,domingo2023taking,fry2023optimizing,pfeffer2022hybrid,pfeffer2023reduced,wudarski2023hybrid,monzani2024non,ricci2026quantum,yasuda2023quantum,hu2024overcoming}, composed of logical quantum gates such as local (random) rotations interleaved with entangling operations.
    \item \emph{Dynamical reservoirs}, generated by the continuous dynamics of quantum many-body systems. Here we focus on Hamiltonians of interacting two-level systems and the quantum dynamics is usually generated in an analogue manner, encompassing interacting spins, Rydberg-atom arrays, and fermionic models ~\cite{fujii2017harnessing,negoro2018machine,nakajima2019boosting,kutvonen2020optimizing,martinez2021dynamical,li2026quantum,hou2026high,ivaki2024quantum,settino2025topology,mccaul2025minimal,cindrak2024enhancing,delorenzis2025harnessing,fuchs2024quantum,xia2023configured,hayashi2025effective,gili2024generating,bravo2022quantum,kornjaca2024large,kobayashi2026edge}. Although less studies in QRP literature, the dynamics can be constructed through digital quantum gates with Trotterization~\cite{domingo2022optimal,mlika2023user,fuchs2024quantum,hu2024overcoming}.
\end{enumerate}

For brickwork-circuit reservoirs, the picture closely parallels the design-based analysis of the previous sections. As the circuit depth grows, local random circuits converge to approximate unitary designs~\cite{harrow2009random,brandao2016local,haferkamp2022random,haah2024efficient}, and the output concentration correspondingly approaches the extreme-scrambling limit. 

The dynamical reservoir case is much richer. We find that the strength of concentration is closely tied to the \emph{dynamical phase} in which the reservoir Hamiltonian operates. Reservoirs deep in an ergodic (quantum-chaotic) phase exhibit strong exponential concentration when evolved sufficiently long, whereas localized or near-integrable reservoirs concentrate significantly more mildly. This correspondence also identifies concrete criteria for mitigating exponential concentration by operating the reservoir away from the deeply ergodic regime, or restricting the evolution time per iteration. 
While our main case study is the disordered transverse-field Ising model, we observe qualitatively consistent behavior for Sachdev--Ye--Kitaev (SYK) and kicked-Ising reservoirs (see the summary at the end of this section and Appendices~\ref{app:syk}--\ref{app:kicked_ising}). 

{Importantly}, we expect these concentration results to be governed by the underlying phases rather than the specific physical system. Consequently, they should extend naturally to other physical reservoirs, such as Rydberg atomic, bosonic and hybrid systems~\cite{sakurai2025quantum,nerenberg2025photon,cimini2026large,garcia2023scalable,nokkala2021gaussian,dudas2023quantum,angelatos2021reservoir,khan2021physical,govia2021quantum,cheamsawat2025dissipation,llodra2025quantum,senanian2024microwave,burgess2022quantum,das2026quantum,zhu2025practical,zhu2025minimalistic, bravo2022quantum,kornjaca2024large,solanki2025harnessing,liu2026practical}. { As a key example, Ref.~\cite{alves2026dual} studies a dual-unitary circuits where local XXZ Hamiltonians are implemented in a Trotterized fashion. Their work also reports a concentration phenomenon that depends on the underlying phases of matter, albeit from a complementary perspective in which the concentration is characterized over the input space (instead of reservoir instances).
}

Table~\ref{tab:QRP_reservoirs_main} summarizes the main reservoir families together with pointers to the corresponding concentration results; a comprehensive survey of reservoir dynamics used in existing works is provided in Table~\ref{tab:QRP_reservoirs_detail} of Appendix~\ref{app:summary_reservoir}.
}

\begin{figure*}[t!]
    \centering
    \includegraphics[width=0.8\textwidth]{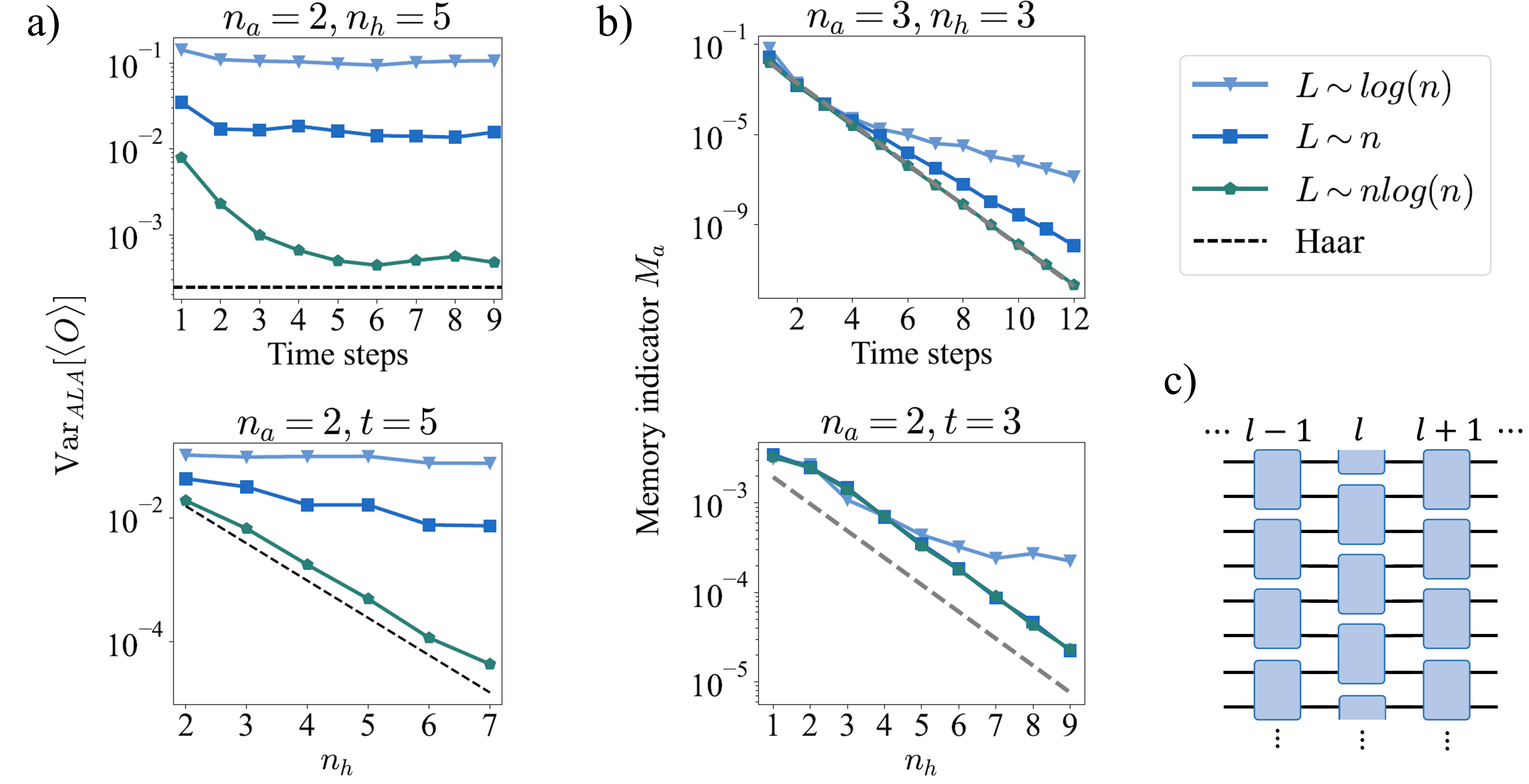}
    \caption{{\textbf{Depth-dependent exponential concentration in QRP with alternating layered reservoirs.} We consider reservoirs consisting of $L$ alternating layers, as illustrated in panel \textbf{c)}: each layer is a tensor product of $2$-local blocks acting on neighbouring qubits, with the pairing shifted in consecutive layers, and each block is independently sampled from the Haar measure. The variance of the output over the reservoir ensemble [panel \textbf{a)}] and the memory indicator $\MC_a$ [panel \textbf{b)}] are plotted for different depth scalings of $L$ and compared with a fully Haar-random reservoir (dashed lines). In the top subfigures, the qubit numbers are fixed and both quantities are plotted against the number of time steps; in the bottom subfigures, the number of accessible qubits and the time step are fixed and both quantities are plotted against the number of hidden qubits.
    }}
\label{fig:ala_all}
\end{figure*}
{
\subsubsection{Brickwork-circuit reservoirs.} \label{sec:digi_rese_conc}
Brickwork-circuit reservoirs are mostly built from layers of randomly sampled or randomly parameterized local gates interleaved with entangling operations, and their degree of scrambling grows with the circuit depth. This intuition can be made quantitative, as local random circuits are known to converge to approximate unitary $t$-designs in a depth that scales (at most) linearly with the number of qubits~\cite{harrow2009random,brandao2016local,haferkamp2022random,haah2024efficient}. Sufficiently deep Brickwork-circuit reservoirs therefore realize the extreme-scrambling ensembles of Section~\ref{main:conc_output}, and their outputs inherit the exponential concentration of Theorem~\ref{Thm:HaaResVar}. Conversely, restricting the depth limits the accumulated randomness, in close analogy with expressibility-dependent concentration in variational circuits~\cite{sim2019expressibility,cerezo2020cost}. The circuit depth thus serves as a natural knob interpolating between weakly scrambling reservoirs and the extreme-scrambling limit.

\paragraph*{Local random brickwork circuits.} To demonstrate this explicitly, we consider reservoirs given by brickwork circuits of $L$ alternating layers of $2$-local blocks, each independently drawn from the Haar measure, as illustrated in panel~\textbf{c)} of Fig.~\ref{fig:ala_all}. In each layer every qubit is coupled to a neighbouring qubit, with the pairing shifted in consecutive layers. For variational circuits with this alternating layered structure, it is known that exponential concentration is avoided when the number of layers satisfies $L\in\OC(\log(n))$~\cite{cerezo2020cost}.

Our numerics, presented in Fig.~\ref{fig:ala_all}, confirm the analogous behavior for QRP. For shallow reservoirs with $L\in\OC(\log(n))$, we observe the absence of exponential concentration in panel~\textbf{a)}. As the depth increases towards $L\in\OC(n\log(n))$, the variance converges to the extreme-scrambling prediction of Theorem~\ref{Thm:HaaResVar}, saturating the Haar value after a few reservoir iterations. Panel~\textbf{b)} shows the corresponding memory indicator $\MC_a$. While it decays exponentially with the number of time steps for all depths considered. This reflects the FMP-like forgetting behavior discussed in Section~\ref{sec:exp_mem_indicators}. In particular, its suppression with system size again interpolates with depth, remaining mild for logarithmic-depth reservoirs and matching the Haar-random decay once the depth becomes linear.

{{ \paragraph*{Is the $2t$-design assumption necessary?}
The $2t$-design assumption used in Theorem~\ref{Thm:HaaResVar} arises from our proof technique.
It does not follow, however, that reproducing the full Haar moment operator up to this
order is necessary for exponential concentration. %
This is illustrated by the
brickwork-circuit results in Fig.~\ref{fig:ala_all}. As the circuit depth is
increased, the output variance progressively approaches the
extreme-scrambling prediction, rather than displaying a sharp transition only
when the circuit becomes a certified $2t$-design.

This gradual emergence of exponential concentration is consistent with the broader picture of pseudorandomness in quantum dynamics, where increasingly complex dynamics progressively reproduce higher-order statistical properties of Haar-random evolution.
For the local brickwork random circuits used in this study, an approximate $2t$-design is guaranteed at depth
\begin{equation}
    L\in\OC\left(n(2t)^{5+o(1)}\right)\;\;,
\end{equation}
up to the dependence on the approximation error~\cite{harrow2009random, brandao2016local, haferkamp2022random}. For any finite fixed $t$, we can see that a depth linear in $n$ suffices\footnote{Note that if each brick is allowed to cover more than 2 qubits, one could construct an approximate $k$-design in much lower depth~\cite{schuster2024random}}. On the other hand, dynamical evolution under certain specially constructed Hamiltonians can also generate unitary designs. In particular, despite being non-physical, Hamiltonians directly drawn from random matrix ensembles can generate Haar-random unitaries~\cite{roberts2017chaos}, while Hamiltonians with logarithmic locality can generate approximate $2t$-designs in constant evolution time~\cite{cui2025random}.}

}

\subsubsection{Dynamical reservoirs with two-level systems.} \label{sec:conc_ising}

Dynamical reservoirs are generated by the unitary dynamics of interacting quantum many-body systems. As summarized in Table~\ref{tab:QRP_reservoirs_main}, we consider four experimentally motivated Hamiltonian families, including (i)~the disordered transverse-field Ising model (TFIM), (ii)~Rydberg-atom arrays, (iii)~the SYK model of interacting fermions and lastly (iv)~periodically kicked Ising dynamics as a Floquet counterpart. Across all of these families we observe the same organizing principle, that is the strength of output concentration is dictated by the dynamical phase in which the reservoir operates, together with the evolution time per iteration. In what follows we present the disordered TFIM as our main case study, and summarize the results for the remaining models at the end of this section.

\begin{figure}[htbp]
    \centering
    \includegraphics[width=0.9\linewidth]{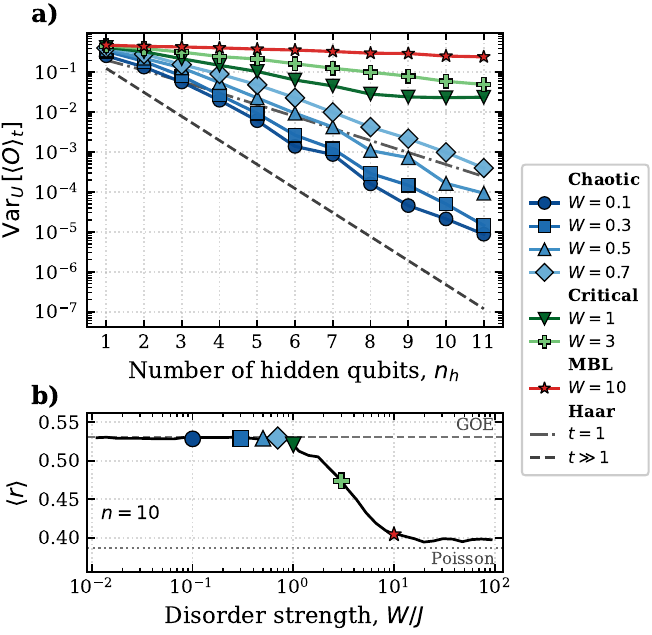}
    \caption{\textbf{Concentration of the readout variance of transverse-field Ising reservoirs in different dynamical phases.} 
    \textbf{a)} The variance of the QRP output over the quenched-disorder reservoir ensemble is plotted against the number of hidden qubits $n_h$, for the TFIM reservoir of Eq.~\eqref{eq:tfim_reservoir} with $n_a = 1$, transverse field $h = J$, evolution time $\Delta t = 10/J$, and $t=15$ reservoir iterations. The disorder strengths $W$ are color-coded according to the three dynamical regimes: chaotic (ergodic) for $W \leq 0.7$ in blue, critical for $W = 1, 3$ in green, and many-body localized for $W = 10$ in red. The dashed black line marks the iterated extreme-scrambling (Haar) prediction $1/(d_a d_h^{2})$ of Theorem~\ref{Thm:HaaResVar}, while the dash-dotted gray line marks the variance after a \emph{single} iteration of a Haar-random reservoir, $1/(d_a d_h + 1)$ (cf.\ the small-$t$ analysis in Appendix~\ref{app:var_small_t}). \textbf{b)} Mean level-spacing ratio $\expval{r}$ of the reservoir Hamiltonian in Eq.~\eqref{eq:tfim_reservoir} along the same cut $h = J$, resolved within the parity sectors of $\prod_i Z_i$ and averaged over disorder realizations, for a system of $n = 10$ qubits. The horizontal lines mark the GOE value $\expval{r}_{\rm GOE} \approx 0.531$ and the Poisson value $\expval{r}_{\rm P} \approx 0.386$, and the colored symbols sit at the disorder strengths of panel \textbf{a)}, confirming the assignment of the three regimes; see panel \textbf{a)} of Fig.~\ref{fig:ising_heatmap} for the full $(h, W)$ plane.}
    \label{fig:var_vs_nh_phases}
\end{figure}

\begin{figure*}[htbp]
    \centering
    \includegraphics[width=\textwidth]{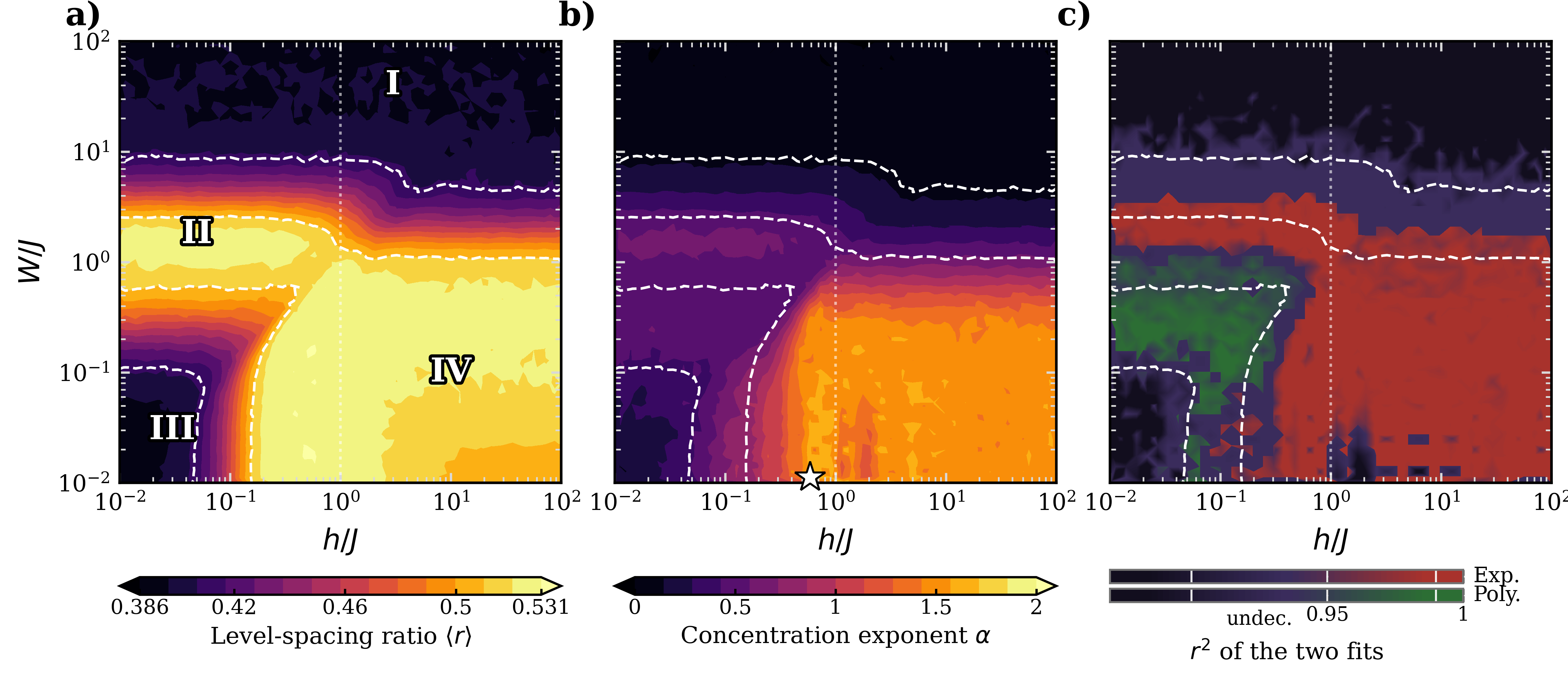}
    \caption{\textbf{Dynamical phase diagram, concentration exponent, and confidence in the exponential model of the disordered TFIM reservoir.} All three panels span the same $(h/J, W/J)$ plane for the Hamiltonian of Eq.~\eqref{eq:tfim_reservoir}, so that the maps can be compared point by point. \textbf{a)} Disorder-averaged level-spacing ratio $\expval{r}$, computed within the parity sectors of $\prod_i Z_i$ for a system of $n = 10$ qubits. The color scale runs from the Poisson value $\expval{r}_{\rm P} \approx 0.386$ (bottom) to the GOE value $\expval{r}_{\rm GOE} \approx 0.531$ (top), so bright regions indicate ergodic (quantum-chaotic) dynamics and dark regions localized or near-integrable dynamics. \textbf{b)} Concentration exponent $\alpha$, obtained by fitting the readout variance to the exponential decay model of Eq.~\eqref{eq:alpha_fit} over $n_h \in \{1,\dots,8\}$, for $n_a=1$, evolution time $\Delta t = 10/J$, and $t = 10$ reservoir iterations. The color scale runs from $\alpha = 0$ (no concentration) to the extreme-scrambling limit $\alpha = 2$ of Theorem~\ref{Thm:HaaResVar}; the star marks the maximal exponent $\alpha_{\max} = 1.68$ at $(h, W) \approx (0.59, 0.011)$. \textbf{c)} Which decay law the data supports at each point: red where the exponential model of Eq.~\eqref{eq:alpha_fit} is selected, green where the polynomial model of Eq.~\eqref{eq:beta_fit} is, and violet to black where the series is undecided, applied over the same window. The two bars share the undecided scale below the cut-off $r^2 \leq 0.95$. In all panels the dashed vertical line marks the cut $h = J$ along which Fig.~\ref{fig:var_vs_nh_phases} is taken.}
    \label{fig:ising_heatmap}
\end{figure*}

\paragraph*{Settings of TFIM reservoirs.}
Ising-type Hamiltonians form the most common reservoir family in QRP~\cite{fujii2017harnessing,negoro2018machine,nakajima2019boosting,kutvonen2020optimizing,martinez2021dynamical,li2026quantum,hou2026high,ivaki2024quantum,settino2025topology,mccaul2025minimal,cindrak2024enhancing,delorenzis2025harnessing,fuchs2024quantum,xia2023configured}. Experimental implementations have been demonstrated in various platforms, including superconducting qubits~\cite{kawanabe2026efficient}, neutral atoms~\cite{kornjaca2024large} as well as nuclear spins~\cite{hou2026high}.

Following Refs.~\cite{fujii2017harnessing,martinez2021dynamical}, we model the reservoir as the quantum evolution of a fully connected TFIM with onsite disorder on $n = n_a + n_h$ total qubits,
\begin{align}\label{eq:tfim_reservoir}
    H = \sum_{i<j} J_{ij} X_i X_j + \sum_{i=1}^{n} (h + D_i) Z_i \;\;,
\end{align}
where $X_i$ and $Z_i$ are Pauli operators acting on the $i$-th qubit. The interaction couplings $J_{ij}$ are all-to-all and independently drawn from the uniform distribution on $[-J/2, J/2]$, while the field along the $z$-axis comprises a uniform component $h$ and onsite disorder $D_i$ drawn uniformly from $[-W, W]$. We fix the characteristic energy scale $J=1$. The reservoir unitary is $U = \exp(-i H \Delta t)$ with evolution time $\Delta t$ per iteration; the disorder realization $\{J_{ij}, D_i\}$ is drawn once per reservoir instance and the same $U$ is applied at every iteration, as prescribed by the QRP framework of Section~\ref{sec:QRP_framework}. At each time step, the accessible qubits are re-initialized in $\ket{+}^{\otimes n_a}$ and the classical input $s_\tau$ is encoded through exponentially scaled single-qubit rotations $\bigotimes_{k=1}^{n_a} e^{-i 3^{k-1} s_\tau Z_k}$, while the hidden qubits are initialized in $\ket{+}^{\otimes n_h}$ at $\tau=0$ and never reset. The input sequence is a fixed window of the Mackey--Glass time series, so that the reported variance isolates the fluctuations over the reservoir ensemble, in direct correspondence with Theorem~\ref{Thm:HaaResVar}. The readout observable is the Pauli-$X$ operator on the first accessible qubit, and we take $n_a = 1$ throughout this section.

\paragraph*{Dynamical phases of the TFIM reservoir.}
The computational behavior of a TFIM reservoir is intrinsically tied to its underlying dynamical phase, which is governed by the interplay between the field strength $h/J$ and the disorder strength $W/J$ \cite{martinez2021dynamical}. The phase diagram of this all-to-all model comprises four dynamical regions~\cite{martinez2021dynamical}:
\begin{itemize}
    \item \textbf{Many-body localized (MBL) paramagnet:} at strong disorder (i.e., $W/J \gg 1$). 
    where localization prevents the system from thermalizing.
    \item \textbf{Localized spin glass:} a localized, glass-like phase at weak disorder and weak fields (i.e., $W/J \ll 1$ and $h/J \ll 1$).
    \item \textbf{Ergodic spin glass:} at weak-to-moderate fields and moderate disorder, marking the onset of thermalizing dynamics.
    \item \textbf{Ergodic paramagnet:} at moderate-to-strong fields and weak disorder (i.e., $W/J \lesssim 1$, $h/J \approx 1$), where the system acts as its own heat bath, rapidly scrambling local information and exhibiting Wigner--Dyson level statistics.
\end{itemize}

To chart these phases for the system sizes used in our simulations, we employ the mean level-spacing $\langle r\rangle$, where each level spacing defined by $r_k = \min(\delta_k, \delta_{k+1})/\max(\delta_k, \delta_{k+1})$ with $\delta_k = E_{k+1} - E_k$ the consecutive spacings of the ordered eigenvalues of $H$ (with $E_k$ an eigenenergy associated with a $k^{\rm th}$ eigenstate)~\cite{martinez2021dynamical}. 
This phase diagram distinguishes ergodic dynamics, for which $\expval{r}$ approaches the Gaussian-orthogonal-ensemble (GOE) value $\expval{r}_{\rm GOE} \approx 0.531$, from many-body localized or near-integrable dynamics, for which $\expval{r}$ approaches the Poisson value $\expval{r}_{\rm P} \approx 0.386$. Fig.~\ref{fig:ising_heatmap}\textbf{a)} shows the disorder-averaged $\expval{r}$ across the $(h/J, W/J)$ plane. The ergodic (chaotic) lobe is clearly visible at moderate fields $h \approx J$ and weak disorder $W \lesssim J$, while Poisson-like statistics dominate at strong disorder (MBL), at weak fields (approaching the classical Ising limit), and at very strong fields (approaching the trivially polarized limit).

\paragraph*{Concentration across the phase diagram.}
Fig.~\ref{fig:var_vs_nh_phases} shows the variance of the QRP output over the reservoir ensemble as a function of the number of hidden qubits $n_h$, for reservoirs at fixed field $h = J$ and disorder strengths. %

Deep in the ergodic phase ($W \leq 0.7$), all variances decay exponentially over the full accessible range of system sizes. Crucially, the decay is slower than the iterated extreme-scrambling prediction (dashed line). However, it is significantly \emph{faster} than a single iteration of a Haar-random unitary (indicated by another dashed line in the plot), highlighting the temporal effect of repeatedly applying the same reservoir instance. %
In contrast, reservoirs in the critical and MBL regimes concentrate much more mildly. The variance decays slowly with $n_h$, and the data are consistent with sub-exponential scaling over the sizes studied. %

To make the link between dynamical phases and concentration more explicit, we fit the variance to the exponential model
\begin{align}\label{eq:alpha_fit}
    \Var_U [\expval{O}_t] = C \cdot 2^{-\alpha n_h}\;,
\end{align}
where the concentration exponent $\alpha$ measures how fast the output concentrates with the number of hidden qubits and $C$ is some coefficient. In this form, the extreme-scrambling reservoirs of Theorem~\ref{Thm:HaaResVar} correspond to $\alpha = 2$, while a \emph{single} iteration with a Haar-random reservoir yields the milder rate $\alpha = 1$ (see the small-$t$ analysis in Appendix~\ref{app:var_small_t}).

Fig.~\ref{fig:ising_heatmap}\textbf{b)} shows the fitted exponent $\alpha$ across the same $(h, W)$ plane as the phase diagram of Fig.~\ref{fig:ising_heatmap}\textbf{a)}. The strongly concentrating region closely mirrors the ergodic lobe identified by the level statistics, with the maximal exponent $\alpha_{\max} = 1.68$ attained at weak disorder and moderate field $(h, W) \approx (0.59, 0.011)$, while the exponent is strongly suppressed throughout the localized regions. Interestingly, appreciable concentration persists at large fields $h \gg J$ where $\expval{r}$ has already dropped below the GOE value, indicating that the finite-time dynamics at the operating point can remain effectively scrambling even when the asymptotic spectral statistics are no longer chaotic.

The exponent of Eq.~\eqref{eq:alpha_fit} is informative only where the exponential form actually describes the data. To test this, we compare it at every point of the plane with a polynomial decay with power $\beta$ and some coefficient $C$,
\begin{align}\label{eq:beta_fit}
     \Var_U [\expval{O}_t] = C \cdot n_h^{-\beta}\;,
\end{align}
fitted over the same sizes and we record the coefficient of determination $r^2$ of each fit. Two conditions must hold before either model is credited. First, the variance must vary by at least a factor of four across the fitted sizes; below that the series is \emph{too flat to tell}, since any decay law with a sufficiently small exponent describes a near-constant equally well. Second, the better of the two fits must reach $r^2 = 0.95$. Points failing either condition are labelled undecided and elsewhere the model with the larger $r^2$ is selected. The outcome is shown in Fig.~\ref{fig:ising_heatmap}\textbf{c)}. The exponential law is selected throughout the ergodic lobe for $81\%$ of the points with $\expval{r} \geq 0.46$, and for $86\%$ of those with $\expval{r} \geq 0.51$. On the other hand, in the localized and near-integrable regions the variance is almost independent of $n_h$ and no law can be selected ($88\%$ undecided). The polynomial law is preferred in a band at weak field ($h \lesssim 0.6 J$, $W \lesssim J$) along the lower edge of the ergodic lobe, where the decay is genuine but slow (mean exponent $\alpha = 0.54$, against $1.12$ in the exponentially concentrating region). Panel~\textbf{b)} should accordingly be read as a concentration rate where panel~\textbf{c)} is red, and as evidence that the concentration is at most algebraic where it is dark.

\paragraph*{Dependence on reservoir iterations.}
We next study how the concentration builds up over reservoir iterations for a reservoir operated deep in the ergodic phase. As shown in Fig.~\ref{fig:ising_var_vs_n_h_t_example}, the variance decays over the first few iterations and then saturates, by $t \approx 10$, to a plateau whose value decreases exponentially with the number of hidden qubits. 
This behavior closely mirrors that of the extreme-scrambling reservoirs in Fig.~\ref{fig:conc_U}. That is, early iterations accumulate additional concentration, but once saturated, further reuse of the reservoir does not worsen the scaling. The qualitative features of the design-based analysis thus carry over to a fixed, physical chaotic Hamiltonian. 

\begin{figure}[htbp]
    \centering
    \begin{minipage}[b]{0.48\textwidth}
        \centering
        \includegraphics[width=\linewidth]{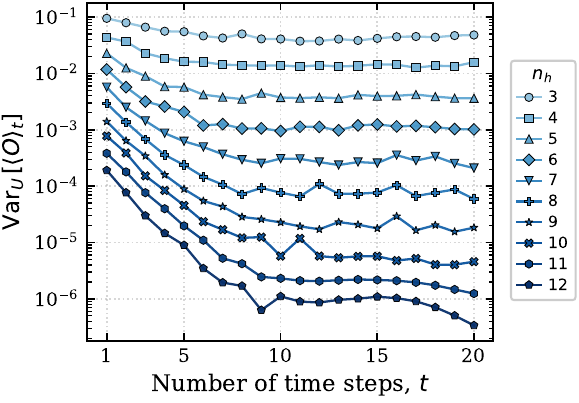}
    \end{minipage}
    
    \caption{\textbf{Build-up and saturation of concentration for an ergodic TFIM reservoir.} The variance of the QRP output over the quenched-disorder ensemble is plotted against the number of reservoir iterations for different hidden-register sizes $n_h$ (left), and against $n_h$ for different iteration numbers $t$ (right). The reservoir of Eq.~\eqref{eq:tfim_reservoir} is operated deep in the ergodic phase, with $h = 0.75 J$, $W = 0.015 J$, $\Delta t = 15/J$, and $n_a = 1$. The variance saturates after $t \approx 10$ iterations to a plateau that decreases exponentially with $n_h$, mirroring the extreme-scrambling behavior of Fig.~\ref{fig:conc_U}.}
    \label{fig:ising_var_vs_n_h_t_example}
\end{figure}

\paragraph*{Dependence on the evolution time.}
The other key factor determining the concentration rate is the Hamiltonian evolution time $\Delta t$ per iteration, which controls how much scrambling each reservoir iteration generates. Fig.~\ref{fig:ising_alpha_dt} shows the fitted exponent $\alpha$ as a function of $\Delta t$ for several disorder strengths in the ergodic phases at fixed $h/J=1$ and the number of iterations $t=15$. For short evolution times, $\Delta t \lesssim 1/J$, the reservoir remains relatively close to the identity and essentially no concentration occurs ($\alpha \approx 0$). The exponent then grows rapidly with the evolution time, becoming significant already for $\Delta t \gtrsim 2/J$, peaking near $\alpha \approx 1.9$ around $\Delta t \approx 15/J$ for the weakest disorder, and settling into a plateau close to the peak value for $\Delta t \gtrsim 100/J$. We observe that increasing the disorder delays the onset and lowers the plateau. 
This identifies the evolution time as a practical design knob, restricting $\Delta t$ bounds the concentration rate even for a reservoir whose Hamiltonian is deeply chaotic. Applying the model selection of Fig.~\ref{fig:ising_heatmap}\textbf{c)} to these fits, Fig.~\ref{fig:ising_alpha_dt}\textbf{b)} shows that the exponential form is selected for all but three of the $64$ fits with $\Delta t \gtrsim 3/J$, while for $\Delta t \lesssim 2/J$ the variance is too flat over the accessible sizes (varying by less than a factor of four) for either law to be credited. The vanishing exponent at short evolution times thus reflects the absence of concentration rather than a slow exponential decay.

\begin{figure}[htbp]
    \centering
    \includegraphics[width=\linewidth]{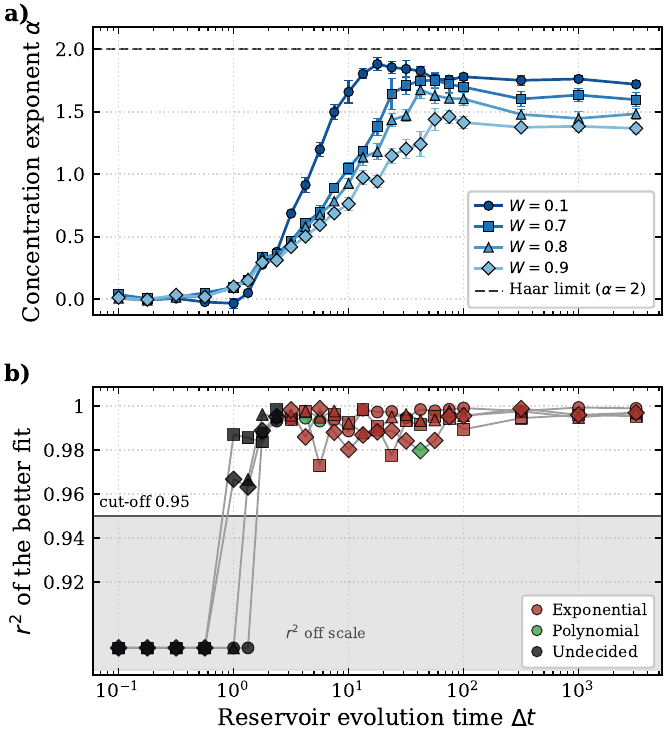}
    \caption{\textbf{Concentration exponent $\alpha$ as a function of the reservoir evolution time $\Delta t$, and the confidence of the underlying fits.} \textbf{a)} The exponent is extracted by fitting the readout variance to the exponential decay model of Eq.~\eqref{eq:alpha_fit}, for TFIM reservoirs with fixed field $h = J$ and varying disorder strengths $W$ ($n_a = 1$). Error bars indicate the fit uncertainty. The dashed black line at $\alpha=2$ marks the extreme-scrambling (Haar) limit of Theorem~\ref{Thm:HaaResVar}. \textbf{b)} Model selection for the same fits, by the criterion of Fig.~\ref{fig:ising_heatmap}\textbf{c)}: the vertical axis is the $r^2$ of the better of the two models and the marker colour is the law selected --- exponential (red), polynomial (green), or undecided (black) --- with the region below the cut-off $r^2 = 0.95$ shaded. Points whose variance is flat across the fitted window are pinned to the bottom of the axis.}
    \label{fig:ising_alpha_dt}
\end{figure}

\paragraph*{Other physical reservoirs.} To verify that these findings are not specific to a certain model, we have performed similar studies for two further experimentally motivated models; the detailed results are presented in Appendices~\ref{app:rydberg}--\ref{app:kicked_ising}, and we summarize the key observations here.
\begin{itemize}
    \item \emph{SYK model} (Appendix~\ref{app:syk}): for a complex SYK reservoir perturbed by a quadratic (integrable) term, the concentration strength tracks the degree of chaos as the relative coupling is tuned; near the edge of chaos, the concentration is substantially milder than deep in the chaotic phase, singling out near-critical operating points as candidates for scalable reservoirs.
    \item \emph{Kicked Ising dynamics} (Appendix~\ref{app:kicked_ising}): a strongly kicked, disordered Floquet Ising reservoir can be tuned such that its output variance reaches the extreme-scrambling (Haar) scaling of Theorem~\ref{Thm:HaaResVar}, demonstrating that physically implementable dynamics can saturate the design-based prediction. 
\end{itemize}

Remarkably, the edge of chaos (i.e., the regime near the critical point) appears markedly less susceptible to concentration. This has been empirically observed both in Fig.~\ref{fig:var_vs_nh_phases} for the Ising reservoir and for the SYK reservoir in Appendix~\ref{app:syk}. This regime has previously been identified as the operating point that maximizes reservoir performance, at least at the system sizes accessible to numerics~\cite{kobayashi2026edge,martinez2021dynamical}. Our results suggest that it may also mitigate output concentration, which further strengthens the case for operating near criticality.

{
While we have not explicitly explored the memory capacity of quantum reservoirs composed of physical dynamics here, we expect many of our conclusions on concentration to carry directly over. As shown in Proposition~\ref{prop:prac-cons}, exponential concentration renders the output increasingly insensitive to the input. In particular, the proof only requires the concentration to be exponential in the system size and is otherwise independent of the precise rate of concentration (see Appendix~\ref{app:concentration-consequence-summary}). Intuitively, the insensitivity of the output to the input implies that the output retains vanishingly little information about its input, and hence has correspondingly limited memory of the input. It is therefore natural to expect the memory capacity of QRP in the deep ergodic phase to decay exponentially with the system size, similar to the scaling established in Theorem~\ref{thm_main:mem_ind_upp_bound}.

Taken together, these results point to a generic relationship independent of some specific physical model between dynamical phases, the degree of scrambling, and output concentration in dynamical reservoirs, which in turn also implies the reservoir's ability of memory retention.
}

}
\section{Discussion}\label{sec:discussion}
In this work, we analytically study the scalability and memory properties of scrambling quantum systems used for temporal information processing, by considering QRP as a general framework, with an extreme scrambling reservoir modeled as a high-order unitary design.
We introduce a tensor-diagram technique that unrolls the sequence of identical reservoir dynamics, casting it into a single-step higher-order moment integration. Our scalability analysis shows that QRP outputs concentrate exponentially with increasing reservoir size, hindering the use of extreme scramblers in the large-system size limit. 

However, once the concentration saturates it no longer worsens with increasing iterations. 
This provides a small window of opportunity in which QRP might yield a polynomial advantage, compared to standard tomography, on temporal tasks that process quantum data in a sufficiently small number of qubits.
In particular, the sampling complexity of performing quantum tomography followed by classical post-processing could be saved by processing the states directly within the QRP framework. This could be specifically relevant for input state ensembles with high rank, e.g. states generated from complex many-body systems, since
tomography requires $\Omega(\mathrm{rank}(\rho)d/\epsilon)$ samples~\cite{yuen2022improved}. These insights thus suggest an opportunity to utilize moderate-sized noisy quantum systems as a substrate for physical computing~\cite{preskill2018quantum, zimboras2025myths} and hint towards the possibility of achieving practical information processing advantages~\cite{senanian2024microwave}.

One potential approach to successfully scale QRPs to larger problem sizes is to incorporate symmetry into the construction of QRP. An analytical analysis on a single reservoir iteration has shown that working in the low dimensional symmetric subspace avoids exponential concentration~\cite{sannia2025ExpConcQRC}, in direct analogy with the case for variational quantum algorithms~\cite{fontana2023theadjoint, ragone2023lie}. Proving this for multiple reservoir iterations requires integrating over the symmetric group with the unrolled form of the QRP output~\cite{xiong2025QRPSymmetry}. Nonetheless, it is important to note that variational QML models which provably avoid exponential concentration are susceptible to classical surrogatability (i.e., they can be classical simulated after performing some form of state tomography)~\cite{cerezo2023does}. A similar relation is also expected in QRP.

In addition, we have analyzed the memory properties of extreme-scrambling reservoirs. By introducing \emph{memory indicators} to capture the characteristics reminiscent of Echo State
Property (ESP) and Fading Memory Property (FMP) \emph{on average}, we applied the
tensor-diagram technique  to show that these indicators indeed decay
exponentially with the time step. This rapid decay reflects the expected
forgetting behaviour associated with ESP and FMP. However, taken together with the scalability barrier, our results provide an explicit example of a QRP that satisfies averaged ESP- and FMP-like characteristics yet remains practically non-operational (input-independent and poor generalization) at large scale.  Thus achieving these memory properties alone is insufficient, and scalable QRP design requires a broader set of considerations{, such as whether the memory-decay timescale is compatible with the task-relevant memory horizon.}

{Measurement-based and feed-forward protocols have been proposed in which intermediate outcomes control later reservoir operations~\cite{mujal2023time,franceschetto2024harnessing,kobayashi2024feedback,murauer2025feedback,monomi2025feedback,morguisancho2026general}. 
Our analysis in Appendix~\ref{app:adaptive_mid_meas_feedforward} shows that while the final outcomes remain exponentially concentrated with reservoir size, such broader hybrid architectures can instead retain an intermediate outcome in a persistent classical register and use it directly in the final prediction. Related hybrid QRP schemes likewise use classical post-processing or a classical reservoir to augment temporal memory~\cite{settino2024memory,nokkala2024retrieving}. The counterexample in the same Appendix shows that information stored in this way need not decay over subsequent reservoir iterations. Consistently, recent work on monitored scrambling dynamics shows that retaining the full mid-circuit measurement trajectory can extend the quantum-information lifetime exponentially with system size, whereas discarding the trajectory yields much shorter memory~\cite{zhang2026scaling}. This does not mean that the quantum reservoir itself retains the information, nor does it establish that classical-record feedback removes the system-size scalability barrier: the present proposition does not determine how the input dependence initially extracted into the classical record scales with reservoir size.}

{Mid-circuit measurement and feed-forward protocols, in which intermediate outcomes control and condition later reservoir operations, have been proposed for temporal information processing~\cite{mujal2023time,franceschetto2024harnessing,kobayashi2024feedback,murauer2025feedback,monomi2025feedback,morguisancho2026general}. 
These measurement-induced feedback protocols admit two distinct output options.
In the first case, the intermediate measurement record only conditions later inputs and measurements, and the final output remains the expectation value of a fixed quantum observable. For the extreme-scrambling model analyzed in Appendix~\ref{app:adaptive_mid_meas_feedforward}, Supplemental Proposition~\ref{sup-prop:adaptive-mid-meas-feedforward-conc} shows that, at a fixed number of reservoir iterations, adaptive measurement and feedback \emph{do not circumvent} the exponential concentration of the final quantum readout even when the sequences of intermediate measurement outcomes are efficiently sampleable, i.e., can be estimated with polynomially many shots.}

\stadd{In the other case, the intermediate measurement records are instead stored in a classical register and then also used in training or inference. The construction in Appendix~\ref{app:adaptive_mid_meas_feedforward} exemplifies that the temporal memory may not need to decay exponentially in this case. Namely, once information has been transferred to a persistent classical record, subsequent reservoir iterations cannot erase it. The memory is therefore stored outside the quantum reservoir state. Related hybrid QRP schemes use classical post-processing or a classical reservoir to improve temporal memory~\cite{settino2024memory,nokkala2024retrieving}; similarly, conditioning on the complete measurement record can extend the quantum-information lifetime in monitored dynamics~\cite{zhang2026scaling}. This classically stored persistent memory does not, however, guarantee that a sufficiently strong input-dependent signal was recorded initially. In our explicit binary example, this signal is already exponentially suppressed with reservoir size when it is recorded. Whether more general feedback protocols or less strongly scrambling reservoirs can produce a recorded signal that remains experimentally resolvable as the system grows remains an open question.}

{Dissipative and Lindbladian reservoirs provide an interesting complementary QRP design approach. Controlled non-unitarity can be a computational resource rather than merely a source of error: by tuning dissipation and, more generally, open-system memory effects, one can control washout, reshape the balance between short- and long-term memory, and improve performance in \emph{finite} reservoirs for suitable tasks and parameter regimes~\cite{kora2024frequency,mifune2024effects,cheamsawat2025dissipation,sannia2024dissipation,yosifov2025dissipation,gotting2025connection,sannia2025non}. These results suggest that carefully designed open-system dynamics could potentially alleviate concentration-related limitations in temporal QRP. This does not contradict our noise bounds. Under the noise models considered, the bounds quantify how rapidly the readout loses sensitivity to past inputs. Dissipation is beneficial when it removes information older than the task's required memory window, but becomes detrimental when it also erases inputs that the task still needs too rapidly. At present, however, neither our analysis nor the cited finite-system studies establish that engineered dissipation removes the system-size-dependent output concentration derived here. Determining whether it can preserve experimentally resolvable recent input signals at scale remains an open question.}

\stadd{Lastly, our empirical study of physical reservoirs sharpens the scope of these barriers. Exponential concentration of QRP output is not an artifact of the unitary-design idealization. Physical reservoirs operated deep in an ergodic phase, and evolved sufficiently long between input injections, concentrate exponentially in the system size with strongly kicked Ising dynamics reaching the extreme-scrambling rate. On the other hand, localized, near-critical, and shallow-circuit reservoirs concentrate only mildly. The dynamical phase of the reservoir, diagnosed by standard spectral statistics, together with the evolution time per iteration, thereby emerges as a concrete design knob for balancing expressivity against scalability. Whether mildly concentrating regimes particularly reservoirs operated near the edge of chaos retain sufficient expressive power while evading both exponential concentration and classical surrogatability remains an important open question. In addition, while we strongly speculate that the relation between scalability and phases of matter is model-agnostic, it would nevertheless be valuable to explicitly extend this study to other physical models, including Rydberg atoms\footnote{Note that a preliminary study of the role of decaying interaction strength through the lens of a Rydberg Hamiltonian is included in Appendix~\ref{app:rydberg}.}, bosonic systems, and hybrid platforms~\cite{nokkala2021gaussian,dudas2023quantum,senanian2024microwave}.}

\medskip

\textit{Related work:} During the preparation of manuscript, we became aware of a related independent work~\cite{sannia2025ExpConcQRC} which also studied exponential concentration in QRP but with a different focus. While our analytical results have focussed on general QRP models with time correlated reservoir dynamics, Ref.~\cite{sannia2025ExpConcQRC}'s analytical results focus on a single reservoir iteration and explores the effect of symmetry.

\section*{Acknowledgements}
The authors are indebted to Marco Cerezo, who showed us how to handle time correlated Haar random unitaries while waiting for tacos. 
The authors also gratefully acknowledge Rodrigo Martínez-Peña for valuable feedback on an earlier draft of this paper, and further acknowledge interesting discussions with Matthew Duschenes, Giorgio Facelli, Gian Luca Giorgi, Daniel Stilck França, Jonas Landman, Martin Larocca, Sacha Lerch, Hela Mhiri, Juan-Pablo Ortega, Barry Sanders, Mehrad Sahebi, Antonio Sannia, and Roberta Zambrini.

WX acknowledges the funding support from NCCR SPIN, a National Centre of Competence in Research, funded by the Swiss National Science Foundation (grant number 225153). ZH, AA, and ST acknowledge support from the Sandoz Family Foundation-Monique de Meuron program for Academic Promotion. TC and ST acknowledge the funding support from the NSRF via the Program Management Unit for Human Resources \& Institutional Development, Research and Innovation [grant number B39G680007]. ST also acknowledges the grants for development of new faculty staff, Ratchadaphiseksomphot Fund, Chulalongkorn University [grant number 3230120336 DNS\_68\_052\_2300\_012], as well as funding from National Research Council of Thailand
(NRCT) [grant number N42A680126].

\section*{Declarations}
\subsection*{Author Contributions}
The project was conceived by Z.H., T.C., and S.T. Z.H. developed the core analytical tool. For the analytical results, W.X., Z.H., T.C. and S.T. performed the scalability and memory analysis, and A.A. and W.X. derived the noise bounds. W.X. performed the numerical simulations. The manuscript was written by W.X., Z.H., A.A., Y.S., T.C. and S.T.

\subsection*{Articial Intelligence (AI)}
Generative AI was used to assist with parts of the numerical simulation code and to polish the manuscript. All AI-assisted content was carefully checked and verified by the authors.

\subsection*{Competing interests}
The authors declare no competing interests.

\bibliography{quantum.bib,otherbib.bib}
\clearpage
\setcounter{theorem}{0}
\setcounter{proposition}{0}
\setcounter{corollary}{0}

\onecolumngrid

\renewcommand\partname{} 
\appendix
{\huge \textbf{Appendices}}
\part{}
\parttoc 
\clearpage

\newpage
\section{Extreme scrambling reservoir dynamics}\label{app:Haar_proof}

In this section, we provide technical details regarding  the impact of the extreme scrambling reservoirs on the scalability and memory properties of QRP, including the proofs of Theorem~\ref{Thm:HaaResVar} and Theorem~\ref{thm_main:mem_ind_upp_bound}. 

\medskip

\noindent\textbf{Setting and technical challenge.} For the input time series of length $t$, we model the extreme scrambling reservoir as an instance drawn from a $2t$-design unitary ensemble of dimensions $d_a d_h$. Formally, we define the ensemble of the $t$-order extreme scrambling reservoirs $\UC_{\SC}^{(t)}(d_ad_h)$ such that it satisfies
\begin{align}\label{def:extreme-scrambling-u}
    \int_{U \sim \UC_{\SC}^{(t)}(d_ad_h)} \left(U \otimes U^\dagger\right)^{\otimes 2t} d\mu(U) = \int_{U \sim \Ubb(d_nd_h)} \left(U \otimes U^\dagger\right)^{\otimes 2t} d\mu_H(U) \;,
\end{align}
where $d\mu(U)$ and $d\mu_H(U)$ denote the measure associated with $\UC_{\SC}^{(t)}(d_ad_h)$ and the Haar measure, respectively.

Modelling the extreme reservoirs with Eq.~\eqref{def:extreme-scrambling-u} enables the Haar integral tools to analytically analyse the QRP dynamics. Nevertheless, the technical challenge stands due to the temporal correlations of the reservoir dynamics at different time steps. To see this, recall that the recursive relations of the QRP states at two consecutive time steps in Eq.~\eqref{eq:qrp-total-state}
\begin{align}\label{eq:app-qrp-total-state}
          \rhot{\tau}
        =
         U\bigl(\rhoa{\tau} \,\otimes\, \rhoh{\tau-1}\bigr)\,
         U^\dagger \;.
\end{align}
In particular, for all time steps, the reservoir dynamic $U$ is identical. This recursive temporal form hinders us from directly applying the standard Haar integral formulae.

\medskip

\noindent\textbf{Overview of our technical contributions.} To tackle the above challenge, we express the QRP output in an alternative analytical form as a single time-step on a larger dimension which is achieved by systematically unrolling multiple time step with the tensor-diagram approach, as shown in Appendix~\ref{app:unroll}. 
Then, we perform high-order Haar integration and carefully keep track of dominant terms for asymptotic scaling analysis of QRP to show exponential concentration of the outputs and exponential memory vanishing. In particular, the subsections are structured as:
\begin{itemize}
    \item In Appendix~\ref{app:prelim-haar}, we provide existing theoretical results necessary to perform high-order Haar integration.
    \item In Appendix~\ref{app:unroll}, we present the alternative form of QRP output unrolled as a single reservoir iteration on an extended space in Supplemental Proposition~\ref{sup-prop:qrp-output} together with the detailed proof.
    \item In Appendix~\ref{app:conc_output}, we prove Theorem~\ref{Thm:HaaResVar} which shows the vanishing variance of QRP output.
    \item In Appendix~\ref{app:mem_conc}, we prove Theorem~\ref{thm_main:mem_ind_upp_bound} which shows the exponential vanishing of memory in QPP.
\end{itemize}

\subsection{Preliminaries}\label{app:prelim-haar}

We present the preliminary results which allow us to perform high moment Haar integration. This is a necessary tool for identifying dominant contributions for scalability and memory analysis.

\begin{lemma}[Leading term of $k$-th moment of Haar integral (Lemma 3 of Ref.~\cite{garcia2023deep})]\label{lemma:L3_deep_QNN} Let $X$ be an operator in $\BC(\HC^{\otimes k})$, the twirl of $X$ over $\Ubb(d)$, defined as
\begin{equation}
    \mathcal{T}_{\Ubb(d)}^{(k)}[X]:=\int_{\Ubb(d)}d\mu(U)U^{\ot k}X(U^\dag)^{\ot k}\,,
\end{equation}
equals
\begin{equation}
    \mathcal{T}_{\Ubb(d)}^{(k)}[X]=\frac{1}{d^k}\sum_{\sg\in P_d(S_k)} \Tr[X \sg]\sg^{-1}+\frac{1}{d^k}\sum_{\sg,\pi\in P_d(S_k)}c_{\sg,\pi}\Tr[X \sg]\pi\,,
\end{equation}
where $P_d(S_k)$ is the set of permutation operators applied on the copies of spaces with dimension $d$ induced by the elements of symmetric group $S_k$, and the constants $c_{\sg,\pi}$ are in $\OC(1/d)$.
\end{lemma}

\noindent\underline{\textit{Remark:}} The elements in the symmetric group $S_k$ can be denoted as collection of distinct tuples of indices from $[k]$. In the representation for the operators, these elements are essentially presented as SWAP operators between different copies/subsystems. More specifically, consider two subsystems $A$ and $B$; each of dimensions $d$ with orthornormal basis $\{|i\rangle\}_{i=1}^d$ , the SWAP operator is defined as ${\rm SWAP}_{A,B} = \sum_{i,j} |i \rangle\langle j |_A \ot |j\rangle\langle i|_B$. For two operators $O_A$ and $O_B$ on the subsystems $A$ and $B$ respectively, one useful identity follows
\begin{align}
    \Tr[(O_A \ot O_B)\,{\rm SWAP}_{A,B}] = \Tr[O_A O_B] \;.
\end{align}

Here we further provide examples of $S_k$ and the induced $\Tr[X\sigma]$ with $\sg\in P_d(S_k)$ for small $k$'s: 
\begin{itemize}
    \item $k=1$: $S_1=\{(1)\}$. \\
    For $s=(1)$, we have $\sg=P_d(s)=\IC$, and hence $\Tr[X\sg]=\Tr[X]$.
    \item $k=2$: $S_2=\{(1)(2),(21)\}$.\\
    For $s=(1)(2)$, we have $\sg=P_d(s)=\IC\otimes\IC$, and hence $\Tr[(X_{h_1}\otimes X_{h_2})\sg]=\Tr[X_{h_1}]\Tr[X_{h_2}]$.\\
    If $s=(21)$, then $\sg=P_d(s)=\rm{SWAP}$, and $\Tr[(X_{h_1}\otimes X_{h_2})\sg]=\Tr[X_{h_1} X_{h_2}]$. 
    \item $k=3$: $S_3=\{(1)(2)(3),(1)(23),(12)(3),(13)(2),(123),(132)\}$.\\
    If $s=(1)(23)$, then $\sg=P_d(s)=\IC_1\otimes\rm{SWAP_{23}}$, and $\Tr[(X_{h_1}\otimes X_{h_2} \otimes X_{h_3})\sg]=\Tr[X_{h_1}]\Tr[X_{h_2} X_{h_3}]$.\\
    If $s=(123)$, then $\sg=P_d(s)=(\IC_1\otimes\rm{SWAP_{23}})(\rm{SWAP_{12}}\otimes \IC_3)$, and $\Tr[(X_{h_1}\otimes X_{h_2} \otimes X_{h_3})\sg]=\Tr[X_{h_1} X_{h_2} X_{h_3}]$.\\
    If $s=(132)$, then $\sg=P_d(s)=(\rm{SWAP_{12}}\otimes \IC_3)(\IC_1\otimes\rm{SWAP_{23}})$, and $\Tr[(X_{h_1}\otimes X_{h_2} \otimes X_{h_3})\sg]=\Tr[X_{h_1} X_{h_3} X_{h_2}]$.
    \item Arbitrary k: $S_k=\{(\alpha_{1,1},\alpha_{1,2},\cdots \alpha_{1,n_1}),(\alpha_{2,1},\alpha_{2,2},\cdots \alpha_{2,n_2}),\cdots (\alpha_{m,1},\alpha_{m,2},\cdots \alpha_{m,n_m})\}$
where $\sum_{k=1}^m n_k=k$ and $\cup_{k=1}^m (\cup_{l=1}^{n_k} \alpha_{k,l}=[k])$ 
\end{itemize}
For more details about the symmetric group and the induced permutation, we refer the readers to Supplemental Information A of Ref.~\cite{garcia2023deep}, and the references therein.

\subsection{Unrolling QRP dynamics with tensor-diagram}\label{app:unroll}

We recall the QRP output at time step $t$ for the state $\rhot{t}$ in Eq.~\eqref{eq:app-qrp-total-state} and an observable $O$ is written in the recursive form as
\begin{align}
    \expval{O}_{t}&=\operatorname{Tr}[O \rhot{t}]\\
    &=\Tr\left[O U \left(\rhoa{t}\otimes\Tr_a(\rhot{t-1})\right) U^\dag\right]\;. \label{eq:app-qrp-output}
\end{align}

The following statement presents an alternative form of QRP output as a single iteration on a larger $(n_a + n_n)t$ composite system. A sketch of the proof of this statement is provided in Fig.~\ref{fig:ProofSketch}.

\begin{supplemental_proposition}\label{sup-prop:qrp-output} Consider a QRP model with the reservoir $U$, the initial reservoir state $\rhoh{0}$ which recursively processes the input states $\vrhoa{t} = (\rhoa{1},..., \rhoa{t})$. The QRP output at the time step $t$ in Eq.~\eqref{eq:app-qrp-output} can be expressed as
\begin{align}
    \expval{O}_{t}
    &=\sum_{\Xv_{h_{t-1}}}\Tr\left[\left(X_{h_1},\IC_{a};X_{h_2},\IC_{a};\cdots;X_{h_{t-1}},\IC_{a};O\right)U^{\otimes t}\left(\rhoh{0},\rhoa{1};X_{h_1},\rhoa{2};X_{h_2},\rhoa{3};\cdots;X_{h_{t-1}},\rhoa{t}\right)(U^\dag)^{\otimes t}\right]\;,\label{Eq:OutputExpression}
\end{align}
where $\Xv_{h_{t-1}}=(X_{h_1},\cdots X_{h_{t-1}})$ and $X_{h_\t} \in \BC_h$ with $\BC_{h}$ being the orthonormal basis over the operator space of the subsystem $h_{\t}$ with dimensions $d_h$.
Here, ``$,$'' denotes the tensor products of operators in accessible and hidden space, and ``$;$'' denotes the tensor products of composite operators in different copies of the composite space. In addition, $\IC_{a}$ is an identity matrix on the accessible subsystem with dimensions $d_a$.
\end{supplemental_proposition}

\begin{figure}[ht]
    \centering
    
    \includegraphics[width=0.95\textwidth]{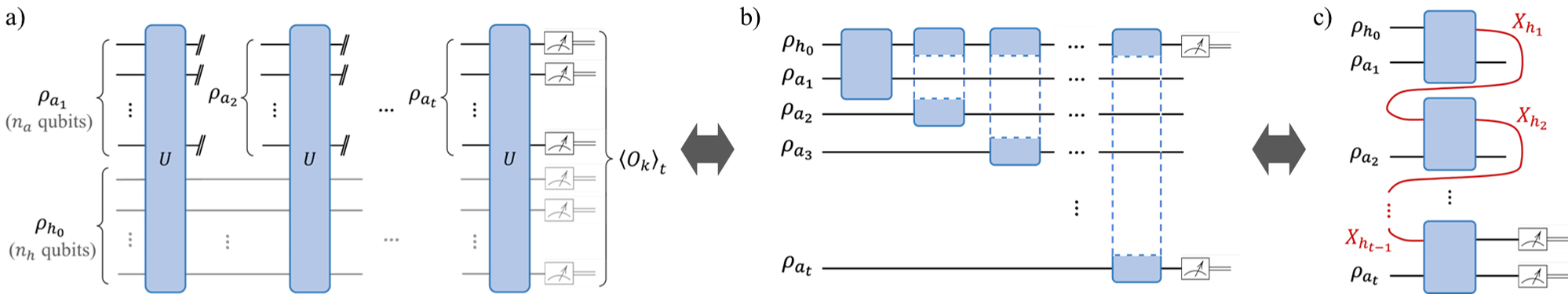}
    
    \caption{\textbf{Sketch of proof for unrolling the multiple step QRP output.} First, from panel a) to b), we simplify the recursive expression of expectation values in Eq.~\eqref{obs_def_QRP} and write it in the effectively extended space of $(t\cdot n_a+n_h)$ qubits. Second, from panel b) to c), we further extend the system to $t\cdot(n_a+n_h)$ qubits, such that the integral for QRP with arbitrary $t$-steps can be transformed to a standard Haar integral of $2t$'th moment. Finally, we calculate the leading term of observables' variance and find its lower bound using Lemma 3 of Ref.~\cite{garcia2023deep}.
    }
    \label{fig:ProofSketch}
\end{figure}

\begin{proof}
We consider the QRP output at time step $t$ in Eq.~\eqref{eq:app-qrp-output} and recursively unroll the early time step QRP states with Eq.~\eqref{eq:app-qrp-total-state}. Our first crucial remark is that throughout the proof we explicitly write the subscripts for all the reservoirs to keep track of which subsystems they act on. For example, at the last time step $t$, we have $U_{a_th}$ which indicates that the reservoir $U$ acts on the accessible subsystem ``$a_t$'' (corresponding to the one being injected by the input $\rhoa{t}$) and the hidden subsystem ``$h$''. Note that the observable $O$ acts on the subsystems ``$a_t$'' and ``$h$''. 

Now, our strategy is to absorb the partial trace with respect to the hidden space into a trace of larger dimension by using the ``trace-partial-trace-identity'', i.e. $\Tr_{AB}[(X_A\otimes\IC_B) Y_{AB}]=\Tr_{A}[X_A \Tr_{B}(Y_{AB})]$ for any operators $X_A$ in subspace $A$, $Y_{AB}$ in composite space $AB$ and $\IC_B$ is an identity on the subspace $B$. 
Let us exemplify the above process by unrolling the last $3$ consecutive time steps, namely $t, t-1$ and $t-2$.  
\begin{align}
    \expval{O}_{t}&= \Tr_{ha_t}\left[O U_{a_th} \left(\rhoa{t}\otimes\Tr_{a_{t-1}}(\rhot{t-1})\right) U_{a_th}^\dag\right]\\
    & = 
    \Tr_{ha_t a_{t-1}}\left[\left(O\otimes\IC_{a_{t-1}}\right) \left(U_{a_th}\otimes\IC_{a_{t-1}}\right) \left(\rhoa{t}\otimes\rhot{t-1}\right) \left(U_{a_th}^\dag\otimes\IC_{a_{t-1}}\right)\right]\\
    &=\Tr_{ha_t a_{t-1}}\left[\left(O\otimes\IC_{a_{t-1}}\right) \left(U_{a_th}\otimes\IC_{a_{t-1}}\right) \left(\rhoa{t}\otimes\left[U_{a_{t-1}h} \left(\rhoa{t-1}\otimes\Tr_{a_{t-2}}[\rhot{t-1}]\right) U_{a_{t-1}h}^\dag\right]\right) \left(U_{a_th}^\dag\otimes\IC_{a_{t-1}}\right)\right]\\
    &=\Tr_{ha_t\cdots a_{t-2}}\Big[\left(O\otimes\IC_{a_{t-1},a_{t-2}}\right) \left(U_{a_th}\otimes\IC_{a_{t-1},a_{t-2}}\right) \left(\rhoa{t}\otimes \IC_{h, a_{t-1},a_{t-2}}\right) \left(U_{a_{t-1}h}\otimes\IC_{a_t,a_{t-2}}\right)\nonumber \\
    &\qquad\qquad \left(\rhoa{t-1}\otimes\rhot{t-2}\otimes\IC_{a_{t}}\right) \left(U_{a_{t-1}h}^\dag\otimes\IC_{a_t,a_{t-2}}\right) \left(U_{a_th}^\dag\otimes\IC_{a_{t-1},a_{t-2}}\right)\Big]\\
    &=\Tr_{ha_t\cdots a_{t-2}}\Big[\left(O\otimes\IC_{a_{t-1},a_{t-2}}\right) \left(U_{a_th}\otimes\IC_{a_{t-1},a_{t-2}}\right)\left(U_{a_{t-1}h}\otimes\IC_{a_t,a_{t-2}}\right)\left(\rhoa{t}\otimes\rhoa{t-1}\otimes\rhot{t-2}\right)\nonumber\\
    &\qquad\qquad \left(U_{a_{t-1}h}^\dag\otimes\IC_{a_t,a_{t-2}}\right) \left(U_{a_th}^\dag\otimes\IC_{a_{t-1},a_{t-2}}\right)\Big] \;,
\end{align}
where to reach the last equality we commute $\left(\rhoa{t}\otimes \IC_{h, a_{t-1},a_{t-2}}\right)$ and $\left(U_{a_{t-1}h}\otimes\IC_{a_t,a_{t-2}}\right)$, as their non-identity parts are in different subspaces, and then do $\left(\rhoa{t}\otimes \IC_{h, a_{t-1},a_{t-2}}\right)\left(\rhoa{t-1}\otimes\rhot{t-2}\otimes\IC_{a_{t}}\right) =\left(\rhoa{t}\otimes\rhoa{t-1}\otimes\rhot{t-2}\right)$. 
We remark that in the equalities above, the subscripts of traces denote the subspaces that the trace operator is acting upon.

By repeating the above procedure until the first time step, we obtain
\begin{equation}
    \expval{O}_{t}=\Tr_{h,a_t,\cdots a_1}\left[\left(O\otimes\IC_{a_{t-1}\cdots a_1}\right) U_t\cdots \,U_{1} \left(\rhoh{0}\otimes\rhoa{t}\otimes\cdots\rhoa{1}\right) U_1^\dag \cdots \,U_t^\dag\right]\;,
\end{equation}
where, for all $\t=1,\cdots, t$, we denote $U_{\tau} = U_{a_{\tau}h}\otimes\IC_{\{a_t,\cdots a_1\}\backslash a_\t}$ which indicates that $U_{\tau}$ acts non-trivially on the hidden space and the accessible space of the $\t$-th iteration, and acts trivially on the other subsystems. 

Our next step is to further extend the space. To do so, it is convenient to work in the vectorized form. Using the identity $\Tr\left[OU\rho U^\dag \right]=\llangle O | U\otimes U^* | \rho \rrangle$ where $| \rho \rrangle$ and $| O \rrangle$ are vectorized forms of $\rho$ and $O$ respectively, we obtain
\begin{align}
    \expval{O}_{t}&=\llangle O\otimes\IC_{a_{t-1}\cdots a_1} |\left(U_{t}\cdots \, U_{1}\right) \otimes \left(U_{t}^*\cdots \, U_{1}^*\right)|\rhoh{0}\otimes\rhoa{t}\otimes\cdots\rhoa{1}\rrangle\\ 
    &=\llangle O\otimes\IC_{a_{t-1}\cdots a_1}| W_{a_th}\cdots W_{a_1 h} | \rhoh{0}\otimes\rhoa{t}\otimes\cdots\rhoa{1}\rrangle \;,\label{C41}
\end{align}
where $W_{a_\t h}:=U_{a_\t h}\otimes U_{a_\t h}^*\otimes \IC_{\{a_t,\cdots a_1\}\backslash a_\t}^{\otimes2}$.  
We note that here the order of spaces and subspaces are interchangeable. 

To further proceed, it is crucial to notice that for distinct $\t$ and $\t'$, $U_{a_\t h}$ and $U_{a_{\t'} h}$ (as well as $W_{a_\t h}$ and $W_{a_{\t'} h}$) only overlap in the hidden subsystem but never in the accessible subsystems, as illustrated in Fig.~\ref{fig:ProofSketch}\,b). Inspired by this, let us first consider the simplified version with only three subsystems $A, B$ and $C$ and the following quantity $\langle \widetilde{O} \rangle_{ABC}$ of the form
\begin{align}
    \langle \widetilde{O} \rangle_{ABC} = \llangle O_{AB} \otimes O_C  | (W_{AB} \otimes \IC_C^{\otimes2}) ( W_{AC} \otimes \IC_B^{\ot 2})|\rho_A \otimes \rho_B \otimes\rho_C\rrangle  \;,
\end{align}
where $W_{AB}$ and $W_{AC}$ are some operators acting on the subsystems $AB$ and $AC$ respectively. Crucially, they both act on subsystem $A$. In addition, $O_{AB}$ is an arbitrary observable on the subsystems $AB$ while $O_C$ is an observable on the subsystem $C$. Similarly, $\rho_A, \rho_B, \rho_C$ are some arbitrary states on the subsystems $A, B, C$ respectively. 

By using the Hilbert-Schmidt resolution of identities on the common subsystem, the following computation steps hold (or, alternatively, are illustrated in Fig.~\ref{fig:qrp-output-td1} with tensor diagrams): 
\begin{align}
    \langle \widetilde{O} \rangle_{ABC} 
 &= \left(\llangle O_{AB} |W_{AB} \otimes\IC_{C}^{\otimes 2}\right)\left(\IC_A^{\otimes 2}\otimes|\rho_{B}\rrangle\otimes\llangle O_C |\right)\left(W_{AC} \otimes \IC_B^{\ot 2}|\rho_A\otimes \rho_C \rrangle\right)\\
    &= \left(\llangle O_{AB} |W_{AB} \otimes\IC_{C}^{\otimes 2}\right)\left[ \sum_{X \in B_A}|X \rrangle \llangle X| \right] \otimes | \rho_{B}\rrangle\otimes \llangle O_C |\left(\IC_{A}^{\otimes 2}\otimes W_{CB}|\rho_A\otimes \rho_C \rrangle\right)\\
    =&\sum_{X \in B_A } \llangle O_{AB} | W_{AB} |X \otimes \rho_{B} \rrangle \llangle X \otimes O_C| W_{AC} |\rho_A\otimes \rho_C\rrangle\label{eq:proof-appx-expand-systems}
\end{align}
where the sum is taken over $X \in B_A$ where $B_A$ is a complete orthonormal basis in the operator space of the subsystem $A$ under the Hilbert-Schmidt norm. For instance, we could choose a set of singleton matrices $B_A=\{\ketbra{i}{j}:i,j=0,'\cdots d_A-1\}$ such that $\{|X \rrangle : X\in B_A\}$ forms the computational basis; alternatively we could also use Pauli strings, if space $A$ is qubit-defining, such that their vectorizations are Bell states.
\begin{figure}[ht]
    \centering
    
    \includegraphics[width=0.80\textwidth]{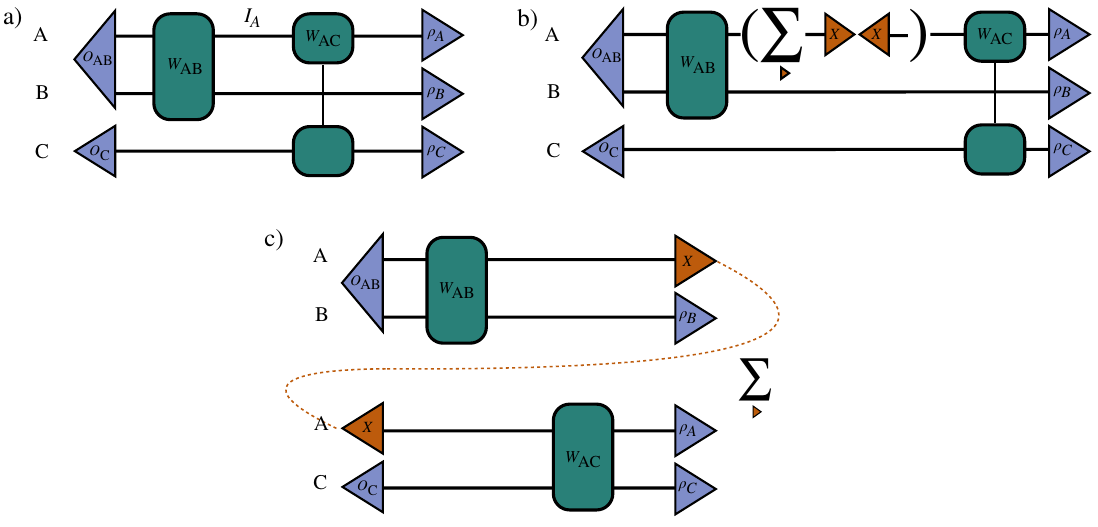}
    
    \caption{\textbf{Tensor diagram for proving Eq~\eqref{eq:proof-appx-expand-systems}}. Unrolling multiple step QRP output relies on recursively applying this tensor diagram for different time steps (with appropriately assigned subsystems $A$, $B$ and $C$). 
    }
    \label{fig:qrp-output-td1}
\end{figure}

We are now ready to tackle the QRP output $\langle O \rangle_t$ in Eq.~\eqref{C41} by recursively applying Eq.~\eqref{eq:proof-appx-expand-systems}. This is done by systematically identifying what the subsystems $A, B$ and $C$ are for each iteration.

\medskip

\noindent\underline{\textit{For the time step $t$.}} We choose $A$ in Eq.~\eqref{eq:proof-appx-expand-systems} to be the hidden subsystem $h_0$, choose $B$ to be the subsystem $a_t$ and choose $C$ to be the rest i.e., $a_{t-1}, \dots, a_1$. 

With this selection, we have $O_{AB} = O$, $O_C = \IC_{a_{t-1},\cdots a_1 }$, as well as $\rho_A = \rhoh{0}$, $\rho_B = \rhoa{t}$, $\rho_C = \rhoa{t-1}\ot \cdots \ot \rhoa{1}$, leading to
\begin{align}
     \expval{O}_{t}&= \sum_{X_{h_{t-1}} \in B_h} \llangle O | W_{a_th}|X_{h_{t-1}} \otimes \rhoa{t} \rrangle \llangle X_{h_{t-1}} \ot \IC_{a_{t-1},\cdots a_1 }| W_{a_{t-1}h} \cdots W_{a_{1}h} | \rhoh{0} \ot \rhoa{t-1}\cdots\rhoa{1} \rrangle  \;, \label{eq:ext-h-t}
\end{align}
where we introduce the subscript $h_t$ to indicate that $X_{h_{t-1}}$ acts on the expanded hidden subspace.

\medskip

\noindent\underline{\textit{For the time step $t-1$.}}  We only apply Eq.~\eqref{eq:proof-appx-expand-systems} to the term $\llangle X_{h_{t-1}} \ot \IC_{a_{t-1},\cdots a_1 }| W_{a_{t-1}h} \cdots W_{a_{1}h} | \rhoh{0} \ot \rhoa{t-1}\cdots\rhoa{1} \rrangle$ in Eq.~\eqref{eq:ext-h-t}. Here, we choose $A$ to be the hidden subsystem $h_0$, choose $B$ to be $\rhoa{t-1}$ and choose $C$ to be the rest i.e, $a_{t-2}, \dots, a_{1}$ (excluding $a_t$).

With this selection, we have $O_{AB} = X_{h_{t-1}} \ot \IC_{a_{t-1}}$, $O_C = \IC_{a_{t-2}\cdots a_1}$, and $\rho_A = \rhoh{0}$, $\rho_B = \rhoa{t-1}$, $\rho_C = \rhoa{t-2}\ot\cdots\ot\rhoa{1}$, leading to
\begin{align}
     \expval{O}_{t}&= \sum_{X_{h_{t-1}}, X_{h_{t-2}}} \llangle O | W_{a_th}|X_{h_{t-1}} \otimes \rhoa{t} \rrangle \llangle X_{h_{t-1}} \ot \IC_{a_{t-1}} | W_{a_{t-1}h}|X_{h_{t-2}} \otimes \rhoa{t-1} \rrangle 
     \llangle X_{h_{t-2}} \ot \IC_{a_{t-2},\cdots a_1 }| W_{a_{t-2}h} \cdots W_{a_{1}h} | \rhoh{0} \ot \rhoa{t-2}\cdots\rhoa{1} \rrangle  \\
     & = \sum_{X_{h_{t-1}}, X_{h_{t-2}}} \llangle O \ot X_{h_{t-1}} \ot \IC_{a_{t-1}}| W_{a_th} \ot W_{a_{t-1}h} |X_{h_{t-1}} \otimes \rhoa{t} \ot X_{h_{t-2}} \otimes \rhoa{t-1} \rrangle \llangle X_{h_{t-2}} \ot \IC_{a_{t-2},\cdots a_1 }| W_{a_{t-2}h} \cdots W_{a_{1}h} | \rhoh{0} \ot \rhoa{t-2}\cdots\rhoa{1} \rrangle \;,
\end{align}
where in the second equality we use $\llangle A | B |C \rrangle \llangle A' | B' |C' \rrangle  = \llangle A \ot A' | B\ot B |C \ot C' \rrangle $. 

By recursively doing this procedure, we eventually obtain
\begin{align}
     \expval{O}_{t} & = \sum_{\vec{X}_{h_{t-1}}} \llangle O ; X_{h_{t-1}}, \IC_{a_{t-1}}; \, \cdots ; X_{h_1}, \IC_{a_1} | W_{a_th} ;  W_{a_{t-1}h};\cdots ;  W_{a_1h} |X_{h_{t-1}}, \rhoa{t} ; \cdots ; X_{h_1} , \rhoa{2} ; \rhoh{0} , \rhoa{1} \rrangle \\
     &= \sum_{\Xv_{h_{t-1}}}\Tr\left[\left(X_{h_1},\IC_{a_1};X_{h_2},\IC_{a_2};\cdots;X_{h_{t-1}},\IC_{a_{t-1}};O\right)U^{\otimes t}\left(\rhoh{0},\rhoa{1};X_{h_1},\rhoa{2};X_{h_2},\rhoa{3};\cdots;X_{h_{t-1}},\rhoa{t}\right)(U^\dag)^{\otimes t}\right] \;,
\end{align}
where $\vec{X}_{h_{t-1}} = (X_{h_{t-1}}, \cdots,X_{h_1})$ with each $X_k$ being sum over all the basis. Here in the above equations, ``$,$'' denotes the tensor products of operators in accessible and hidden space, and ``$;$'' denotes the tensor products of composite operators in different copies of the composite space. In the last equality, we undo the vectorization together with explicitly writing the reservoir $U$ (instead of $W_{a_\t h}$).
\end{proof}

\subsection{Proof of Theorem~\ref{Thm:HaaResVar}: Exponential concentration}\label{app:conc_output}
In this section, we provide a formal version of Theorem~\ref{Thm:HaaResVar} addressing the concentration of QRP's output and its detailed proof. We note that the formal version explicitly includes the temporal dependent term which vanishes exponentially in the reservoir iterations.
\begin{customthm}{1}[Variance of QRP output with extreme scrambling reservoir, formal]\label{Thm:HaaResVar_app}
Given that a reservoir is sampled from a $2t$-design on $\Ubb(d_a d_h)$, the variance of QRP output with some observable $O$ at time step $t$ vanishes exponentially with the number of accessible $n_a$ and hidden $n_h$ qubits. The scaling depends on neither the input states nor the time step $t$. Formally, at a large time step $t$, the variance is given by
\begin{equation}
\Var_{U\sim\Ubb(d_a d_h)}[\expval{O}_t]=\frac{\Tr[O^2]\left(1+ \beta\right)}{ d_a^2 d_h^3} +  \Delta_{\rm temp}(t)\;\;,
\end{equation}
with $\beta \in \OC(\frac{1}{d_a d_h})$ and the temporal dependent term is
\begin{align}
    \Delta_{\rm temp}(t):=\frac{1}{d_a^{t+1}d_h^2}\Tr[O^2] \;.
\end{align}
For large $t$, since $\Delta_{\rm temp}(t)$ vanishes exponentially in $t$, the non-temporal dependent term dominates and the QRP output exponentially concentrates as
\begin{equation}\label{eq:qrp-scrambling-concentration}
\operatorname{Pr}_{U\sim\Ubb(d_a d_h)}\left(\left|\expval{O}_t-\Ebb_{U}[\expval{O}_t]\right|\geq\delta\right)\leq \frac{\Tr[O^2]}{\delta^2 d_a^2 d_h^3}\,,
\end{equation}
where the concentration point is independent of the reservoir and the input states
\begin{equation}
    \Ebb_{\Ubb(d_a d_h)}[\expval{O}]=\frac{\Tr[O]}{d_a d_h}\left(1+ \beta' \right)\;,\label{Eq:FirstMoment-thm} 
\end{equation}
with $\beta' \in \OC\left(\frac{1}{d_a d_h}\right)$.
\end{customthm}

\begin{specialproof}
By Supplemental Proposition~\ref{sup-prop:qrp-output}, we obtain the standard form of $t$-th moment for QRP's output, which enables Haar integral. Our proof strategy is by applying Lemma~\ref{lemma:L3_deep_QNN} to find the leading term of the variance of QRP output. 

We first recall the expression of reservoir output from Supplemental Proposition~\ref{sup-prop:qrp-output}
\begin{align}
    \expval{O}_{t}&=\sum_{\Xv_{h_{t-1}}}\Tr\left[\left(X_{h_1},\IC_{a};X_{h_2},\IC_{a};\cdots;X_{h_{t-1}},\IC_{a};O\right)U^{\otimes t}\left(\rhoh{0},\rhoa{1};X_{h_1},\rhoa{2};X_{h_2},\rhoa{3};\cdots;X_{h_{t-1}},\rhoa{t}\right)(U^\dag)^{\otimes t}\right]\,.
\end{align}
Here, ``$,$'' denotes the tensor products of operators in accessible and hidden space, and ``$;$'' denotes the tensor products of composite operators in different copies of the composite space.
Then, the square of the quantity is given by
\begin{align}
    \expval{O}^2_{t}=&\sum_{\Xv_{h_{t-1}}, \Yv_{h_{t-1}}}\Tr\left[\left(X_{h_1},\IC_{a};X_{h_2},\IC_{a};\cdots;X_{h_{t-1}},\IC_{a};O\right)U^{\otimes t}\left(\rhoh{0},\rhoa{1};X_{h_1},\rhoa{2};X_{h_2},\rhoa{3};\cdots;X_{h_{t-1}},\rhoa{t}\right)(U^\dag)^{\otimes t}\right]\nonumber\\
    &\quad\quad\cdot \Tr\left[\left(Y_{h_1},\IC_{a};Y_{h_2},\IC_{a};\cdots;Y_{h_{t-1}},\IC_{a};O\right)U^{\otimes t}\left(\rhoh{0},\rhoa{1};Y_{h_1},\rhoa{2};Y_{h_2},\rhoa{3};\cdots;Y_{h_{t-1}},\rhoa{t}\right)(U^\dag)^{\otimes t}\right]\\
    =&\sum_{\Xv_{h_{t-1}}, \Yv_{h_{t-1}}}\Tr\Bigr[\left(X_{h_1},\IC_{a};X_{h_2},\IC_{a};\cdots;X_{h_{t-1}},\IC_{a};O;Y_{h_1},\IC_a;Y_{h_2},\IC_a;\cdots;Y_{h_{t-1}},\IC_{a};O\right)U^{\otimes 2t}\nonumber\\
    &\quad\quad\cdot \left(\rhoh{0},\rhoa{1};X_{h_1},\rhoa{2};X_{h_2},\rhoa{3};\cdots;X_{h_{t-1}},\rhoa{t};\rhoh{0},\rhoa{1};Y_{h_1},\rhoa{2};Y_{h_2},\rhoa{3};\cdots;Y_{h_{t-1}},\rhoa{t}\right)(U^\dag)^{\otimes 2t}\Bigr]
\end{align}
To obtain the variance, we proceed by determining the leading terms of first and second moments, respectively.

\subsubsection{First moment}
We first determine the leading term of the expectation value of QRP's output:
\begin{align}
    &\mathbb{E}_{\Ubb(d_a d_h)}\left[\expval{O}_t \right]\\
    =&\sum_{\Xv_{h_{t-1}}}\int d\mu(U)\Tr\left[\left(X_{h_1},\IC_{a};X_{h_2},\IC_{a};\cdots;X_{h_{t-1}},\IC_{a};O\right)U^{\otimes t}\left(\rhoh{0},\rhoa{1};X_{h_1},\rhoa{2};X_{h_2},\rhoa{3};\cdots;X_{h_{t-1}},\rhoa{t}\right)(U^\dag)^{\otimes t}\right]\\
    =&\frac{1}{(d_a d_h)^{t}}\sum_{\Xv_{h_{t-1}}} \biggr(\sum_{\sg\in P_{d_a d_h}(\mathcal{S}_{t})}\Tr\left[\left(X_{h_1},\IC_{a};X_{h_2},\IC_{a};\cdots;X_{h_{t-1}},\IC_{a};O\right)\sg\right]\Tr\left[\left(\rhoh{0},\rhoa{1};X_{h_1},\rhoa{2};X_{h_2},\rhoa{3};\cdots;X_{h_{t-1}},\rhoa{t}\right)\sg^{-1}\right]\nonumber\\
     &\qquad+\sum_{\sg,\pi\in P_{d_a d_h}(\mathcal{S}_{t})} c_{\sg,\pi}\Tr\left[\left(X_{h_1},\IC_{a};X_{h_2},\IC_{a};\cdots;X_{h_{t-1}},\IC_{a};O\right)\sg\right]\Tr\left[\left(\rhoh{0},\rhoa{1};X_{h_1},\rhoa{2};X_{h_2},\rhoa{3};\cdots;X_{h_{t-1}},\rhoa{t}\right)\pi\right]\biggr)\,,
\end{align}
where the last equality is obtained using Lemma~\ref{lemma:L3_deep_QNN}, and the prefactors $c_{\sg,\pi}\in\OC\left((d_a d_h)^{-1}\right)$ which are exponentially small in number of qubits. We further note that both sums have polynomially many terms. Hence for the leading order terms we have:
\begin{align}
     &\mathbb{E}_{\Ubb(d_a d_h)}\left[\expval{O}_t\right]\nonumber\\
    =&\frac{1}{(d_a d_h)^{t}}\left(1+\OC\left((d_a d_h)^{-1}\right)\right)\sum_{\Xv_{h_{t-1}}} \sum_{\sg\in P_{d_a d_h}(\mathcal{S}_{t})}\\
    &\qquad\qquad\Tr\left[\left(X_{h_1},\IC_a;X_{h_2},\IC_a;\cdots;X_{h_{t-1}},\IC_a;O\right)\sg\right]\Tr\left[\left(\rhoh{0},\rhoa{1};X_{h_1},\rhoa{2};X_{h_2},\rhoa{3};\cdots;X_{h_{t-1}},\rhoa{t}\right)\sg^{-1}\right]\\
    =&\frac{1}{(d_a d_h)^{t}}\left(1+\OC\left((d_a d_h)^{-1}\right)\right)\sum_{\sg\in P_{d_a d_h}(\mathcal{S}_{t})} g^{(1)}_{\sg}(O)\label{A36}\,,
\end{align}
where we denote
\begin{equation}\label{A34}
    g^{(1)}_{\sg}(O):=\sum_{\Xv_{h_{t-1}}}\Tr\left[\left(X_{h_1},\IC_a;X_{h_2},\IC_a;\cdots;X_{h_{t-1}},\IC_a;O\right)\sg\right]\Tr\left[\left(\rhoh{0},\rhoa{1};X_{h_1},\rhoa{2};X_{h_2},\rhoa{3};\cdots;X_{h_{t-1}},\rhoa{t}\right)\sg^{-1}\right]\,.
\end{equation}

In what follows, our argument hinges on two fundamental observations:
\begin{enumerate}
    \item $g^{(1)}_{\sg}(O)$ with the identity operator $\sg = \bigotimes_{i=1}^t \IC_{a_i}\otimes\IC_{h_i}$ is the dominant term in Eq.~\eqref{A36}.
    \item Other $g^{(1)}_{\sg}(O)$ (with other permutation operators $\sg \neq \bigotimes_{i=1}^t \IC_{a_i}\otimes\IC_{h_i}$) are smaller than the dominant term by at least one order of $d_ad_h$.
\end{enumerate}

Proving them is indeed not straight forward. We proceed by first gaining some intuition with the toy example for $t=2$, and then consider the general setting with arbitrary $t$ time steps by performing the computation with the identity element as well as arguing that other permutation elements give smaller contribution.

\begin{figure}[ht]
    \centering
    
    \includegraphics[width=0.9\textwidth]{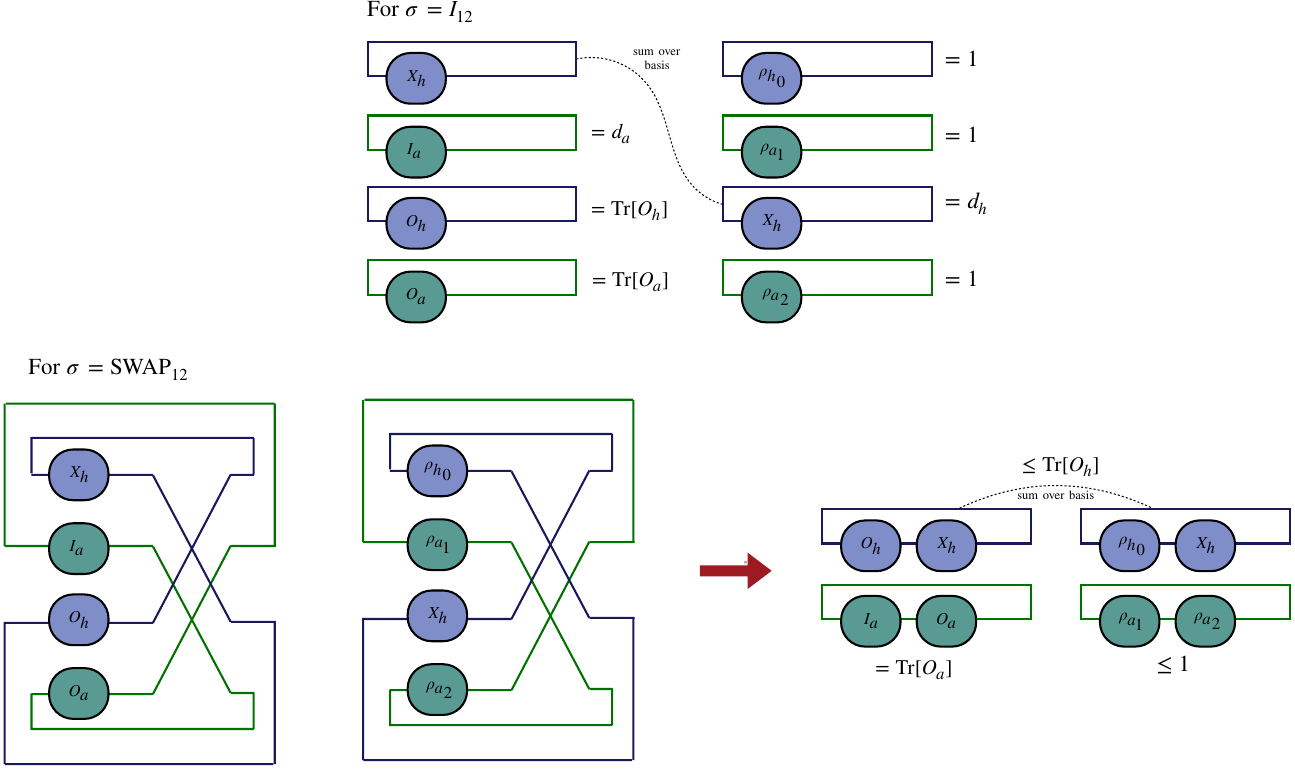}
    
    \caption{\textbf{Tensor diagrams of $g^{(1)}_{\sg}(O)$ for the toy example with $t=2$.}  
    We illustrate the tensor diagrams for $g^{(1)}_{\sg}(O)$ in Eq.~\eqref{eq:g-term-toy} associated with identity and SWAP operators. 
    The effect of the SWAP operator is such that the factor of $d_ad_h$ vanishes, and hence $g^{(1)}_{\sg}(O)$ with $\sg = \IC_{12}$ is the dominant term.
    }
    \label{fig:qrp-td2-dominant-term}
\end{figure}

\bigskip

\paragraph*{\underline{Toy example with $t=2$ and a separable observable.}} We gain some intuitive understanding regarding the dominant term with a simple example~\footnote{-- just like how other physics problems start to get solved with a simple harmonic oscillator --}. Let us consider $t =2$ and $O = O_h\ot O_a$ with some observables $O_h$ and $O_a$ acting on hidden and accessible subsystems respectively. In this setting, $g^{(1)}_{\sg}(O)$ is of the form
\begin{align}\label{eq:g-term-toy}
    g^{(1)}_{\sg}(O)=\sum_{\Xv_{h_{1}}\in\{ |i\rangle\langle j| \}_{i,j=1}^{d_h}}\Tr\left[\left(X_{h_1},\IC_a;O_h,O_a\right)\sg\right]\Tr\left[\left(\rhoh{0},\rhoa{1};X_{h_1},\rhoa{2}\right)\sg\right] \;.
\end{align}
Here, the set of permutation operators is $\sigma \in \{ \IC_{12}, {\rm SWAP}_{12}\}$ with $\IC_{12} = (\IC_{a}\otimes\IC_{h_0})^{\ot 2}$. We now compute $g^{(1)}_{\sg}(O)$ for these two operators and compare their values. Fig.~\ref{fig:qrp-td2-dominant-term} illustrates this computation with tensor network diagrams.

For $\sg = \IC_{12}$, we have
\begin{align}
    g^{(1)}_{\sg}(O)&=\sum_{\Xv_{h_{1}}\in\{ |i\rangle\langle j| \}_{i,j=1}^{d_h}}\Tr\left[\left(X_{h_1},\IC_a;O_h,O_a\right)\IC_{12}\right]\Tr\left[\left(\rhoh{0},\rhoa{1};X_{h_1},\rhoa{2}\right)\IC_{12}\right] \\
    & = \sum_{\Xv_{h_{1}}\in\{ |i\rangle\langle j| \}_{i,j=1}^{d_h}} \Tr[X_{h_1}]\Tr[\IC_a] \Tr[O_h]\Tr[O_a] \Tr[X_{h_1}] \\
    & = \sum_{i,j = 1}^{d_h} \delta^2_{ij} d_a \Tr[O] \\
    & = d_h d_a \Tr[O] \;,
\end{align}
where we use $\Tr[X_{h_1}] = \delta_{ij}$ together with $\delta_{ij}^2 = \delta_{ij}$, $\Tr[O_h]\Tr[O_a] =\Tr[O]$ and $\Tr[\rho] = 1$ for any quantum state $\rho$.

On the other hand, for $\sg= {\rm SWAP}_{12}$ we have
\begin{align}
     g^{(1)}_{\sg}(O)&=\sum_{\Xv_{h_{1}}\in\{ |i\rangle\langle j| \}_{i,j=1}^{d_h}}\Tr\left[\left(X_{h_1},\IC_a;O_h,O_a\right){\rm SWAP}_{12}\right]\Tr\left[\left(\rhoh{0},\rhoa{1};X_{h_1},\rhoa{2}\right){\rm SWAP}_{12}\right] \\
     & = \sum_{\Xv_{h_{1}}\in\{ |i\rangle\langle j| \}_{i,j=1}^{d_h}} \Tr[X_h O_h] \Tr[\IC_a O_a] \Tr\left[\rhoh{0}X_{h_1}\right]\Tr\left[\rhoa{1}\rhoa{2}\right] \\
     & = \sum_{i,j=1}^{d_h} \langle j | O_h |i\rangle \Tr[O_a] \langle j | \rhoh{0} |i\rangle \Tr[\rhoa{1}\rhoa{2}] \\
     & \leq \sum_{i,j=1}^{d_h} \langle j | O_h |i\rangle \Tr[O_a] \delta_{ij} \\
     & = \Tr[O] \;,
\end{align}
where in the second equality we have used $\Tr[(P_1 \ot P_2){\rm SWAP}_{12}] = \Tr[P_1P_2]$ with $P_1$ and $P_2$ being some operators on the first and second copies. The inequality is from $\Tr[\rhoh{0}X_{h_1}] \leq \|\rhoh{0}\|_{\infty} \|X\|_1 \leq \delta_{ij}$ (with $\| \rhoh{0}\|_{\infty} \leq 1$), and $\Tr[\rhoa{1}\rhoa{2}] \leq 1$. 

Hence, we observe that $g^{(1)}_{\sg}(O)$ with $\sg$ being the identity is the leading term and larger than the other term by the factor $d_a d_h$. This can be intuitively understood from the role of $\rm SWAP$ operator. In particular, we observe that
\begin{enumerate}
    \item The disappear of $d_a$ is due to the fact that the SWAP operation leads to $\Tr[\IC_a O_a] = \Tr[O_a]$ (instead of $ \Tr[\IC_a] \Tr[O_a] = d_a \Tr[O_a]$ with the identity).
    \item The disappear of $d_h$ is due to the fact that the SWAP operation leads to the upper bound of $\sum_{i,j} \langle j | O_h |i\rangle \delta_{ij} = \Tr[O_h]$ (instead of the double Kronecker deltas $\sum_{i,j} \Tr[O_h]\delta_{ij}^2 = d_h \Tr[O_h]$ with the identity).
\end{enumerate}
This intuitive understanding of how the SWAP operator can impact the terms does indeed carry out to more general setting. We argue this more thoroughly in the coming step.

\medskip

\paragraph*{\underline{General setting with arbitrary time steps.}}
We now compute $g^{(1)}_{\sg}(O)$ in Eq.~\eqref{A34} with the symmetric group element $s=((1)(2)\cdots(t))$ for which $\sigma=\bigotimes_{i=1}^t \IC_{a_i}\otimes\IC_{h_i}$ leading to
\begin{align}
    g^{(1)}_{\sg}(O)&=\sum_{\Xv_{h_{t-1}}}\Tr\left[\left(X_{h_1},\IC_a;X_{h_2},\IC_a;\cdots;X_{h_{t-1}},\IC_a;O\right)\bigotimes_{i=1}^t \IC_{a_i}\otimes\IC_{h_i}\right]\Tr\left[\left(\rhoh{0},\rhoa{1};X_{h_1},\rhoa{2};X_{h_2},\rhoa{3};\cdots;X_{h_{t-1}},\rhoa{t}\right)\bigotimes_{i=1}^t \IC_{a_i}\otimes\IC_{h_i}\right]\\
    &=\sum_{\Xv_{h_{t-1}}}\Tr[X_{h_1}]^2\cdots\Tr[X_{h_{t-1}}]^2\Tr[I_a]^{t-1}\Tr[O]\\
    &=d_a^{t-1}\Tr[O]\sum_{{i_1,j_1;\cdots i_{t-1},j_{t-1}}\in\{1,
    \cdots, d_h\} }\Tr[\ketbra{i_1}{j_1}]^2\cdots\Tr[\ketbra{i_{t-1}}{j_{t-1}}]^2\\
    &=d_a^{t-1}\Tr[O]\sum_{{i_1,j_1;\cdots i_{t-1},j_{t-1}}\in\{1,
    \cdots, d_h\} }\delta_{i_1,j_1}^2\cdots\delta_{i_{t-1},j_{t-1}}^2\\
    &=d_a^{t-1}\Tr[O]\sum_{{i_1;\cdots i_{t-1}}\in\{1,
    \cdots, d_h\} }1\\
    &=d_a^{t-1}d_h^{t-1}\Tr[O]\,, \label{eq:1moment-leading-term}
\end{align}
where in the second line we ignore the traces of states as they equal 1.

\medskip

We next claim that this is the leading order term. This is justified with two arguments.

\medskip
 
(i).~The scaling with respect to $d_a$ is maximized with the identity. That is, for all permutations, only $s$ consisting of identities leads the term in Eq.~\eqref{A34} to $\Tr[I_a]^{t-1}=d_a^{t-1}$. Similar to the toy example presented above, other permutations containing ``SWAP'' operations between copies of spaces ranging from $1$ to $(t-1)$ will lead to terms containing $\Tr[\IC_a \cdot \IC_a]=\Tr[\IC_a]=d_a$, while with identities permutation we will obtain $\Tr[\IC_a]\Tr[\IC_a]=\Tr[\IC_a]^2=d_a^2$. We also remark that with the SWAP operators the terms get smaller by at least a factor of $d_a$.

\medskip

(ii).~The scaling with respect to $d_h$ is maximized with the identity.
Consider the summation over all the operators $X_{h_1},\cdots X_{h_{t-1}}$, when computing product of their traces, we see that the identity permutation will maximize the number of identical delta functions, which is the situation when the number of $d_h$ factors is maximal.
To better illustrate this argument, we further provide an example:

Consider $X_1$ and $X_2$ with $X_k=\ketbra{i_k}{j_k}$ and $i_k, j_k=0,\cdots d_h-1$, then we calculate two quantities
\begin{itemize}
    \item Case 1: 
    \begin{align}
        \sum_{X_1,X_2}\Tr[X_1]\Tr[X_2]\Tr[X_1 X_2]&=\sum_{i_1,j_1,i_2,j_2}\Tr[\ketbra{i_1}{j_1}]\Tr[\ketbra{i_2}{j_2}]\Tr[\ket{i_1}\braket{j_1}{i_2}\bra{j_2}]\label{A46}\\ &=\sum_{i_1,j_1,i_2,j_2}\delta_{i_1,j_1}\delta_{i_2,j_2}\delta_{i_1,j_2}\delta_{i_2,j_1}\\
        &=\sum_{i_1}1\\
        &=d_h
    \end{align}
    \item Case 2:
    \begin{align}
        \sum_{X_1,X_2}\Tr[X_1]^2\Tr[X_2]^2&=\sum_{i_1,j_1,i_2,j_2}\Tr[\ketbra{i_1}{j_1}]^2\Tr[\ketbra{i_2}{j_2}]^2\label{A50}\\ &=\sum_{i_1,j_1,i_2,j_2}\delta_{i_1,j_1}^2\delta_{i_2,j_2}^2\\
        &=\sum_{i_1,j_1,i_2,j_2}\delta_{i_1,j_1}\delta_{i_2,j_2}\\
        &=\sum_{i_1,i_2}1\\
        &=d_h^2
    \end{align}
\end{itemize}
We see that case 2 leads to a larger scaling than case 1, since in case 2 the product traces in Eq.~\eqref{A50} eventually lead to only 2 distinct Kronecker delta functions. In contrast, in case 1 although the degree in operators are the same, the traces in Eq~\eqref{A46} induce 4 distinct delta functions and hence the magnitude of the sum is smaller.

More generally. Given a set of operators $\{X_k\}_{k=1}^K$ with $X_k=\ketbra{i_k}{j_k}$ and $i_k, j_k=0,\cdots d_h-1$, given two sets of indices $\{a_k\}_{k=1}^{p}=\{b_k\}_{k=1}^{p}\subseteq[K]$ which are identical sets but the elements could be in different order, i.e. not necessarily $a_k=b_k$, then the traces of their products are
    \begin{equation}
        \Tr[M_{a_1}\cdots M_{a_p}]=\delta_{j_{a_p},i_{a_1}}\delta_{j_{a_1},i_{a_2}}\cdots \delta_{j_{a_{p-1}},i_{a_p}}
    \end{equation}
    and
    \begin{equation}
        \Tr[M_{b_1}\cdots M_{b_p}]=\delta_{j_{b_p},i_{b_1}}\delta_{j_{b_1},i_{b_2}}\cdots \delta_{j_{b_{p-1}},i_{b_p}}\,.
    \end{equation}
    And the product of both traces is
    \begin{equation}
        \Tr[M_{a_1}\cdots M_{a_p}]\Tr[M_{b_1}\cdots M_{b_q}]=\delta_{j_{a_p},i_{a_1}}\delta_{j_{a_1},i_{a_2}}\cdots \delta_{j_{a_{p-1}},i_{a_p}}\delta_{j_{b_q},i_{a_1}}\delta_{j_{b_1},i_{b_2}}\cdots \delta_{j_{b_{q-1}},i_{b_q}}\,.
    \end{equation}
    Summing over all the $i_k$ and $j_k$, where $k\in\{a_l\}_{l=1}^{p}=\{b_l\}_{l=1}^{p}$, we obtain 
    \begin{align}
        \sum_{\substack{i_k,j_k\in\{0,\cdots d-1\} \\ \forall k\in\{a_l\}_{l=1}^{p}=\{b_l\}_{l=1}^{p}}}  \Tr[M_{a_1}\cdots M_{a_p}]\Tr[M_{b_1}\cdots M_{b_q}]&=\frac{d_h^{2p}}{d_h^{r_1}}\\
        &=\frac{d_h^{2p}}{d_h^{2p-r_2}}\\
        &=d_h^{r_2}\,,
    \end{align}
    where $r_1$ denotes the number of distinct Kronecker delta functions since $\delta_{i,j}^2=\delta_{i,j}\,\forall i,j$ and $r_2$ is the number of identical pairs $(a_{k},a_{k+1})=(b_{k},b_{k'+1})$ among all $k\in [p]$. Consequently, to maximize the sum above, one has to maximize the number of identical pairs of indices $r_2$. In addition, we note that failing to maximize the scaling leads to at least one order of $d_h$ magnitude smaller. Together with the $d_a$ case, we have that the less dominant terms are at least $d_a d_h$ smaller than the leading term.

Altogether, we insert Eq.~\eqref{eq:1moment-leading-term} into Eq.~\eqref{A36} and obtain an asymptotic form of the first moment as
\begin{equation}
    \Ebb_{\Ubb(d_a d_h)}[\expval{O}_t]=\frac{\Tr[O]}{d_a d_h}\left(1+\OC\left(\frac{1}{d_a d_h}\right)\right) \; .\label{Eq:FirstMoment}
\end{equation}

\subsubsection{Second moment}
We then similarly calculate the second moment of the QRP output at time step $t$:
\begin{align}
   &\mathbb{E}_{\Ubb(d_a d_h)}\left[\expval{O}^2\right]\nonumber\\
    =&\sum_{{\Xv_{h_{t-1}}, \Yv_{h_{t-1}} } }\int d\mu(U)\Tr\left[\left(X_{h_1},\IC_{a};X_{h_2},\IC_{a};\cdots;X_{h_{t-1}},\IC_{a};O\right)U^{\otimes t}\left(\rhoh{0},\rhoa{1};X_{h_1},\rhoa{2};X_{h_2},\rhoa{3};\cdots;X_{h_{t-1}},\rhoa{t}\right)(U^\dag)^{\otimes t}\right]\nonumber\\
    &\qquad\qquad\cdot \Tr\left[\left(Y_{h_1},\IC_{a};Y_{h_2},\IC_{a};\cdots;Y_{h_{t-1}},\IC_{a};O\right)U^{\otimes t}\left(\rhoh{0},\rhoa{1};Y_{h_1},\rhoa{2};Y_{h_2},\rhoa{3};\cdots;Y_{h_{t-1}},\rhoa{t}\right)(U^\dag)^{\otimes t}\right]\\
    =&\sum_{\Xv_{h_{t-1}}, \Yv_{h_{t-1}}}\int d\mu(U)\Tr\Bigr[\left(X_{h_1},\IC_{a};X_{h_2},\IC_{a};\cdots;X_{h_{t-1}},\IC_{a};O;Y_{h_1},\IC_{a};Y_{h_2},\IC_{a};\cdots;Y_{h_{t-1}},\IC_{a};O\right)U^{\otimes 2t}\nonumber\\
    &\quad\quad\cdot \left(\rhoh{0},\rhoa{1};X_{h_1},\rhoa{2};X_{h_2},\rhoa{3};\cdots;X_{h_{t-1}},\rhoa{t};\rhoh{0},\rhoa{1};Y_{h_1},\rhoa{2};Y_{h_2},\rhoa{3};\cdots;Y_{h_{t-1}},\rhoa{t}\right)(U^\dag)^{\otimes 2t}\Bigr]\\
    =&\frac{1}{(d_a d_h)^{2t}}\biggr(\sum_{{\Xv_{h_{t-1}}, \Yv_{h_{t-1}} } }\bigg(\sum_{\sg\in P_{d_a d_h}(\mathcal{S}_{2t})}\Tr\left[\left(X_{h_1},\IC_{a};X_{h_2},\IC_{a};\cdots;X_{h_{t-1}},\IC_{a};O;Y_{h_1},\IC_{a};Y_{h_2},\IC_{a};\cdots;Y_{h_{t-1}},\IC_{a};O\right)\sg\right]\nonumber\\
    &\qquad\qquad\cdot \Tr\left[\left(\rhoh{0},\rhoa{1};X_{h_1},\rhoa{2};X_{h_2},\rhoa{3};\cdots;X_{h_{t-1}},\rhoa{t};\rhoh{0},\rhoa{1};Y_{h_1},\rhoa{2};Y_{h_2},\rhoa{3};\cdots;Y_{h_{t-1}},\rhoa{t}\right)\sg^{-1}\right]\nonumber\\
     &\qquad+\sum_{\sg\in P_{d_a d_h}(\mathcal{S}_{2t})} c_{\sg,\pi}\Tr\left[\left(X_{h_1},\IC_{a};X_{h_2},\IC_{a};\cdots;X_{h_{t-1}},\IC_{a};O;Y_{h_1},\IC_{a};Y_{h_2},\IC_{a};\cdots;Y_{h_{t-1}},\IC_{a};O\right)\sg\right]\nonumber\\
     &\quad\quad\cdot \Tr\left[\left(\rhoh{0},\rhoa{1};X_{h_1},\rhoa{2};X_{h_2},\rhoa{3};\cdots;X_{h_{t-1}},\rhoa{t};\rhoh{0},\rhoa{1};Y_{h_1},\rhoa{2};Y_{h_2},\rhoa{3};\cdots;Y_{h_{t-1}},\rhoa{t}\right)\pi\right]\bigg)\biggr)\\
     =&\frac{1}{(d_a d_h)^{2t}}\left(\sum_{\sg\in P_{d_a d_h}(\mathcal{S}_{2t})}g^{(2)}_\sg(O)\right)\cdot\left(1+\OC\left((d_a d_h)^{-1}\right)\right)\,,\label{Eq:Summands}
\end{align}
where in the second equality we use the identity $\Tr[A\otimes B]=\Tr[A]\Tr[B]$, in the third one the Lemma~\ref{lemma:L3_deep_QNN}, and in the last one the fact that the prefactors $c_{\sg,\pi}\in\OC\left((d_a d_h)^{-1}\right)$ are exponentially small in numbers of qubits, where we also define
\begin{align}
    g^{(2)}_\sg(O):=&\sum_{\Xv_{h_{t-1}}, \Yv_{h_{t-1}} }\Tr\left[\left(X_{h_1},\IC_{a};\cdots;X_{h_{t-1}},\IC_{a};O;Y_{h_1},\IC_{a};\cdots;Y_{h_{t-1}},\IC_{a};O\right)\sg\right]\nonumber\\
    &\qquad\qquad\cdot\Tr\left[\left(\rhoh{0},\rhoa{1};X_{h_1},\rhoa{2};\cdots;X_{h_{t-1}},\rhoa{t};\rhoh{0},\rhoa{1};Y_{h_1},\rhoa{2};\cdots;Y_{h_{t-1}},\rhoa{t}\right)\sg^{-1}\right]\,.
\end{align}
Using the same arguments as before, here we identify two types of leading terms that maximize the scaling in $d_a$ and $d_h$. First, the one containing $\Tr[O^2]$ will be induced by the symmetric group element $s=((1)(2)\cdots (t-1) (t,2t) (t+1)(t+2)\cdots (2t-1))$, such that the permutation only swaps the $t$-th and $2t$-th spaces and applies identity for the rests. This contributes to the second moment by:
\begin{align}
     g^{(2)}_\sg(O)&=\sum_{{\Xv_{h_{t-1}}, \Yv_{h_{t-1}} } } \Tr[X_{h_1}]^2\cdots\Tr[X_{h_{t-2}}]^2\Tr[Y_{h_1}]^2\cdots\Tr[Y_{h_{t-2}}]^2\Tr[X_{h_{t-1}}]\Tr[Y_{h_{t-1}}]\Tr[X_{h_{t-1}}Y_{h_{t-1}}]\Tr[I_{d_a}]^{2(t-1)}\Tr[O^2]\\
    &=d_a^{2t-2}\sum_{{\Xv_{h_{t-1}}, \Yv_{h_{t-1}} } } \Tr[X_{h_1}]^2\cdots\Tr[X_{h_{t-2}}]^2\Tr[Y_{h_1}]^2\cdots\Tr[Y_{h_{t-2}}]^2\Tr[X_{h_{t-1}}]\Tr[Y_{h_{t-1}}]\Tr[X_{h_{t-1}}Y_{h_{t-1}}]\Tr[O^2]\\
    &=d_a^{2t-2}\sum_{\substack{i_1,j_1;\cdots i_{t-1},j_{t-1} \\ k_1,l_1;\cdots k_{t-1},k_{t-1} }\in\{1,
    \cdots, d_h\} }\Tr[\ketbra{i_1}{j_1}]^2\cdots\Tr[\ketbra{i_{t-2}}{j_{t-2}}]^2\nonumber\\
    &\qquad\Tr[\ketbra{k_1}{l_1}]^2\cdots\Tr[\ketbra{k_{t-2}}{l_{t-2}}]^2 \Tr[\ketbra{i_{t-1}}{j_{t-1}}]\Tr[\ketbra{k_{t-1}}{l_{t-1}}]\Tr[\ket{i_{t-1}}\braket{j_{t-1}}{k_{t-1}}\bra{l_{t-1}}]\Tr[O^2]\\
    &=d_a^{2t-2} d_h^{4(t-1)-2(t-2)-3}\Tr[O^2]\\
    &=d_a^{2t-2} d_h^{2t-3}\Tr[O^2]\,,\label{eq:TrO2}
\end{align}
where in the first line we drop the traces of states which equal 1. 

Another type of potential leading term contains $\Tr[O]^2$ and is given by the symmetric group element $s=((1)(2)\cdots (2t))$, such that the permutation applies identity $\IC_a\otimes\IC_h$ in all spaces. This leads to:
\begin{align}
    g^{(2)}_\sg(O)&=\sum_{{\Xv_{h_{t-1}}, \Yv_{h_{t-1}} } } \Tr[X_{h_1}]^2\cdots\Tr[X_{h_{t-2}}]^2\Tr[Y_{h_1}]^2\cdots\Tr[Y_{h_{t-2}}]^2\Tr[X_{h_{t-1}}]^2\Tr[Y_{h_{t-1}}]^2\Tr[I_{d_a}]^{2(t-1)}\Tr[O]^2\\
    &=d_a^{2t-2}\sum_{{\Xv_{h_{t-1}}, \Yv_{h_{t-1}} } } \Tr[X_{h_1}]^2\cdots\Tr[X_{h_{t-2}}]^2\Tr[Y_{h_1}]^2\cdots\Tr[Y_{h_{t-2}}]^2\Tr[X_{h_{t-1}}]^2\Tr[Y_{h_{t-1}}]^2\Tr[O]^2\\
    &=d_a^{2t-2}\sum_{\substack{i_1,j_1;\cdots i_{t-1},j_{t-1} \\ k_1,l_1;\cdots k_{t-1},k_{t-1} }\in\{1,
    \cdots, d_h\} }\nonumber\\
    &\qquad\Tr[\ketbra{i_1}{j_1}]^2\cdots\Tr[\ketbra{i_{t-2}}{j_{t-2}}]^2\Tr[\ketbra{k_1}{l_1}]^2\cdots\Tr[\ketbra{k_{t-2}}{l_{t-2}}]^2 \Tr[\ketbra{i_{t-1}}{j_{t-1}}]^2\Tr[\ketbra{k_{t-1}}{l_{t-1}}]^2\Tr[O]^2\\
    &=d_a^{2t-2} d_h^{4(t-1)-2(t-1)}\Tr[O]^2\\
    &=d_a^{2t-2} d_h^{2t-2}\Tr[O]^2\,.\label{eq:TrO}
\end{align}
Insert Eq.~\eqref{eq:TrO} and \eqref{eq:TrO2} into Eq.~\eqref{Eq:Summands}, we obtain
\begin{equation}
    \mathbb{E}_{\Ubb(d_a d_h)}\left[\expval{O}_t^2\right]=\left(\frac{\Tr[O^2]}{d_a^2 d_h^3}+\frac{\Tr[O]^2}{d_a^2 d_h^2}\right)\left(1+\OC\left(\frac{1}{d_a d_h}\right)\right)\,.
\end{equation}
We now subtract the square of first moment as in Eq.~\eqref{Eq:FirstMoment} and obtain output's variance over reservoirs sampled from Haar measure on $\Ubb(d_a d_h)$:
\begin{equation}
    \Var_{\Ubb(d_a d_h)}\left[\expval{O}^2_t\right]=\frac{\Tr[O^2]}{d_a^2 d_h^3}\left(1+\OC\left(\frac{1}{d_a d_h}\right)\right) \;.
\end{equation}

Finally, by applying the Chebyshev's tail bound with the average and the variance, we obtain the statement of probabilistic exponential concentration as claimed in Eq.~\eqref{eq:qrp-scrambling-concentration}. 
\subsubsection{Time-dependent terms}\label{app:var_small_t}
Aforementioned calculations focus on variance terms that do not vanish in the number of iterations $t$. And hence those are the asymptotic values of variance for large $t$. Nevertheless, for small $t$, the real variance may deviate from those values, as demonstrated by the numerical simulations in Fig.~\ref{fig:conc_U}. Here, we discuss the non-negligible contributions in the regime of small $t$, which analytically explains the saturation behaviours observed in Fig.~\ref{fig:conc_U}. We recall the variance of output given by Eq.~\eqref{Eq:Summands}:
\begin{align}
   \mathbb{E}_{\Ubb(d_a d_h)}\left[\expval{O}^2\right]
     =&\frac{1}{(d_a d_h)^{2t}}\left(\sum_{\sg\in P_{d_a d_h}(\mathcal{S}_{2t})}g^{(2)}_\sg(O)\right)\cdot\left(1+\OC\left((d_a d_h)^{-1}\right)\right)\,,\label{A88}
\end{align}
where
\begin{align}
    g^{(2)}_\sg(O):=&\sum_{\Xv_{h_{t-1}}, \Yv_{h_{t-1}} }\Tr\left[\left(X_{h_1},\IC_{a};\cdots;X_{h_{t-1}},\IC_{a};O;Y_{h_1},\IC_{a};\cdots;Y_{h_{t-1}},\IC_{a};O\right)\sg\right]\nonumber\\
    &\qquad\qquad\cdot\Tr\left[\left(\rhoh{0},\rhoa{1};X_{h_1},\rhoa{2};\cdots;X_{h_{t-1}},\rhoa{t};\rhoh{0},\rhoa{1};Y_{h_1},\rhoa{2};\cdots;Y_{h_{t-1}},\rhoa{t}\right)\sg^{-1}\right]\,.
\end{align}
\begin{figure}[ht]
    \centering
    \includegraphics[width=0.85\textwidth]{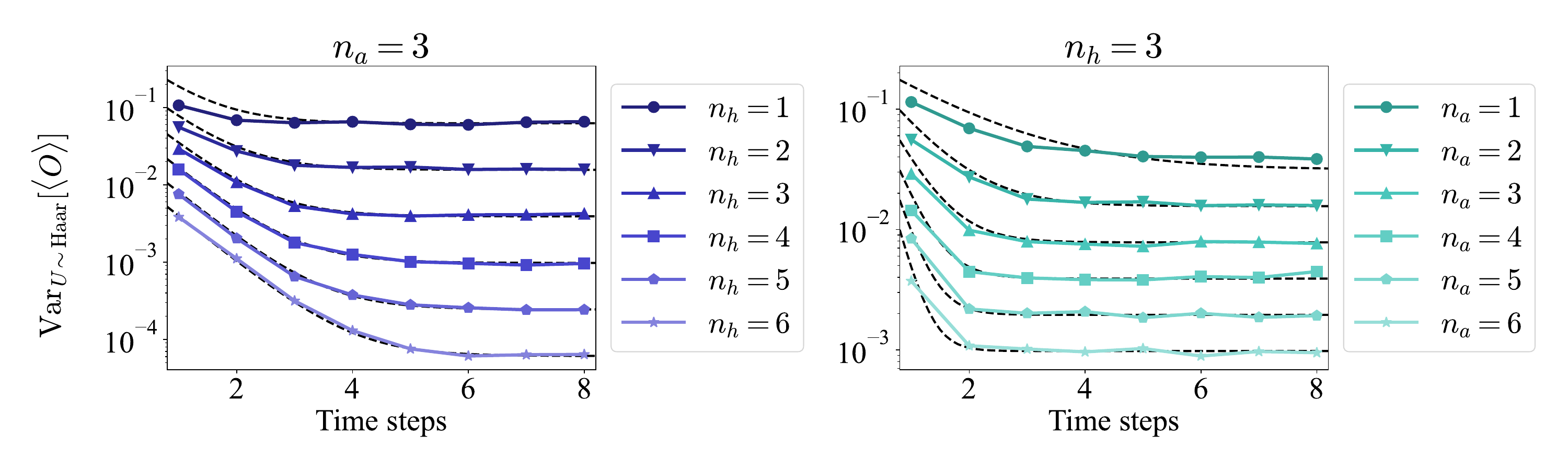}
    \caption{\textbf{Output variance of QRP compared with theoretical results considering time-dependent term.} Output variances of QRP with different qubit numbers are plotted and compared with analytical results (black dashed curves) including the correction term for small $t$ given in Eq.~\eqref{eq:var_correction}.
    }
    \label{fig:conc_u_t}
\end{figure}

Following previous arguments, here the leading terms for smalls should still contains $\Tr[O^2]$. We consider instead the permutation $\sigma=\bigotimes_{\t=1}^t \rm{SWAP}_{\t,\t+t}$, which maximizes the scaling with respect to $d_h$ but leads to a term that vanishes exponentially in $t$ with $1/d_a$ as base. More specifically, this term is:
\begin{align}
     g^{(2)}_\sg(O)&=\sum_{{\Xv_{h_{t-1}}, \Yv_{h_{t-1}} } } \Tr[X_{h_1} Y_{h_1}]\cdots\Tr[X_{h_{t-1}}Y_{t-1}]\Tr[I_{d_a}]^{t-1}\Tr[O^2]\\
    &=d_a^{t-1}\Tr[O^2]\left(\sum_{{X_{h_{1}}, Y_{h_{1}} } } \Tr[X_{h_1}Y_{h_1}]^2\right)^{t-1}\\
    &=d_a^{t-1}\Tr[O^2]\left(\sum_{i,j,k,l\in\{1,\cdots d_h\} } \Tr[\ket{i}\braket{j}{k}\bra{l}]^2\right)^{t-1}\\
    &=d_a^{t-1}\Tr[O^2]\left(\sum_{i,j,k,l\in\{1,\cdots d_h\} } \delta_{i,l}\delta_{j,k}\right)^{t-1}\\
    &=d_a^{t-1}d_h^{2(t-1)}\Tr[O^2]\,,
\end{align}
where in the first line we again ignore all the traces of density matrices and their squares that equal $1$. Combining with Eq.~\eqref{A88}, this leads to a temporally vanishing contribution of the variance:
\begin{equation}
    \Delta_{\rm temp}(t):=\frac{1}{d_a^{t+1}d_h^2}\Tr[O^2]\,.\label{eq:var_correction}
\end{equation}
Moreover, all other temporally vanishing terms are dominated by $\Delta_{\rm temp}(t)$ with a relatively suppressed scaling of $\OC(1/d_h)$. This is because among all the terms containing $\Tr[O^2]$, the one resulted by $\sigma=\bigotimes_{\t=1}^t \rm{SWAP}_{\t,\t+t}$ contributes the maximal scaling in $d_h$, as it leads to the most duplicated Kronecker delta functions when summing over $X_h$'s and $Y_h$'s. 

We further numerically verify this argument. As illustrated in Fig.~\ref{fig:conc_u_t}, we observe that, in particular for large $n_h$, the analytical values are well aligned with simulation results, not only in asymptotic regime with large $t$, but also for small $t$ where the saturation is not achieved yet.

In addition, compared to the asymptotic value of $\frac{\Tr[O^2]}{d_a^2 d_h^3}$ as given in Theorem~\ref{Thm:HaaResVar}, we find the ratio of $\Delta_{\rm temp}(t)=\frac{1}{d_a^{t+1}d_h^2}\Tr[O^2]$ to the asymptotic variance to be $\frac{d_h}{d_a^{t-1}}$. That is, the smaller this ratio, the closer is the variance to the asymptotic, saturated variance. Consequently, the approach to variance saturation presented in Fig.~\ref{fig:conc_U} has been analytically explained --the more the hidden qubits, the slower the saturation; the more the accessible qubits, the faster the saturation.
\end{specialproof}

\subsection{Proof of Theorem~\ref{thm_main:mem_ind_upp_bound}: Exponential memory decay}\label{app:mem_conc}
By using the QRP output form presented in Supplemental Proposition~\ref{sup-prop:qrp-output}, we prove Theorem~\ref{thm_main:mem_ind_upp_bound} in the main text, which shows the the exponential tight upper bound of memory indicators for an arbitrary ensemble of early state (see Definition~\ref{def:init_state_dep} and~\ref{def:input_dep} of the memory indicators). However, before going to the proof of the theorem, we present the following simplified result where the input states are also extremely scrambled. This provides a guideline and exemplify some key steps in more simple setting for the proof of Theorem~\ref{thm:mem_ind_upp_bound}.

\begin{supplemental_proposition}[Decay of memory indicator with scrambling initial states]  For extreme scrambling reservoirs and also an ensemble of extreme scrambling inputs $\SC_a=\{V\ketbra{0}{0}V^\dag \}$ with $V$ drawn from the scrambling ensemble, we have \label{prop:mem_ind_upp_bound_Haar}
\be
   \Ebb_{U\sim\Ubb(d_a d_h)} \left[\MC_a(t;\SC_a,U)\right]=\frac{1}{d_h d_a^t}\left(1+\OC\left(\frac{1}{d_h}\right)\right)\,,
\ee
where $\MC_a(t;\SC_a,U):=\Var_{\rhoa{\t}\sim \SC_a}\left[\expval{O}_{(t+\t)}\right]$. Similarly, for extreme scrambling initial reservoir  states $\SC_h=\{V\ketbra{0}{0}V^\dag : V\in\nu(\Ubb(d_h))\}$, we have 
\be
    \Ebb_{U\sim\Ubb(d_a d_h)} \left[\MC_h(t;\SC_h,U)\right]=\frac{1}{d_h d_a^t}\left(1+\OC\left(\frac{1}{d_h}\right)\right)\,,
\ee
where $\MC_h(t;\SC_h,U):=\Var_{\rhoh{0}\sim \SC_h}\left[\expval{O}_{t}\right]$.
\end{supplemental_proposition}

\begin{specialproof}
We recall the expression in Eq.~\eqref{Eq:OutputExpression} from Supplemental Proposition~\ref{sup-prop:qrp-output} and identify the term containing the initial state $\rhoh{0}$ of hidden qubits:
\begin{align}
    \expval{O}_{t}&=\sum_{\Xv_{h_{t-1}}\in\{|r\rangle\langle r'|_h\}}\Tr\left[\left(X_{h_1},\IC_{a};X_{h_2},\IC_{a};\cdots;X_{h_{t-1}},\IC_{a};O\right)U^{\otimes t}\left(\rhoh{0},\rhoa{1};X_{h_1},\rhoa{2};X_{h_2},\rhoa{3};\cdots;X_{h_{t-1}},\rhoa{t}\right)(U^\dag)^{\otimes t}\right]\\
    &=\sum_{\Xv_{h_{t-1}}\in\{|r\rangle\langle r'|_h\}}\Tr\left[\left(X_{h_1},\IC_{a};X_{h_2},\IC_{a};\cdots;X_{h_{t-1}},\IC_{a};O\right)U^{\otimes t}\left(\rhoh{0},\rhoa{1};X_{h_1},\rhoa{2};X_{h_2},\rhoa{3};\cdots;X_{h_{t-1}},\rhoa{t}\right)(U^\dag)^{\otimes t}\right]\\
    &=\sum_{\Xv_{h_{t-1}}} \Tr[(X_{h_1},\IC_{a}) U (\rhoh{0},\rhoa{1}) U^\dag]\nonumber\\
    &\qquad\cdot \Tr\left[\left(X_{h_2},\IC_{a};\cdots;X_{h_{t-1}},\IC_{a};O\right)U^{\otimes (t-1)}\left(X_{h_1},\rhoa{2};X_{h_2},\rhoa{3};\cdots;X_{h_{t-1}},\rhoa{t}\right)(U^\dag)^{\otimes (t-1)}\right]\,.
\end{align}
Here, ``$,$'' denotes the tensor products of operators in accessible and hidden space, and ``$;$'' denotes the tensor products of composite operators in different copies of the composite space.

We next assume the initial state $\rhoh{0}$ is obtained via a random unitary $U_0$ such that $\rhoh{0}=U_0 \ketbra{0}{0} U_0^\dag$. We calculate the expectation values over $U_0$ sampled from Haar measure. The first moment is
\begin{align}
    \Ebb_{\rhoh{0}}[\expval{O}_t]&=\sum_{\Xv_{h_{t-1}}} \Ebb_{U_0\sim\Ubb(d_h)}\left[\Tr[(X_{h_1},\IC_{a}) U (U_0 \ketbra{0}{0} U_0^\dag,\rhoa{1}) U^\dag]\right]\nonumber\\
    &\qquad\cdot \Tr\left[\left(X_{h_2},\IC_{a};\cdots;X_{h_{t-1}},\IC_{a};O\right)U^{\otimes (t-1)}\left(X_{h_1},\rhoa{2};X_{h_2},\rhoa{3};\cdots;X_{h_{t-1}},\rhoa{t}\right)(U^\dag)^{\otimes (t-1)}\right]\,,\label{A101}
\end{align}
where
\begin{align}
    \Ebb_{U_0\sim\Ubb(d_h)}\left[\Tr[(X_{h_1},\IC_{a}) U (U_0 \ketbra{0}{0} U_0^\dag,\rhoa{1}) U^\dag]\right]
    &=\Tr[(X_{h_1},\IC_{a}) U \left[\left(\int_{\Ubb(d_h)} dU_0 U_0 \ketbra{0}{0} U_0^\dag \right),\rhoa{1}\right] U^\dag]\\
    &=\Tr[(X_{h_1},\IC_{a}) U \left(\frac{\IC_h}{d_h},\rhoa{1}\right) U^\dag]\,.
\end{align}
The second moment is
\begin{align}
    &\Ebb_{\rhoh{0}}[\expval{O}^2_t]\nonumber\\
    =&\sum_{\Xv_{h_{t-1}}} \sum_{\Yv_{h_{t-1}}}\Ebb_{\rhoh{0}}\left[ \Tr[(X_{h_1},\IC_{a}) U (\rhoh{0},\rhoa{1}) U^\dag]\Tr[(Y_{h_1},\IC_{a}) U (\rhoh{0},\rhoa{1}) U^\dag]\right]\nonumber\\
    &\cdot \Tr\bigg[\left(X_{h_2},\IC_{a};\cdots;X_{h_{t-1}},\IC_{a};O;Y_{h_2},\IC_{a};\cdots;Y_{h_{t-1}},\IC_{a};O\right)U^{\otimes t}\\
    &\qquad\left(X_{h_1},\rhoa{2};X_{h_2},\rhoa{3};\cdots;X_{h_{t-1}},\rhoa{t};Y_{h_1},\rhoa{2};Y_{h_2},\rhoa{3};\cdots;Y_{h_{t-1}},\rhoa{t}\right)(U^\dag)^{\otimes t}\bigg]\,,\label{A105}
\end{align}
where
\begin{align}
&\Ebb_{\rhoh{0}}\left[\Tr[(X_{h_1},\IC_{a}) U (\rhoh{0},\rhoa{1}) U^\dag]\Tr[(Y_{h_1},\IC_{a}) U (\rhoh{0},\rhoa{1}) U^\dag]\right]\nonumber\\
    =&\Tr[(X_{h_1},\IC_{a};Y_{h_1},\IC_{a}) U^{\otimes 2} \left(\int_{\Ubb(d_h)} dU_0 U_0^{\otimes 2} \ketbra{0}{0}^{\otimes 2} (U_0^\dag)^{\otimes 2} \otimes\rhoa{1}^{\otimes 2} \right)(U^\dag)^{\otimes 2}]\\
    =&\frac{1}{d_h(d_h+1)}\Tr[(X_{h_1},\IC_{a};Y_{h_1},\IC_{a}) U^{\otimes 2} \left[\left( \IC_h^{\otimes 2} + \rm{SWAP_{h,h}} \right) \otimes\rhoa{1}^{\otimes 2} \right](U^\dag)^{\otimes 2}]\\
    =&\frac{1}{d_h(d_h+1)}\Tr[(X_{h_1},\IC_{a}) U (I ,\rhoa{1} )U^\dag] \Tr[(Y_{h_1},\IC_{a}) U (I  ,\rhoa{1} )U^\dag]\nonumber\\
    &\qquad+\frac{1}{d_h(d_h+1)}\Tr[(X_{h_1},\IC_{a};Y_{h_1},\IC_{a}) U^{\otimes 2} (
    {\rm SWAP_{h,h}} \otimes\rhoa{1}^{\otimes 2} )(U^\dag)^{\otimes 2}]\,.
\end{align}
Here in the last equality the SWAP operator is applied on the two copies of hidden spaces of the first iteration, i.e. ${\rm SWAP}_{h,h}=\sum_{i,j=0}^{d_h-1} \ketbra{i}{j}\otimes\ketbra{j}{i}$. We next combine the expressions of the first and second moment, Eq.~\eqref{A101} and~\eqref{A105}, to calculate the variance of initial state:
\begin{align}
    \Var_{\rhoh{0}}\expval{O}_t &=\Ebb_{\rhoh{0}}[\expval{O}^2_t]-\Ebb_{\rhoh{0}}[\expval{O}_t]^2\\
    &=\sum_{\Xv_{h_{t-1}}} \sum_{\Yv_{h_{t-1}}}\nonumber\\
    &\quad\Bigg(\frac{1}{d_h(d_h+1)}\Tr[(X_{h_1},\IC_{a}) U (\IC_h ,\rhoa{1} )(U^\dag)] \Tr[(Y_{h_1},\IC_{a}) U (\IC_h  ,\rhoa{1} )(U^\dag)]\nonumber\\
    &\qquad+\frac{1}{d_h(d_h+1)}\Tr[(X_{h_1},\IC_{a};Y_{h_1},\IC_{a}) U^{\otimes 2} (
    {\rm SWAP_{h,h}} \otimes\rhoa{1}^{\otimes 2} )(U^\dag)^{\otimes 2}]\nonumber\\
    &\qquad\quad-\Tr[(X_{h_1},\IC_{a}) U \left(\frac{\IC_h}{d_h},\rhoa{1}\right) U^\dag]\Tr[(Y_{h_1},\IC_{a}) U \left(\frac{\IC_h}{d_h},\rhoa{1}\right) U^\dag]\Bigg)\nonumber\\
    &\qquad\cdot \Tr\left[\left(X_{h_2},\IC_{a};\cdots;X_{h_{t-1}},\IC_{a};O;Y_{h_2},\IC_{a};\cdots;Y_{h_{t-1}},\IC_{a};O\right)U^{\otimes t}\left(X_{h_1},\rhoa{2};\cdots;X_{h_{t-1}},\rhoa{t};Y_{h_1},\rhoa{2};\cdots;Y_{h_{t-1}},\rhoa{t}\right)(U^\dag)^{\otimes t}\right]\\
    &=\frac{1}{d_h(d_h+1)}\sum_{\Xv_{h_{t-1}}} \sum_{\Yv_{h_{t-1}}}\nonumber\\
    &\qquad\Bigg(\Tr[(X_{h_1},\IC_{a};Y_{h_1},\IC_{a}) U^{\otimes 2} (
    {\rm SWAP_{h,h}} \otimes\rhoa{1}^{\otimes 2} )(U^\dag)^{\otimes 2}]\nonumber\\
    &\qquad\qquad-\frac{1}{d_h}\Tr[(X_{h_1},\IC_{a}) U (\IC_h ,\rhoa{1} )(U^\dag)] \Tr[(Y_{h_1},\IC_{a}) U (\IC_h  ,\rhoa{1} )(U^\dag)]\Bigg)\nonumber\\
    &\quad\cdot \Tr\big[\left(X_{h_2},\IC_{a};\cdots;X_{h_{t-1}},\IC_{a};O;Y_{h_2},\IC_{a};\cdots;Y_{h_{t-1}},\IC_{a};O\right)U^{\otimes 2(t-1)}\nonumber \\
    &\qquad\qquad\left(X_{h_1},\rhoa{2};\cdots;X_{h_{t-1}},\rhoa{t};Y_{h_1},\rhoa{2};\cdots;Y_{h_{t-1}},\rhoa{t}\right)(U^\dag)^{\otimes 2(t-1)}\big]\,.
\end{align}
We then decompose the $\rm{SWAP}$ operator using its definition and obtain 
\begin{align}
    \Var_{\rhoh{0}}\expval{O}_t&=\frac{1}{d_h(d_h+1)}\sum_{\Xv_{h_{t-1}}} \sum_{\Yv_{h_{t-1}}}\nonumber \\ \label{Eq:var_O_init}
    &\Bigg(\sum_{i,j=0}^{d_h-1}\Tr\Big[\left(X_{h_1},\IC_{a};X_{h_2},\IC_{a};\cdots;X_{h_{t-1}},\IC_{a};O;Y_{h_1},\IC_{a};Y_{h_2},\IC_{a};\cdots;Y_{h_{t-1}},\IC_{a};O\right)U^{\otimes 2t}\nonumber\\
    &\qquad\left( \ketbra{i}{j},\rhoa{1}  ;X_{h_1},\rhoa{2};\cdots;X_{h_{t-1}},\rhoa{t};\ketbra{j}{i},\rhoa{1}; Y_{h_1},\rhoa{2};\cdots;Y_{h_{t-1}},\rhoa{t}\right)(U^\dag)^{\otimes 2t}\Big]\nonumber\\
    &\qquad\qquad\qquad-\frac{1}{d_h}\Tr\Big[\left(X_{h_1},\IC_{a};X_{h_2},\IC_{a};\cdots;X_{h_{t-1}},\IC_{a};O;Y_{h_1},\IC_{a};Y_{h_2},\IC_{a};\cdots;Y_{h_{t-1}},\IC_{a};O\right)U^{\otimes 2t}\nonumber\\
    &\qquad\qquad\qquad\qquad\left(\IC_h,\rhoa{1};X_{h_1},\rhoa{2};\cdots;X_{h_{t-1}},\rhoa{t};\IC_h,\rhoa{1};Y_{h_1},\rhoa{2};\cdots;Y_{h_{t-1}},\rhoa{t}\right)(U^\dag)^{\otimes 2t}\Big]\Bigg)\,,
\end{align}
where in the last equality we obtain the standard form for Haar integral over reservoirs $U$. This quantity gives the input dependence of QRP for a fix reservoir. To find the averaged input dependence given an ensemble of reservoirs, we further calculate the expectation value of the variance, and here over Haar random reservoirs. We use again Lemma~\ref{lemma:L3_deep_QNN} to find the dominant term of the memory indicator:
\begin{align}
    \Ebb_{U}[\Var_{\rhoh{0}}\expval{O}_t]=&\frac{1}{d_h(d_h+1)}\frac{1}{d_a^{2t} d_h^{2t}}\left(1+\OC\left(\frac{1}{d_a d_h}\right)\right)\sum_{\Xv_{h_{t-1}}} \sum_{\Yv_{h_{t-1}}}\sum_{\sg\in P_{d_a d_h}(\mathcal{S}_{t})} \nonumber\\
    &\Tr\Big[\left(X_{h_1},\IC_{a};X_{h_2},\IC_{a};\cdots;X_{h_{t-1}},\IC_{a};O;Y_{h_1},\IC_{a};Y_{h_2},\IC_{a};\cdots;Y_{h_{t-1}},\IC_{a};O\right) \sg \Big]\nonumber\\
    &\qquad\Bigg(\sum_{i,j=0}^{d_h-1}\Tr\Big[\left( \ketbra{i}{j},\rhoa{1}  ;X_{h_1},\rhoa{2};\cdots;X_{h_{t-1}},\rhoa{t};\ketbra{j}{i},\rhoa{1}; Y_{h_1},\rhoa{2};\cdots;Y_{h_{t-1}},\rhoa{t}\right)\sg^{-1}\Big]\nonumber\\
    &\qquad\qquad-\frac{1}{d_h}\Tr\Big[\left(I,\rhoa{1};X_{h_1},\rhoa{2};\cdots;X_{h_{t-1}},\rhoa{t};I,\rhoa{1};Y_{h_1},\rhoa{2};\cdots;Y_{h_{t-1}},\rhoa{t}\right)\sg^{-1}\Big]\Bigg)\,.
\end{align}
Here we observe that the final expression is given by the difference of two Haar integrals. We compare the terms therein and find that all the terms with $P_{d_a d_h}(\sg^{-1})$ containing identity operators on the first or $(t+1)$-th space will be cancelled, since $\sum_{i,j}\Tr[\ketbra{i}{j}]g(\ketbra{j}{i})=g(\IC)$ for arbitrary function $g$. Among all the rest terms, we identify the leading term corresponding to the permutation $\bigotimes_{\t=1}^{t}{\rm SWAP}_{\t,t+\t}$ and obtain the initial state dependence for QRP with Haar random initial state and reservoir:
\begin{align}
    \Ebb_{U}[\Var_{\rhoh{0}}\expval{O}_t]=&\frac{1}{d_h(d_h+1)}\frac{1}{d_a^{2t} d_h^{2t}}\left(1+\OC\left(\frac{1}{d_h}\right)\right)\sum_{\Xv_{h_{t-1}}} \sum_{\Yv_{h_{t-1}}}\left(\sum_{i,j=0}^{d_h-1}\Tr[|i\rangle\langle j|j\rangle\langle i|]\right) \left(\prod_{\t=1}^{t-1} \Tr[X_\t Y_\t]^2\Tr[I_a]\right)\Tr[O^2]\\
    =&\frac{1}{d_h d_a^t}\left(1+\OC\left(\frac{1}{d_h}\right)\right)\,.
\end{align}
To calculate the expectation of memory indicator for early inputs $\MC_a(t;\SC_a,U):=\Var_{\rhoa{\t}\sim \SC_a}\left[\expval{O}_{(t+\t)}\right]$, we use the same approach. Without loss of generality, we here set $\t=0$. We first calculate the output's variance over the input state at an arbitrary time step. Compared with Eq.~\eqref{Eq:var_O_init} the only difference is that the Haar integral is taken over unitaries in accessible space instead of hidden space. That is, we view the objective input step as the first step and consider the ensemble of input states $\{\rhoa{1}=V\rho_0 V^\dag : V\in\Ubb(d_a)\}$. Thus we obtain:
\begin{align}
    \Var_{\rhoa{1}}\expval{O}_t&=\frac{1}{d_a(d_a+1)}\sum_{\Xv_{h_{t-1}}} \sum_{\Yv_{h_{t-1}}} \nonumber\\ 
    &\Bigg(\sum_{i,j=0}^{d_a-1}\Tr\Big[\left(X_{h_1},\IC_{a};X_{h_2},\IC_{a};\cdots;X_{h_{t-1}},\IC_{a};O;Y_{h_1},\IC_{a};Y_{h_2},\IC_{a};\cdots;Y_{h_{t-1}},\IC_{a};O\right)U^{\otimes 2t}\nonumber\\
    &\qquad\left( \rhoh{0},\ketbra{i}{j} ;X_{h_1},\rhoa{2};\cdots;X_{h_{t-1}},\rhoa{t}; \rhoh{0},\ketbra{j}{i}; Y_{h_1},\rhoa{2};\cdots;Y_{h_{t-1}},\rhoa{t}\right)(U^\dag)^{\otimes 2t}\Big]\nonumber\\
    &\qquad\qquad\qquad-\frac{1}{d_a}\Tr\Big[\left(X_{h_1},\IC_{a};X_{h_2},\IC_{a};\cdots;X_{h_{t-1}},\IC_{a};O;Y_{h_1},\IC_{a};Y_{h_2},\IC_{a};\cdots;Y_{h_{t-1}},\IC_{a};O\right)U^{\otimes 2t}\label{B82}\nonumber\\
    &\qquad\qquad\qquad\qquad\left(\rhoh{0},\IC_{a};X_{h_1},\rhoa{2};\cdots;X_{h_{t-1}},\rhoa{t};\rhoh{0},\IC_{a};Y_{h_1},\rhoa{2};\cdots;Y_{h_{t-1}},\rhoa{t}\right)(U^\dag)^{\otimes 2t}\Big]\Bigg)\,.
\end{align}
Finally we average this quantity over Haar random reservoirs and find that the leading term is given by the same permutation as before:
\begin{align}
    \Ebb_{U}[\Var_{\rhoa{1}}\expval{O}_t]=&\frac{1}{d_a(d_a+1)}\frac{1}{d_a^{2t} d_h^{2t}}\left(1+\OC\left(\frac{1}{d_a d_h}\right)\right)\sum_{\Xv_{h_{t-1}}} \sum_{\Yv_{h_{t-1}}}\sum_{\sg\in P_{d_a d_h}(\mathcal{S}_{t})}\nonumber \\
    &\Tr\Big[\left(X_{h_1},\IC_{a};X_{h_2},\IC_{a};\cdots;X_{h_{t-1}},\IC_{a};O;Y_{h_1},\IC_{a};Y_{h_2},\IC_{a};\cdots;Y_{h_{t-1}},\IC_{a};O\right) \sg \Big]\nonumber\\
    &\qquad\Bigg(\sum_{i,j=0}^{d_a-1}\Tr\Big[\left( \rhoh{0},\ketbra{i}{j}  ;X_{h_1},\rhoa{2};\cdots;X_{h_{t-1}},\rhoa{t};\rhoh{0},\ketbra{j}{i}; Y_{h_1},\rhoa{2};\cdots;Y_{h_{t-1}},\rhoa{t}\right)\sg^{-1}\Big]\nonumber\\
    &\qquad\qquad-\frac{1}{d_a}\Tr\Big[\left(\rhoh{0},\IC_{a};X_{h_1},\rhoa{2};\cdots;X_{h_{t-1}},\rhoa{t};\rhoh{0},\IC_{a};Y_{h_1},\rhoa{2};\cdots;Y_{h_{t-1}},\rhoa{t}\right)\sg^{-1}\Big]\Bigg)\label{B87}\\
    =&\frac{1}{d_a(d_a+1)}\frac{1}{d_a^{2t} d_h^{2t}}\left(1+\OC\left(\frac{1}{d_h}\right)\right)\sum_{\Xv_{h_{t-1}}} \sum_{\Yv_{h_{t-1}}}\left(\sum_{i,j=0}^{d_a-1}\Tr[|i\rangle\langle j|j\rangle\langle i|]\right) \left(\prod_{\t=1}^{t-1} \Tr[X_\t Y_\t]^2\Tr[I_a]\right)\Tr[O^2]\\
    =&\frac{1}{d_h d_a^t}\left(1+\OC\left(\frac{1}{d_h}\right)\right)\,.
\end{align}
\end{specialproof}

\bigskip

\paragraph*{\underline{We are now ready for the main proof of Theorem~\ref{thm_main:mem_ind_upp_bound}.}} Here, we consider an arbitrary ensemble of input states instead of Haar random inputs. That is, $\rhoa{1}$ is sampled from an ensemble $\SC_a$ with a measure $d\nu(\rhoa{1})$ and the average of input states is $\Bar{\rho}:=\int_{\rhoa{1}} d\nu(\rhoa{1}) \rhoa{1}$. Similarly we consider an initial state ensemble $\SC_h$ with the measure $d\nu(\rhoh{0})$ and the averaged initial state 
$\overline{\rhoh{0}}$.

We restate the theorem again below for readers' convenience. In addition, we note that complementary to Fig~\ref{fig:inp_dep} in the main text, further numerical simulation is provided in  Fig.~\ref{fig:init_dep}. In particular, it shows the scaling of $\MC_h(t;\SC_h,U)$ which is found to be in an agreement with the theoretical prediction, further supporting Theorem~\ref{thm_main:mem_ind_upp_bound}.

\begin{customthm}{2}[Upper bound of memory of arbitrary early input and initial state ensemble)]\label{thm:mem_ind_upp_bound}
    For extreme scrambling reservoirs and an arbitrary ensemble of inputs $\SC_a$,  we have
\begin{align}
       \Ebb_{U\sim\Ubb(d_a d_h)}\left[\MC_a(t;\SC_a,U)\right]\in\OC\left(\frac{1}{d_h d_a^t}\right)\;.
\end{align}
Similarly, for an arbitrary ensemble of initial reservoir states $\SC_h$, we have 
\begin{align}
    \Ebb_{U\sim\Ubb(d_a d_h)}\left[\MC_h(t;\SC_h,U)\right]\in\OC\left(\frac{1}{d_h d_a^t}\right)\;.
\end{align}
\end{customthm}

\begin{figure}[ht]
    \centering
    \includegraphics[width=0.48\textwidth]{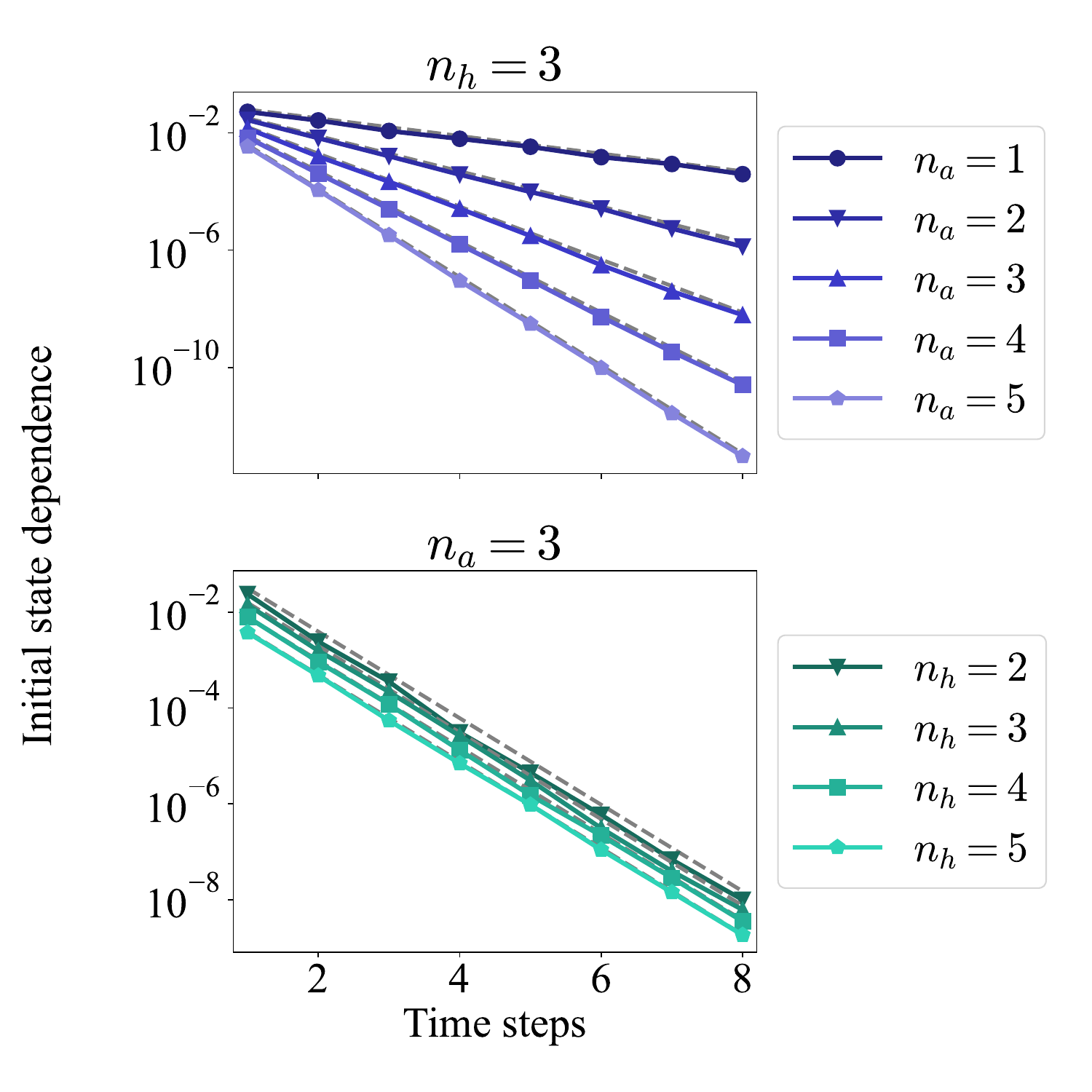}
    
    \caption{\textbf{Initial reservoir state dependence.} The averaged dependence of outputs on initial state of hidden qubits of $\Ebb_{U\sim\Ubb(d_a d_h)}[\MC_h(t;\SC_h,U)]$ are plotted for different qubit numbers. Here both the reservoirs and initial states are Haar random. The grey lines indicate the bounds given in Theorem~\ref{thm_main:mem_ind_upp_bound}.
    }
    \label{fig:init_dep}
\end{figure}
\begin{specialproof}
Using the similar approach as before, we obtain
\begin{align}
    \Var_{\rhoa{1}\sim \SC_a }\expval{O}_t&=\sum_{\Xv_{h_{t-1}}} \sum_{\Yv_{h_{t-1}}} \nonumber\\
    &\Bigg(\int d\nu(\rhoa{1}){\rhoa{1}}\Tr\Big[\left(X_{h_1},\IC_{a};X_{h_2},\IC_{a};\cdots;X_{h_{t-1}},\IC_{a};O;Y_{h_1},\IC_{a};Y_{h_2},\IC_{a};\cdots;Y_{h_{t-1}},\IC_{a};O\right)U^{\otimes 2t}\nonumber\\
    &\qquad\left( \rhoh{0},\rhoa{1} ;X_{h_1},\rhoa{2};\cdots;X_{h_{t-1}},\rhoa{t}; \rhoh{0},\rhoa{1}; Y_{h_1},\rhoa{2};\cdots;Y_{h_{t-1}},\rhoa{t}\right)(U^\dag)^{\otimes 2t}\Big]\nonumber\\
    &\qquad\qquad\qquad-\Tr\Big[\left(X_{h_1},\IC_{a};X_{h_2},\IC_{a};\cdots;X_{h_{t-1}},\IC_{a};O;Y_{h_1},\IC_{a};Y_{h_2},\IC_{a};\cdots;Y_{h_{t-1}},\IC_{a};O\right)U^{\otimes 2t}\nonumber\\
    &\qquad\qquad\qquad\qquad\left(\rhoh{0},\Bar{\rho};X_{h_1},\rhoa{2};\cdots;X_{h_{t-1}},\rhoa{t};\rhoh{0},\Bar{\rho};Y_{h_1},\rhoa{2};\cdots;Y_{h_{t-1}},\rhoa{t}\right)(U^\dag)^{\otimes 2t}\Big]\Bigg)\,,
\end{align}
where $\Bar{\rho}:=\int_{\rhoa{1}} d\nu(\rhoa{1}) \rhoa{1}$. We average the variance over Haar random reservoirs and obtain
\begin{align}
    \Ebb_{U}[\Var_{\rhoa{1}\sim \SC_a }\expval{O}_t]=&\frac{1}{d_a^{2t} d_h^{2t}}\left(1+\OC\left(\frac{1}{d_a d_h}\right)\right)\sum_{\Xv_{h_{t-1}}} \sum_{\Yv_{h_{t-1}}}\sum_{\sg\in P_{d_a d_h}(\mathcal{S}_{t})} \nonumber\\
    &\Tr\Big[\left(X_{h_1},\IC_{a};X_{h_2},\IC_{a};\cdots;X_{h_{t-1}},\IC_{a};O;Y_{h_1},\IC_{a};Y_{h_2},\IC_{a};\cdots;Y_{h_{t-1}},\IC_{a};O\right) \sg \Big]\nonumber\\
    &\qquad\Bigg(\int d\nu(\rhoa{1})\Tr\Big[\left( \rhoh{0},\rhoa{1}  ;X_{h_1},\rhoa{2};\cdots;X_{h_{t-1}},\rhoa{t};\rhoh{0},\rhoa{1}; Y_{h_1},\rhoa{2};\cdots;Y_{h_{t-1}},\rhoa{t}\right)\sg^{-1}\Big]\nonumber\\
    &\qquad\qquad-\Tr\Big[\left(\rhoh{0},\Bar{\rho};X_{h_1},\rhoa{2};\cdots;X_{h_{t-1}},\rhoa{t};\rhoh{0},\Bar{\rho};Y_{h_1},\rhoa{2};\cdots;Y_{h_{t-1}},\rhoa{t}\right)\sg^{-1}\Big]\Bigg)\\
    =&\frac{1}{d_a^{2t} d_h^{2t}}\left(1+\OC\left(\frac{1}{d_h}\right)\right)\sum_{\Xv_{h_{t-1}}} \sum_{\Yv_{h_{t-1}}}\left(\int d\nu(\rhoa{1})\Tr[\rhoa{1}^2]-\Tr[\Bar{\rho}^2]\right) \left(\prod_{\t=1}^{t-1} \Tr[X_\t Y_\t]^2\Tr[\IC_a]\right)\Tr[O^2]\\
    =&\frac{1}{d_h d_a^t}\left(1+\OC\left(\frac{1}{d_h}\right)\right)\left(1-\Tr[\Bar{\rho}^2]\right)\\
    \leq&\frac{1}{d_h d_a^t}\left(1+\OC\left(\frac{1}{d_h}\right)\right)\,.\label{eq:A117}
\end{align}
Here in the third equality, the term $\Tr[\Bar{\rho}^2]$ is the purity of the averaged input states. The limit of this purity, as the distribution of inputs approaches Haar measure, equals $1/d_a$, which can also be identified in Eq.~\eqref{B87} and~\eqref{B82}.

Using the same approach, we also upper bound the initial state dependence for an arbitrary ensemble of initial state. We obtain
\begin{align}
    \Ebb_{U}[\Var_{\rhoh{0}\sim \SC_h }\expval{O}_t]=&\frac{1}{d_h d_a^t}\left(1+\OC\left(\frac{1}{d_h}\right)\right)\left(1-\Tr[\overline{\rhoh{0}}^2]\right)\\
    \leq&\frac{1}{d_h d_a^t}\left(1+\OC\left(\frac{1}{d_h}\right)\right)\;,
\end{align}
which completes the proof.
\end{specialproof}

Next, we prove the Corollary~\ref{cor:var_mem} on exponential upper bound of memory indicator's variance. To do this, we first show a lemma:
\begin{lemma}\label{lemma:upper_bound_var_by_E}
    Given a positive function $g(A)$, we have 
    \be
        \Var_A[g(A)]\leq\max_A[g(A)]\Ebb_A[g(A)]\,.
    \ee
\end{lemma}
\begin{specialproof}
    \begin{align}
        \Var_A[g(A)]&=\Ebb_A[g(A)^2]-\Ebb_A[g(A)]^2\\
        &\leq\Ebb_A[g(A)^2]\\
        &\leq\Ebb_A[g(A)\cdot \max_A[g(A)]]\\
        &=\max_A[g(A)]\Ebb_A[g(A)]
    \end{align}
\end{specialproof}

By applying this lemma to $\MC_a$, we obtain
\begin{equation}
    \Var_U[\MC_a(t;\SC_a,U)]\leq\max_U[\MC_a(t;\SC_a,U)]\Ebb_U[\MC_a(t;\SC_a,U)]\,, \label{eq:A125}
\end{equation}
where $\MC_a(t;\SC_a,U)=\Var_{\rhoa{\t}\sim \SC_a }\expval{O}_{\t+t}$, for which we can apply Lemma~\ref{lemma:upper_bound_var_by_E} again and obtain
\begin{align}
    \Var_{\rhoa{\t}\sim \SC_a }\expval{O}_{\t+t}&\leq\max_{\rhoa{\t}}[\expval{O}_{\t+t}]\Ebb_{\rhoa{\t}}[\expval{O}_{\t+t}]\\
    &\leq O_{\rm{max}}^2\in\OC(1)\,,\label{eq:A127}
\end{align}
where $O_{\rm{max}}$ denotes the maximal eigenvalue of $O$.

We then combine Eq.~\eqref{eq:A117}, \eqref{eq:A125}, and~\eqref{eq:A127}, and obtain the following corollary:
\begin{corollary}[Concentration of memory indicators]\label{cor:var_mem}
    For extreme scrambling reservoirs and an arbitrary ensemble of inputs $\SC_a$, we have
\begin{align}
       \Var_{U\sim\Ubb(d_a d_h)}\left[\MC_a(t;\SC_a,U)\right]\in\OC\left(\frac{1}{d_h d_a^t}\right)\;.
\end{align}
Similarly, for an arbitrary ensemble of initial reservoir states $\SC_h$, we have 
\begin{align}
    \Var_{U\sim\Ubb(d_a d_h)}\left[\MC_h(t;\SC_h,U)\right]\in\OC\left(\frac{1}{d_h d_a^t}\right)\;.
\end{align}
\end{corollary}

{
\subsection{QRP with adaptive mid-circuit measurements and feed-forward}
\label{app:adaptive_mid_meas_feedforward}

Mid-circuit measurements and feedback have been proposed as mechanisms for extracting information from a quantum reservoir during temporal processing, while allowing the measurement outcomes to influence subsequent reservoir operations. Early work on online QRC introduced weak
measurement protocols that balance information extraction against
measurement-induced disturbance, allowing the reservoir to retain part of its
temporal memory during continuous processing~\cite{mujal2023time}. Subsequent
work showed that the corresponding quantum back-action can itself be tuned as
a computational resource~\cite{franceschetto2024harnessing}. More explicit
feedback architectures use the observed outcomes to modify later reservoir
operations. In particular, Ref.~\cite{kobayashi2024feedback} proposed
projective measurements of all reservoir qubits followed by re-encoding of
the outcomes into the next processing cycle. Ref.~\cite{murauer2025feedback}
introduced mid-circuit measurements, resets, and classical feed-forward
operations performed within the coherence time of the device, while
Ref.~\cite{monomi2025feedback} combined weak measurements with feedback to
enhance the nonlinearity and memory capacity of the reservoir. These schemes
can be viewed within the broader framework of monitored QRC, which includes
projective, weak, partial, and dissipative measurement protocols~\cite{morguisancho2026general}.

We now consider a QRP variant in which both the input states and the
mid-circuit measurements may depend on previous measurement outcomes. We show
that this adaptive measurement and feed-forward structure does not remove the
exponential concentration induced by an extreme scrambling reservoir, provided
the complete measurement records used by the model are efficiently sampleable, i.e. with at least inverse-polynomially probability. 

\begin{definition}[Adaptive mid-circuit-measurement feed-forward QRP]
\label{def:adaptive_mid_meas_feedforward_qrp}
Consider the same accessible and hidden spaces as in
Eq.~\eqref{eq:app-qrp-total-state}. For each previous measurement record
$\mv_{\tau-1}$, the protocol chooses a POVM
\begin{equation}
    \left\{
    \Pi_{m_\tau}^{\mv_{\tau-1}}
    \right\}_{m_\tau}
\end{equation}
on the accessible subsystem, satisfying
\begin{equation}
    \Pi_{m_\tau}^{\mv_{\tau-1}}\geq 0,
    \qquad
    \sum_{m_\tau}
    \Pi_{m_\tau}^{\mv_{\tau-1}}
    =
    \IC_a .
\end{equation}
For projective measurements, the POVM effects are orthogonal projectors. For
a general POVM, we assume an update with
$(\Pi_{m_\tau}^{\mv_{\tau-1}})^{1/2}$. The adaptive measurement rule
$\mv_{\tau-1}\mapsto \{\Pi_{m_\tau}^{\mv_{\tau-1}}\}_{m_\tau}$ is fixed
before sampling the reservoir $U$.

Let
\begin{equation}
    \mv_\tau=(m_1,\ldots,m_\tau)
\end{equation}
be the measurement record up to time $\tau$, with $\mv_0$ denoting the empty
record. The feed-forward input at time $\tau$ may depend on the previous
record and is denoted by
\begin{equation}
    \rhoa{\tau}^{\mv_{\tau-1}} .
\end{equation}
The input feed-forward rule is also fixed before sampling $U$. The initial
hidden state is
\begin{equation}
    \rhoh{0}^{\mv_0}:=\rhoh{0}.
\end{equation}
Before the adaptive measurement at time $\tau$, the total state is
\begin{equation}
    \rhot{\tau}^{\mv_{\tau-1}}
    =
    U
    \left(
        \rhoa{\tau}^{\mv_{\tau-1}}
        \otimes
        \rhoh{\tau-1}^{\mv_{\tau-1}}
    \right)
    U^\dag .
\end{equation}
For $1\leq \tau\leq t-1$, the conditional probability of observing outcome
$m_\tau$ is
\begin{equation}
    \P(m_\tau\mid \mv_{\tau-1},U)
    =
    \Tr\left[
        \left(
        \Pi_{m_\tau}^{\mv_{\tau-1}}
        \otimes
        \IC_h
        \right)
        \rhot{\tau}^{\mv_{\tau-1}}
    \right].
\end{equation}
The hidden state propagated to the next time step is
\begin{equation}
    \rhoh{\tau}^{\mv_\tau}
    =
    \frac{
    \Tr_a\left[
        \left(
        (\Pi_{m_\tau}^{\mv_{\tau-1}})^{1/2}
        \otimes
        \IC_h
        \right)
        \rhot{\tau}^{\mv_{\tau-1}}
        \left(
        (\Pi_{m_\tau}^{\mv_{\tau-1}})^{1/2}
        \otimes
        \IC_h
        \right)
    \right]
    }{
    \P(m_\tau\mid \mv_{\tau-1},U)
    } .
\end{equation}
The probability of the complete intermediate record is
\begin{equation}
    \P(\mv_{t-1}\mid U)
    =
    \prod_{\tau=1}^{t-1}
    \P(m_\tau\mid \mv_{\tau-1},U).
\end{equation}
The final conditional readout is
\begin{equation}
    \expval{O}_{t,\mv_{t-1}}^{\rm ff}
    =
    \Tr\left[
        O\rhot{t}^{\mv_{t-1}}
    \right].
\end{equation}
\end{definition}

{We note that the final readout in this definition is a fixed quantum observable. An additional classical output computed directly from the retained measurement record is not included in the proposition below. Such a record-only output can
preserve early information without temporal decay, as we now illustrate.}

{\paragraph*{Memory retained in a classical measurement record.} The exponential memory decay established for the original QRP does not generally extend to a broader hybrid protocol that retains and classically postprocesses intermediate measurement outcomes. This does not increase the quantum reservoir's internal memory; instead, an intermediate measurement transfers information from the reservoir to a persistent classical record. As a simple counterexample, suppose that, after the reservoir evolution at step $\tau$, the accessible subsystem is measured with a binary POVM $\{\Pi_0,\Pi_1\}$, producing an outcome $m_\tau\in\{0,1\}$. Retain this outcome and, at any later time $\tau+t$, define the classical-record readout
\begin{equation}
    z_{\tau+t}^{\rm rec}=(-1)^{m_\tau}.
\end{equation}
Its expectation is
\begin{align}
    \overline z_{\tau+t}^{\rm rec}
    &=
    \sum_{\mv_{\tau+t-1}}
    \P(\mv_{\tau+t-1}\mid\rhoa{\tau},U)
    (-1)^{m_\tau}\nonumber\\
    &=
    \P(m_\tau=0\mid\rhoa{\tau},U)
    -
    \P(m_\tau=1\mid\rhoa{\tau},U),
\end{align}
which is independent of the elapsed time $t$. Consequently, the corresponding record-memory indicator is
\begin{align}
    \MC_a^{\rm rec}(t;\SC_a,U)
    =
    \Var_{\rhoa{\tau}\sim\SC_a}
    \left[
        \P(m_\tau=0\mid\rhoa{\tau},U)
        -
        \P(m_\tau=1\mid\rhoa{\tau},U)
    \right],
\end{align}
and need not decay with $t$. An analogous construction prevents temporal decay of a record-memory indicator for the initial hidden state whenever the intermediate outcome contains information about that state. Thus, subsequent quantum scrambling cannot erase information already transferred to a persistent classical register. This counterexample is distinct from the
fixed-quantum-readout proposition below, which does not determine how the input dependence of the record probabilities scales with reservoir size.}

The probability lower bound below is an operational efficiency condition. If
a complete record has exponentially small probability, then estimating the
conditional output associated with that record requires exponentially many
experimental repetitions. Hence we restrict to coarse-grained adaptive
measurement records that are efficiently sampleable.

\begin{supplemental_proposition}[Concentration with efficiently sampleable adaptive records]
\label{sup-prop:adaptive-mid-meas-feedforward-conc}
Consider Definition~\ref{def:adaptive_mid_meas_feedforward_qrp}, with
$U$ sampled from a $2t$-design on $\Ubb(d_a d_h)$. Let
$\mu_{\rm ff}:=\Ebb_{U,\mv_{t-1}}[\expval{O}_{t,\mv_{t-1}}^{\rm ff}]$ and
$O_{\rm c}:=O-\mu_{\rm ff}(\IC_a\otimes\IC_h)$. Assume that every complete
adaptive record retained by the protocol satisfies
$\P(\mv_{t-1}\mid U)\geq p_{\min}$ for every reservoir instance $U$, with
$p_{\min}^{-1}\in\OC(\poly(n))$. Then
\begin{align}
\Var_{U,\mv_{t-1}}
\left[
\expval{O}_{t,\mv_{t-1}}^{\rm ff}
\right]
\in
\OC\left[
\frac{1}{p_{\min}}
\left(
\frac{\Tr[O_{\rm c}^2]}{d_a^2d_h^3}
+
\frac{\Tr[O_{\rm c}]^2}{d_a^2d_h^2}
+
\frac{\Tr[O_{\rm c}^2]}{d_a^{t+1}d_h^2}
\right)
\right].
\end{align}
{Thus, for a fixed bounded-norm quantum readout and efficiently sampleable
records, the adaptive protocol of
Definition~\ref{def:adaptive_mid_meas_feedforward_qrp} remains exponentially
concentrated up to the polynomial prefactor $p_{\min}^{-1}$.}
\end{supplemental_proposition}

\begin{specialproof}
We first center the output. By definition,
\begin{equation}
    \mu_{\rm ff}
    =
    \Ebb_U
    \sum_{\mv_{t-1}}
    \P(\mv_{t-1}\mid U)
    \expval{O}_{t,\mv_{t-1}}^{\rm ff}.
\end{equation}
Hence
\begin{align}
&
\Var_{U,\mv_{t-1}}
\left[
    \expval{O}_{t,\mv_{t-1}}^{\rm ff}
\right]
\nonumber\\
&=
\Ebb_U
\sum_{\mv_{t-1}}
\P(\mv_{t-1}\mid U)
\left(
    \expval{O}_{t,\mv_{t-1}}^{\rm ff}
    -
    \mu_{\rm ff}
\right)^2 .
\label{eq:adaptive-midmeas-centered-var}
\end{align}
This is the second moment of the conditional readout of the centered
observable
\begin{equation}
    O_{\rm c}
    =
    O-\mu_{\rm ff}(\IC_a\otimes\IC_h),
\end{equation}
since
\begin{equation}
    \Tr\left[
        O_{\rm c}\rhot{t}^{\mv_{t-1}}
    \right]
    =
    \expval{O}_{t,\mv_{t-1}}^{\rm ff}
    -
    \mu_{\rm ff}.
\end{equation}

By the efficiently sampleable record assumption,
\begin{equation}
    \P(\mv_{t-1}\mid U)\geq p_{\min}
\end{equation}
for every retained complete record and every reservoir instance. Therefore,
Eq.~\eqref{eq:adaptive-midmeas-centered-var} gives
\begin{align}
&
\Var_{U,\mv_{t-1}}
\left[
    \expval{O}_{t,\mv_{t-1}}^{\rm ff}
\right]
\nonumber\\
&=
\Ebb_U
\sum_{\mv_{t-1}}
\frac{
\left[
\P(\mv_{t-1}\mid U)
\left(
    \expval{O}_{t,\mv_{t-1}}^{\rm ff}
    -
    \mu_{\rm ff}
\right)
\right]^2
}{
\P(\mv_{t-1}\mid U)
}
\nonumber\\
&\leq
\frac{1}{p_{\min}}
\Ebb_U
\sum_{\mv_{t-1}}
\left[
\P(\mv_{t-1}\mid U)
\left(
    \expval{O}_{t,\mv_{t-1}}^{\rm ff}
    -
    \mu_{\rm ff}
\right)
\right]^2 .
\label{eq:adaptive-pmin-centered-reduction}
\end{align}

It remains to bound the unnormalised centered branch second moment. For a
fixed complete record $\mv_{t-1}=(m_1,\ldots,m_{t-1})$, the unnormalised
centered contribution is
\begin{align}
&
\P(\mv_{t-1}\mid U)
\left(
    \expval{O}_{t,\mv_{t-1}}^{\rm ff}
    -
    \mu_{\rm ff}
\right)
\nonumber\\
&=
\sum_{\Xv_{h_{t-1}}}
\Tr\Big[
\left(
X_{h_1},\Pi_{m_1}^{\mv_0};
X_{h_2},\Pi_{m_2}^{\mv_1};
\cdots;
X_{h_{t-1}},\Pi_{m_{t-1}}^{\mv_{t-2}};
O_{\rm c}
\right)
U^{\otimes t}
\nonumber\\
&\qquad\qquad\cdot
\left(
\rhoh{0},\rhoa{1}^{\mv_0};
X_{h_1},\rhoa{2}^{\mv_1};
X_{h_2},\rhoa{3}^{\mv_2};
\cdots;
X_{h_{t-1}},\rhoa{t}^{\mv_{t-1}}
\right)
(U^\dag)^{\otimes t}
\Big].
\label{eq:adaptive-centered-branch}
\end{align}
This follows from the same trace--partial-trace identity used in
Supplemental Proposition~\ref{sup-prop:qrp-output}. The effect
$\Pi_{m_\tau}^{\mv_{\tau-1}}$ appears because, for each adaptive POVM outcome,
\begin{align}
&
\Tr\left[
(X_h\otimes \IC_a)
\left(
(\Pi_{m_\tau}^{\mv_{\tau-1}})^{1/2}\otimes \IC_h
\right)
\rhot{\tau}^{\mv_{\tau-1}}
\left(
(\Pi_{m_\tau}^{\mv_{\tau-1}})^{1/2}\otimes \IC_h
\right)
\right]
\nonumber\\
&=
\Tr\left[
(X_h\otimes \Pi_{m_\tau}^{\mv_{\tau-1}})
\rhot{\tau}^{\mv_{\tau-1}}
\right],
\end{align}
up to the ordering convention for accessible and hidden tensor factors.

Squaring Eq.~\eqref{eq:adaptive-centered-branch}, summing over adaptive
records, and averaging over a $2t$-design gives the same $2t$-copy Haar
integral structure as in the proof of Theorem~\ref{Thm:HaaResVar_app}, with
the replacements
\begin{equation}
    \IC_a
    \longrightarrow
    \Pi_{m_\tau}^{\mv_{\tau-1}},
    \qquad
    O
    \longrightarrow
    O_{\rm c},
    \qquad
    \rhoa{\tau}
    \longrightarrow
    \rhoa{\tau}^{\mv_{\tau-1}}.
\end{equation}
Thus
\begin{align}
&
\Ebb_U
\sum_{\mv_{t-1}}
\left[
\P(\mv_{t-1}\mid U)
\left(
    \expval{O}_{t,\mv_{t-1}}^{\rm ff}
    -
    \mu_{\rm ff}
\right)
\right]^2
\nonumber\\
&=
\frac{1}{(d_ad_h)^{2t}}
\left(
\sum_{\sigma\in P_{d_ad_h}(\mathcal S_{2t})}
\text{adaptive branch contraction associated with }\sigma
\right)
\left(
1+\OC\left(\frac{1}{d_ad_h}\right)
\right).
\label{eq:adaptive-centered-haar}
\end{align}

We now identify the leading contractions. The first leading term, containing
$\Tr[O_{\rm c}^2]$, is induced by the permutation
\begin{equation}
    s=((1)(2)\cdots(t-1)(t,2t)(t+1)(t+2)\cdots(2t-1)),
\end{equation}
which swaps only the two final readout spaces. For a fixed record, the
accessible contribution from the intermediate measured spaces is
\begin{equation}
    \prod_{\tau=1}^{t-1}
    \Tr\left[
        \Pi_{m_\tau}^{\mv_{\tau-1}}
    \right]^2,
\end{equation}
and the hidden-space summations give the same factor $d_h^{2t-3}$ as in the
original proof. Hence the contribution is bounded by
\begin{align}
&
\frac{\Tr[O_{\rm c}^2]}{d_a^{2t}d_h^3}
\sum_{\mv_{t-1}}
\prod_{\tau=1}^{t-1}
\Tr\left[
    \Pi_{m_\tau}^{\mv_{\tau-1}}
\right]^2 .
\label{eq:adaptive-TrOc2-before-tree}
\end{align}
For every outcomes $\mv_{\tau-1}$,
\begin{equation}
    \sum_{m_\tau}
    \Tr\left[
        \Pi_{m_\tau}^{\mv_{\tau-1}}
    \right]
    =
    d_a,
\end{equation}
and therefore
\begin{equation}
    \sum_{m_\tau}
    \Tr\left[
        \Pi_{m_\tau}^{\mv_{\tau-1}}
    \right]^2
    \leq
    d_a^2 .
\end{equation}
Applying this bound at each node of the adaptive measurement tree gives
\begin{equation}
    \sum_{\mv_{t-1}}
    \prod_{\tau=1}^{t-1}
    \Tr\left[
        \Pi_{m_\tau}^{\mv_{\tau-1}}
    \right]^2
    \leq
    d_a^{2(t-1)}.
\end{equation}
Substituting into Eq.~\eqref{eq:adaptive-TrOc2-before-tree}, the first
leading contribution is at most
\begin{equation}
    \frac{\Tr[O_{\rm c}^2]}{d_a^2d_h^3}.
\label{eq:adaptive-leading-TrOc2}
\end{equation}

Second, the identity permutation on all $2t$ copies gives the term containing
$\Tr[O_{\rm c}]^2$. The same adaptive tree bound yields
\begin{equation}
    \frac{\Tr[O_{\rm c}]^2}{d_a^{2t}d_h^2}
    \sum_{\mv_{t-1}}
    \prod_{\tau=1}^{t-1}
    \Tr\left[
        \Pi_{m_\tau}^{\mv_{\tau-1}}
    \right]^2
    \leq
    \frac{\Tr[O_{\rm c}]^2}{d_a^2d_h^2}.
\label{eq:adaptive-leading-TrOc}
\end{equation}

Third, the temporally decaying contribution is induced by the permutation that
swaps the two replicas at every time step,
\begin{equation}
    \sigma=\bigotimes_{\tau=1}^{t}{\rm SWAP}_{\tau,\tau+t}.
\end{equation}
At an intermediate measured accessible space, this gives
\begin{equation}
    \Tr\left[
        \left(
        \Pi_{m_\tau}^{\mv_{\tau-1}}
        \right)^2
    \right].
\end{equation}
Since $0\leq \Pi_{m_\tau}^{\mv_{\tau-1}}\leq \IC_a$,
\begin{equation}
    \Tr\left[
        \left(
        \Pi_{m_\tau}^{\mv_{\tau-1}}
        \right)^2
    \right]
    \leq
    \Tr\left[
        \Pi_{m_\tau}^{\mv_{\tau-1}}
    \right],
\end{equation}
and hence, for every outcomes,
\begin{equation}
    \sum_{m_\tau}
    \Tr\left[
        \left(
        \Pi_{m_\tau}^{\mv_{\tau-1}}
        \right)^2
    \right]
    \leq
    d_a.
\end{equation}
Applying this recursively gives
\begin{equation}
    \sum_{\mv_{t-1}}
    \prod_{\tau=1}^{t-1}
    \Tr\left[
        \left(
        \Pi_{m_\tau}^{\mv_{\tau-1}}
        \right)^2
    \right]
    \leq
    d_a^{t-1}.
\end{equation}
Therefore the temporally decaying contribution is bounded by
\begin{equation}
    \frac{\Tr[O_{\rm c}^2]}{d_a^{2t}d_h^2}
    d_a^{t-1}
    =
    \frac{\Tr[O_{\rm c}^2]}{d_a^{t+1}d_h^2}.
\label{eq:adaptive-temp}
\end{equation}

\paragraph*{\underline{Bounds for all remaining permutations.}}
We now verify that the adaptive measurements and feed-forward inputs cannot
increase the system-size scaling associated with any other permutation in the
$2t$-copy Haar expansion.

Fix a complete record $\mv_{t-1}$ and a permutation
$\sigma\in P_{d_ad_h}(\mathcal{S}_{2t})$. Since the permutation acts on copies
of the composite accessible-hidden space, its accessible and hidden index
contractions can be counted separately, except for the final factors involving
$O_{\rm c}$, which are bounded by either $\Tr[O_{\rm c}^2]$ or
$\Tr[O_{\rm c}]^2$ using the Hilbert--Schmidt Cauchy--Schwarz inequality.

We first consider the accessible-space contractions. Every adaptive POVM effect
satisfies
\begin{equation}
    0\leq
    \Pi_{m_\tau}^{\mv_{\tau-1}}
    \leq
    \IC_a,
\end{equation}
and hence
\begin{equation}
    \left\|
    \Pi_{m_\tau}^{\mv_{\tau-1}}
    \right\|_\infty
    \leq 1,
    \qquad
    \Tr\left[
    \Pi_{m_\tau}^{\mv_{\tau-1}}
    \right]
    \leq d_a .
\end{equation}
Accordingly, for any product of adaptive POVM effects appearing along one
accessible index cycle,
\begin{equation}
    \left|
    \Tr\left[
    \Pi_{m_{\tau_1}}^{\mv_{\tau_1-1}}
    \cdots
    \Pi_{m_{\tau_r}}^{\mv_{\tau_r-1}}
    \right]
    \right|
    \leq d_a .
\label{eq:adaptive-effect-cycle-bound}
\end{equation}
Thus every closed accessible index cycle contributes at most one factor of
$d_a$. A permutation which joins two cycles therefore cannot increase the
accessible-space scaling. This is the adaptive analogue of the replacement
\begin{equation}
    \Tr[\IC_a]\Tr[\IC_a]=d_a^2
    \quad\longrightarrow\quad
    \Tr[\IC_a^2]=d_a
\end{equation}
in the original proof.

The record-dependent input states likewise cannot create additional powers of
$d_a$. Indeed, for any density matrix
$\rhoa{\tau}^{\mv_{\tau-1}}$ and any product $A$ of density matrices and POVM
effects,
\begin{equation}
    \left|
    \Tr\left[
    \rhoa{\tau}^{\mv_{\tau-1}} A
    \right]
    \right|
    \leq
    \left\|
    \rhoa{\tau}^{\mv_{\tau-1}}
    \right\|_1
    \|A\|_\infty
    \leq 1 .
\label{eq:adaptive-input-cycle-bound}
\end{equation}
Therefore, although the value of an accessible contraction may depend on the
previous outcomes through
$\rhoa{\tau}^{\mv_{\tau-1}}$, its dimension scaling is bounded uniformly over
all records. This is the only point at which feed-forward enters the
permutation counting.

We next consider the hidden-space contractions. Choose the same matrix-unit
operator basis as in the proof of Theorem~\ref{Thm:HaaResVar_app},
\begin{equation}
    X_{h_\tau}=\ketbra{i_\tau}{j_\tau},
    \qquad
    Y_{h_\tau}=\ketbra{k_\tau}{l_\tau}.
\end{equation}
Each hidden trace induced by $\sigma$ gives a product of Kronecker delta
functions. Summing over all $X_{h_\tau}$ and $Y_{h_\tau}$ produces one factor
of $d_h$ for each independent index that remains after imposing these delta
constraints. Joining previously distinct hidden index cycles introduces
additional equalities among the indices and therefore cannot increase the
number of free sums.

The adaptive POVMs and feed-forward states act only on the accessible spaces.
Consequently, they do not modify these hidden Kronecker-delta constraints.
The hidden loop counting is therefore identical to that in the original QRP
proof.

There are three extremal contraction patterns:

\begin{enumerate}
    \item The identity permutation leaves the two final readout factors
    disconnected and maximizes both the accessible and hidden loop counts
    among terms proportional to $\Tr[O_{\rm c}]^2$. It gives the scaling
    \begin{equation}
        d_a^{2t-2}d_h^{2t-2}
        \Tr[O_{\rm c}]^2 .
    \end{equation}

    \item The permutation which swaps only the two final readout copies
    maximizes the accessible loop count among the non-temporally-decaying
    terms proportional to $\Tr[O_{\rm c}^2]$. It gives
    \begin{equation}
        d_a^{2t-2}d_h^{2t-3}
        \Tr[O_{\rm c}^2] .
    \end{equation}

    \item The permutation
    $\bigotimes_{\tau=1}^{t}{\rm SWAP}_{\tau,\tau+t}$ pairs
    $X_{h_\tau}$ with $Y_{h_\tau}$ at every intermediate time and therefore
    maximizes the number of duplicated hidden Kronecker deltas among the
    temporally decaying terms. It gives
    \begin{equation}
        d_a^{t-1}d_h^{2t-2}
        \Tr[O_{\rm c}^2] .
    \end{equation}
\end{enumerate}

Any other permutation either merges at least one accessible cycle relative to
the first two cases, or introduces at least one additional independent hidden
Kronecker-delta constraint relative to the corresponding extremal case.
Hence, for every fixed record and every permutation, its absolute contribution
is bounded, up to a constant depending only on $t$, by
\begin{align}
&
d_a^{2t-2}d_h^{2t-3}\Tr[O_{\rm c}^2]
+
d_a^{2t-2}d_h^{2t-2}\Tr[O_{\rm c}]^2
\nonumber\\
&\qquad+
d_a^{t-1}d_h^{2t-2}\Tr[O_{\rm c}^2].
\label{eq:adaptive-all-permutation-envelope}
\end{align}

Let $N_{\rm rec}$ denote the number of retained complete adaptive records.
Since the bound above holds uniformly for every fixed record, summing over
records gives at most an additional factor $N_{\rm rec}$. Moreover, the
efficient-sampling assumption implies
\begin{equation}
    1
    \geq
    \sum_{\mv_{t-1}}
    \P(\mv_{t-1}\mid U)
    \geq
    N_{\rm rec}p_{\min},
\end{equation}
and therefore
\begin{equation}
    N_{\rm rec}
    \leq
    \frac{1}{p_{\min}}.
\label{eq:number-record-bound}
\end{equation}
After division by $(d_ad_h)^{2t}$, the complete permutation sum is thus
bounded as
\begin{align}
&
\Ebb_U
\sum_{\mv_{t-1}}
\left[
\P(\mv_{t-1}\mid U)
\left(
\expval{O}_{t,\mv_{t-1}}^{\rm ff}
-
\mu_{\rm ff}
\right)
\right]^2
\nonumber\\
&\in
\OC_t\left[
N_{\rm rec}
\left(
\frac{\Tr[O_{\rm c}^2]}{d_a^2d_h^3}
+
\frac{\Tr[O_{\rm c}]^2}{d_a^2d_h^2}
+
\frac{\Tr[O_{\rm c}^2]}{d_a^{t+1}d_h^2}
\right)
\right].
\label{eq:adaptive-full-permutation-bound}
\end{align}
Here $\OC_t$ indicates that the implicit constant may depend on the fixed
number of time steps $t$, but not on $d_a$ or $d_h$. The subleading
Weingarten terms carry an additional factor
$\OC((d_ad_h)^{-1})$ and therefore do not change the stated system-size
scaling.

Replacing $\IC_a$ by adaptive
POVM effects cannot increase the accessible-space scaling, because at every
outcomes the effects are positive, bounded above by $\IC_a$, and satisfy the
POVM completeness relation. The feed-forward input states also do not increase
the scaling, since they are density matrices with unit trace and operator norm
at most one.

Combining Eqs.~\eqref{eq:adaptive-leading-TrOc2},
\eqref{eq:adaptive-leading-TrOc}, and~\eqref{eq:adaptive-temp}, we obtain
\begin{align}
&
\Ebb_U
\sum_{\mv_{t-1}}
\left[
\P(\mv_{t-1}\mid U)
\left(
    \expval{O}_{t,\mv_{t-1}}^{\rm ff}
    -
    \mu_{\rm ff}
\right)
\right]^2
\nonumber\\
&\in
\OC\left[
\left(
\frac{\Tr[O_{\rm c}^2]}{d_a^2d_h^3}
+
\frac{\Tr[O_{\rm c}]^2}{d_a^2d_h^2}
\right)
\left(
1+\OC\left(\frac{1}{d_ad_h}\right)
\right)
+
\frac{\Tr[O_{\rm c}^2]}{d_a^{t+1}d_h^2}
\right].
\label{eq:adaptive-centered-second-bound}
\end{align}
Substituting Eq.~\eqref{eq:adaptive-centered-second-bound} into
Eq.~\eqref{eq:adaptive-pmin-centered-reduction} proves the proposition.

Finally, the first moment has the same leading scaling as in the original
QRP proof. The identity permutation gives
\begin{equation}
    \mu_{\rm ff}
    =
    \frac{\Tr[O]}{d_ad_h}
    \left(
    1+\OC\left(\frac{1}{d_ad_h}\right)
    \right),
\end{equation}
because the adaptive tree satisfies, recursively,
\begin{equation}
    \sum_{\mv_{t-1}}
    \prod_{\tau=1}^{t-1}
    \Tr\left[
        \Pi_{m_\tau}^{\mv_{\tau-1}}
    \right]
    =
    d_a^{t-1}.
\end{equation}
Consequently,
\begin{equation}
    \Tr[O_{\rm c}]
    =
    \Tr[O]
    -
    \mu_{\rm ff}d_ad_h
    =
    \OC\left(
    \frac{\Tr[O]}{d_ad_h}
    \right).
\end{equation}
Using $(\Tr[O])^2\leq d_ad_h\Tr[O^2]$, the term
$\Tr[O_{\rm c}]^2/(d_a^2d_h^2)$ is subleading. Moreover, for bounded-norm
readouts, $\Tr[O_{\rm c}^2]\in\OC(d_ad_h)$. Hence the bound is exponentially
small whenever $p_{\min}^{-1}$ is polynomial.
\end{specialproof}

\noindent\textit{Remark.}
If the efficient-sampling assumption is imposed on each single-step
conditional probability,
\begin{equation}
    \P(m_\tau\mid\mv_{\tau-1},U)\geq p_0,
\end{equation}
then the complete-record probability satisfies
\begin{equation}
    \P(\mv_{t-1}\mid U)\geq p_0^{t-1}.
\end{equation}
In this case one should replace $p_{\min}$ in
Supplemental Proposition~\ref{sup-prop:adaptive-mid-meas-feedforward-conc} by
$p_0^{t-1}$. The overhead remains polynomial for fixed $t$ whenever
$p_0^{-1}\in\OC(\poly(n))$.
}

\section{Practical consequences of exponential concentration on QRP}\label{appendix_prob_theory_refresher}
\subsection{Preliminaries: Hypothesis testing}
We first provide some standard results of statistical indistinguishability of probability distributions using hypothesis testing tools
\begin{lemma}\label{lemma:guess-distribution}
Consider two probability distributions $\PC$ and $\QC$ over some finite set of outcomes $\IC$. Suppose we are given a single sample $S$ drawn from either $\PC$ or $\QC$ with equal probability, and we have the following two hypotheses:
\begin{itemize}
    \item Null hypothesis $\mathcal{H}_0$: $S$ is drawn from $\mathcal{P}$\,,
    \item Alternative hypothesis $\mathcal{H}_1:$ $S$ is drawn from $\mathcal{Q}$\,.
\end{itemize}
The probability of correctly deciding the true hypothesis is upper bounded as
\begin{align}
    {\rm Pr}[``{\rm right \; decision \; between \, } \HC_0 \, {\rm and} \, \HC_1"] \leq \frac{1}{2} + \frac{\| \PC - \QC \|_1}{4} \;, 
\end{align}
where we denote $\| \PC - \QC \|_1 = \sum_{s \in \IC} |p(s) - q(s)|$ as the 1-norm between the probability vectors (2 $\times$ the total variation distance).
\end{lemma}
\begin{proof}
There exists a region $\AC$ such that $p(s) > q(s)$ for all $s \in \AC$. The optimal decision making strategy is to choose that the given sample $S$ is drawn from $\PC$ if it falls in the region i.e., $S \in \AC$ and choose $\QC$, otherwise. The probability of making the right decision can be expressed as
\begin{align}
    {\rm Pr}[``{\rm right \; decision \; between \, } \HC_0 \, {\rm and} \, \HC_1"] =&   {\rm Pr}( S \in \AC | S \sim \PC) {\rm Pr}(S \in \PC) + {\rm Pr}( S \notin \AC | S \sim \QC){\rm Pr}(S \in \QC) \\
    =& \frac{1}{2}\left[ {\rm Pr}( S \in \AC | S \sim \PC) + {\rm Pr}( S \notin \AC | S \sim \QC)  \right] \\
    =& \frac{1}{2}\left[ \sum_{s \in \AC} p(s) + \sum_{s \notin \AC} q(s) \right] \;, \label{eq:appx-swap1}
\end{align}
where the second equality is due to the sample being equally likely to be drawn from either $\PC$ or $\QC$. In the last equality, we use the fact that given that the sample is from $\PC$, the probability that this sample takes any value within the region $\AC$ is simply $ \sum_{s \in \AC} p(s)$, and similarly for $s \notin \AC$.

The 1-norm between probability vectors can be written as 
\begin{align}
    \| \PC - \QC \|_1 = & \sum_{s \in \IC} |p(s) - q(s)| \\
    = & \sum_{s \in \AC} (p(s) - q(s)) + \sum_{s \notin \AC} (q(s) - p(s)) \label{eq:tvd}\;,
\end{align}
where we have separated terms in the sum based on the region $\AC$.
Lastly, we notice that
\begin{align}
    \frac{2 + \| \PC - \QC \|_1}{2} & = \frac{1}{2}\left(\sum_{s\in\IC}p(s) + \sum_{s\in\IC}q(s)+ \| \PC - \QC \|_1 \right)\\
    & = \sum_{s \in \AC} p(s) + \sum_{s \notin \AC} q(s) \;,
\end{align}
where in the second line we have used Eq.~\eqref{eq:tvd}. Substituting this back to Eq.~\eqref{eq:appx-swap1}, we obtain the desired result. 
\end{proof}

\medskip

\begin{lemma}\label{sup-prop:indistin-prob}
Consider two binary probability distributions $\PC_0=(p_0,1-p_0)$ and $\PC_\varepsilon=(p_\varepsilon,1-p_\varepsilon)$ where $p_\epsilon = p_0+\varepsilon$. Suppose we are given $N$ samples (denoted as $\MC$) drawn from either $\PC_0$ or $\PC_\varepsilon$ with equal probability, and we have the following two hypotheses:
\begin{itemize}
    \item Null hypothesis $\mathcal{H}_0$: $\mathcal{M}$ is drawn from $\mathcal{P}_0$\,,
    \item Alternative hypothesis $\mathcal{H}_1: \mathcal{M}$ is drawn from $\PC_\varepsilon$\,.
\end{itemize}
The probability of correctly deciding the true hypothesis is upper bounded as
\begin{align}
    {\rm Pr}[``{\rm right \; decision \; between \, } \HC_0 \, {\rm and} \, \HC_1"] \leq \frac{1}{2} + \frac{N|\varepsilon|}{2} \;.
\end{align}
\end{lemma}

\begin{proof}
We remark that the combination of the standard upper bound 
\begin{align}
    \| \PC^{\otimes N} - \QC^{\otimes N} \|_1 \leq N \| \PC - \QC \|_1 \;,
\end{align}
along with Lemma~\ref{lemma:guess-distribution} gives success probability 
\begin{align}
     {\rm Pr}[``{\rm right \; decision \; between \, } \HC_0 \, {\rm and} \, \HC_1"] \leq \frac{1}{2} + \frac{N \|\PC_0 - \PC_{\varepsilon}\|_1}{4}
\end{align}
and explicit evaluation shows that $\|\PC_0 - \PC_{\epsilon}\|_1 = 2|\varepsilon|$.
\end{proof}

\subsection{Main results: Statistical indistinguishability in QRP}\label{app:concentration-consequence-summary}
In the presence of concentration, the QRP predictions estimated with a polynomial number of measurement shots become indistinguishable from random variables that are insensitive to input data. This result is indeed the interplay between the concentration and the polynomial measurement shots. In what follows, we formally show that statistical estimates of these expectations as well as model predictions are indistinguishable from some data-independent random variables. 

\subsubsection{Setting}
To rigorously see the practical implications of the concentration, we consider observables $\{O_k \}_{k=1}^M$ to be Pauli operators and the measurements are made in the respective eigen basis. For the Pauli operator $O_k$, after polynomial number of measurements $N$, the statistical estimate of the expectation can be obtained by 
\begin{align}
    \widehat{O}_{k;t}  = \frac{1}{N} \sum_{i=1}^N  \lambda_i \;,
\end{align}
where $\lambda_i$ is the $i^{\rm th}$ measurement outcome which takes the value $+1$ with the probability $\pplus$ and the value $-1$ with the probability $\pminus$. That is, each sample is drawn form a distribution
\begin{align} \label{eq:prob-dist-sample-expectation}
    \PC_{k;t} = \left\{ \pplus, \pminus \right\} \;.
\end{align}

Next, we assume the exponentially vanishing variance over the instances of quantum reservoir
\begin{align}\label{eq:appx-exp-con-qrp}
   \Var_{U} \left[ \langle O_{k} \rangle_{t} \right] \leq \beta \;, \;\; \beta \in \OC(1/b^n) \; \;,
\end{align}
for all $k$ and some constant $b >1$. By using Chebyshev's inequality, the vanishing variance leads to the exponential concentration as in Definition~\ref{def_prob_exp_concentration} with the concentration point being the average over the reservoirs i.e., $\mu = \Ebb_{U} [\langle O_{k} \rangle_{t}]$. We further assume that this concentration point is zero
\begin{align}\label{eq:appx-exp-con-qrp-fixed-point}
    \Ebb_{U} [\langle O_{k} \rangle_{t}] = 0 \;.
\end{align}

Now, the concentration in expectation values also implies the concentration in probability outcomes (when measured in eigen basis) which is formalized in the following proposition
\begin{supplemental_proposition}\label{prop:con-prob}
    Given the concentration of the expectation value of some observable $O_k$ according to Eq.~\eqref{eq:appx-exp-con-qrp} and Eq.~\eqref{eq:appx-exp-con-qrp-fixed-point}, the outcome probabilities when measuring in eigen basis also concentrate with
    \begin{align}
        \Var_{U} \left[ \pplus\right], \; \Var_{U} \left[ \pminus \right] \in \OC(1/b^n) \;,
    \end{align}
    where $\pplus$ and $ \pminus$ are the probabilities of obtaining $+1$ and $-1$ respectively. In addition, the concentration points are $1/2$ i.e.,
    \begin{align}
        \Ebb_{U}\left[\pplus\right] = \Ebb_{U}\left[\pminus\right] = \frac{1}{2} \;.
    \end{align}
\end{supplemental_proposition}
\begin{proof}
    Let us first decompose the expectation into the eigen basis
    \begin{align}\label{eq:expectation-decompose-eigen-basis}
        \langle O_k\rangle_t =  \pplus - \pminus \;,
    \end{align}
    where the eigenvalues for a Pauli operators are $+1$ and $-1$ respectively. By using Eq.~\eqref{eq:appx-exp-con-qrp-fixed-point} and Eq.~\eqref{eq:expectation-decompose-eigen-basis} together with the normalization condition on the outcome probabilities i.e., $ \pplus +  \pminus = 1$, we have the concentration points to be
    \begin{align}
        \Ebb_{U}\left[\pplus\right] = \Ebb_{U}\left[\pminus\right] = \frac{1}{2}  \;.
    \end{align}
    Similarly, by using the same conditions, the variances of the outcome probabilities can be expressed as
    \begin{align}
        \Var_{U} \left[ \langle O_{k} \rangle_{t} \right] & = \Var_{U} \left[ \pplus \right] + \Var_{U} \left[ \pminus \right] - 2 \Cov_{U}\left[ \pplus , \pminus \right] \;\;, \\
        0 & = \Var_{U} \left[ \pplus \right] + \Var_{U} \left[ \pminus \right] + 2 \Cov_{U}\left[ \pplus , \pminus \right] \;,
    \end{align}
    which leads to
\begin{align}\label{eq:appx-exp-con-qelm-var-prob}
        \Var_{U} \left[ \pplus\right], \; \Var_{U} \left[ \pminus \right] \in \OC(1/b^n) \;.
\end{align}
This completes the proof.
\end{proof}

\subsubsection{Statistical indistinguishability with QRP}
We formally introduce the concepts of statistical indistinguishability at the levels of distributions and of outcomes.
\begin{definition}[Statistical indistinguishability (of distributions)]\label{def:StatIndist} 
    Two probability distributions $\PC$ and $\QC$ are statistically indistinguishable with $N$ samples if a binary hypothesis test cannot be passed with  probability at least $0.51$. That is, given a set of $N$ samples $\MC$ drawn from either $\PC$ or $\QC$ (with an equal probability), consider the following hypotheses
    \begin{itemize}
        \item Null hypothesis $\HC_0$: $\MC$ is drawn from $\PC$\,,
        \item Alternative hypothesis $\HC_1$: $\MC$ is drawn from $\QC$ \,,
    \end{itemize}
    where $\PC$ and $\QC$ are statistically indistinguishable (with $N$ samples) if for any algorithm the probability of correctly identifying the correct hypothesis, ${\rm Pr}[``{\rm right \; decision \; between \, } \HC_0 \, {\rm and} \, \HC_1"]$, satisfies:
    \begin{align}
         {\rm Pr}[``{\rm right \; decision \; between \, } \HC_0 \, {\rm and} \, \HC_1"] \leq 0.51 \;.
    \end{align}
    Note that the threshold $0.51$ in the definition is arbitrary chosen to be close to that of random guessing.
\end{definition}

\begin{definition}[Statistical indistinguishability (of outputs)]\label{def:StatIndistOutputs}
    Consider a map $\Phi:\mathbb{R}^N\rightarrow \mathbb{R}^M$ (with $M$ being the dimension of the output) and two distributions $\PC$ and $\QC$ which are statistically indistinguishable under $N$ samples according to Definition~\ref{def:StatIndist}. Draw $N$ respective samples from $\PC$ and $\QC$, which we respectively denote as $\MC_{\PC}$ and $\MC_{\QC}$.  We say that $\Phi(\mathcal{M}_{\mathcal{P}})$ and $\Phi(\mathcal{M}_{\mathcal{Q}})$ are statistically indistinguishable outputs.
\end{definition}

The formal formal version of Proposition~\ref{prop:prac-cons} in the main text which shows the statistical indistinguishability of the QRP output is stated formally below. We note that it is also stated in a more general manner which includes exponential concentration from any other sources (not just the extreme scrambling reservoirs).
\begin{supplemental_theorem}\label{sup-thm:stat-indis-qrp-output}[Statistical indistinguishability of an expectation value] 
Suppose that the expectation value of some Pauli operator $O_k$ exponentially concentrates towards zero according to Eq.~\eqref{eq:appx-exp-con-qrp} and Eq.~\eqref{eq:appx-exp-con-qrp-fixed-point}. For any reservoir instance $U$, we consider estimating the expectation value in the respective eigen basis. That is, samples are drawn from the distribution $\PC_{k;t}$ in Eq.~\eqref{eq:prob-dist-sample-expectation}. Then, for polynomial number of samples $N \in \OC(\poly(n))$, $\PC_{k;t}$ is indistinguishable from the uniform distribution $\PC_{\rm unif} = \{1/2, 1/2\}$ according to Definition~\ref{def:StatIndist}. In particular, let $\MC$ be a polynomial sized set of samples drawn either from $\PC_{k;t}$ or $\PC_{\rm unif}$ (with an equal probability). We then perform a hypothesis test with:
    \begin{itemize}
        \item Null hypothesis $\HC_0$: $\MC_s$ is drawn from the uniform distribution $\PC_{\rm unif}$\,,
        \item Alternative hypothesis $\HC_s$: $\MC$ is drawn from $\PC_{k;t}$ \,.
    \end{itemize}
With probability at least $1 - \delta$ over the choice of reservoir $U$ regardless of the input signals we have that 
\begin{align}\label{eq:swaptestsuccessv2}
 {\rm Pr}\left[{\rm ``right \, decision \, between \,} \HC_{s} \, {\rm and} \, \HC_0" \right]  \leq \frac{1}{2} + \epsilon \, ,
\end{align}
with $\delta \in \OC(c^{-n})$ for some $c>1$ and $\; \epsilon  \in \OC(c'^{-n})$ for some $c' > 1$. In addition, with high probability exponentially close $1$, the statistical estimate of $\langle O_k \rangle_t$ is indistinguishable from a input-independent random variable of the form
\begin{align}\label{eq:rand-output}
    \widehat{O}_{\text{rand}} = \dfrac{1}{N}\sum_{i=1}^N z_i \;,
\end{align}
where each $z_i$ is drawn from $\PC_{\rm unif}$.
\end{supplemental_theorem}

Consequently, Supplemental Theorem~\ref{sup-thm:stat-indis-qrp-output} leads to poor performance of QRP on unseen input sequences. In particular, while training the model with estimated QRP outputs is guaranteed to minimize the training loss thanks to the convex optimization with the hard coded training labels, the trained QRP model generalizes poorly. The following statement captures this subtlety.
\begin{supplemental_corollary}[Impact of concentration on trainability and generalization]\label{sup-coro:qrp-performance}
    Under the same setting as Supplemental Theorem~\ref{sup-thm:stat-indis-qrp-output}, consider the loss in Eq.~\eqref{eq:loss} with the training dataset $\SC_{\rm tr}$. Given that each QRP output (necessarily for estimating the loss and making prediction on unseen input) is estimated with polynomial measurement shots $N \in \OC(\poly(n))$, then after training the QRP model we have:
    \begin{itemize}
        \item With high probability  $1-\delta$ with $\delta' \in \OC(c'^{-n})$ and $c'>1$, the optimal weights $\hat{\vec{\eta}}^*$ are statistically indistinguishable from an input-independent random vector of the form
        \begin{align}
            \vec{\eta}_{\rm rand} = \left(R_0^T R_0\right)^{-1}R_0^T \, \vec{y} \;,
        \end{align}
        where $R_0$ is an input-independent random matrix of the dimensions $N_{\rm tr}\times M$ whose elements are instances of $ \widehat{O}_{\text{rand}}$ defined in Eq.~\eqref{eq:rand-output} and $\vec{y}$ is a vector whose elements corresponding to the training labels of the dimension $N_{\rm tr}$. Here, $N_{\rm tr} = D_{\rm tr}|\TC_{\rm tr}|$ corresponds to the number of terms in the training loss, and $M$ is the number of observables. 
        \item The empirical training loss evaluated with $\vec{\eta}^*$ is at the global minimum (given those estimated QRP outputs).
        \item For an unseen input sequence $\vrhoa{\t} \notin \SC_{\rm tr}$, with high probability $1-\delta'$ with $\delta' \in \OC(c'^{-n})$, the estimated QRP prediction is independent of $\vrhoa{\t}$ and statistically indistinguishable from
        \begin{align}
            \hat{f}_{\rm rand}  = \vec{O}_{\rm rand}^T \,\vec{\eta}_{\rm rand} \;,
        \end{align}
        where $\vec{O}_{\rm rand}$ is an input-independent vector of length $M$ whose elements are instances of $ \widehat{O}_{\text{rand}}$.
    \end{itemize}
\end{supplemental_corollary}

\subsection{Proofs of analytical results}\label{app:concentration-consequence-proof}
\subsubsection{Proof of Supplemental Theorem~\ref{sup-thm:stat-indis-qrp-output}}
\begin{proof}
First, recall that the concentration of $\langle O_k \rangle_t$ implies the concentration of outcome probabilities as in Proposition~\ref{prop:con-prob}.
Our proof strategy is to show that, due to exponential concentration in outcome probabilities, $\Delta_{p}^{(+)} =|\pplus -1 /2|$ and $\Delta_{p}^{(-)}  = |\pminus-1/2|$ are very likely (i.e., with probability exponentially close to 1) to be exponentially small for any given reservoir $U$. This exponentially small value corresponds to an exponentially small perturbation from the uniform distribution $\PC_{\rm unif}$. Combining this observation with Lemma~\ref{sup-prop:indistin-prob} we can then establish that it is hard to decide the correct hypothesis with polynomial resources. 

More explicitly, we first note that it follows from Lemma~\ref{sup-prop:indistin-prob} that 
\begin{align}
 {\rm Pr}\left({\rm ``right \, decision \, between \,} \HC_{s} \, \text{and} \, \HC_0" \big| \Delta_p^{(+)}= s, \Delta_p^{(+)}= s' \right)
    \leq &  \left( \frac{1}{2} + \frac{N(s+s')}{4} \right) \, .
\end{align}
For $s,s' \in \OC(c'^{-n})$ for some $c' > 1$ and $N \in \OC(\poly(n))$ we have $\epsilon  \in \OC(c'^{-n})$ as claimed. It remains to determine with what probability we have $s,s' \in \OC(c'^{-n})$. 
By invoking Proposition~\ref{prop:con-prob}, for $\pplus$ we have
\begin{align}\label{eq:appx-chebyshev}
    {\rm Pr}_{U}\left[ \left| \pplus - \frac{1}{2}\right| \geq \delta_c \right] \leq \frac{\beta}{\delta_c^2} \;,
\end{align}
with
\begin{align}
    \beta \in \OC(1/b^n) \;,
\end{align}
for some $b,b'>1$. We then choose $\delta_c = \beta^{1/4}$ and invert the inequality of Eq.~\eqref{eq:appx-chebyshev}, leading to (with $\Delta_p^{(+)} = |\pplus -1 /2|$)
\begin{align}
    {\rm Pr}_{U}\left[\Delta_p^{(+)} \leq \beta^{1/4} \right] \geq 1 - \sqrt{\beta} \; .
\end{align}
It follows that $\Delta_p^{(+)}$ takes value between $- \beta^{1/4}$ and $ \beta^{1/4}$ (which are exponentially small) with probability at least $1 - \sqrt{\beta}$ (which is exponentially close to $1$). By following the identical steps, one can show that $\Delta_p^{(-)}$ is exponentially small with high probability. This completes the proof of the first part. To show the indistinguishability of the statistical estimate to $\widehat{O}_{\text{rand}}$, we can prove by contradiction. Basically, if they were distinguishable, we could have used this as a quantity to distinguish the distributions.
\end{proof}

\subsubsection{Proof of Supplemental Corollary~\ref{sup-coro:qrp-performance}}
\begin{proof}
We first recall the training procedure of the QRP model. The loss function in Eq.~\eqref{eq:loss} can be written in a compact form as
\begin{align}
    \LC(\vec{\eta}) & = \frac{1}{ N_{\rm tr}} \sum_{l,\t} \left(f_{\t}\bigl(\vrhoa{\t}^{(l)};\bm{\eta}\bigr)
           -
           y_{\tau}^{(l)} \right)^2  \;,
\end{align}
where $N_{\rm tr} = D_{\rm tr}|\TC_{\rm tr}|$. Minimizing this loss over $\vec{\eta}$ corresponds to doing a least square fitting and has an analytical closed-form solution~\cite{fujii2017harnessing}
\begin{align}\label{eq:optimal-params}
    \vec{\eta}^* = \left( R^T R \right)^{-1} R \,\vec{y} \;,
\end{align}
where $R$ is a matrix whose elements are the appropriate QRP outputs $\langle O_k\rangle_{\t}^{(l)} $ with dimensions $N_{\rm tr} \times M$, and $\vec{y}$ is a vector corresponding to the training labels in an appropriated order of the dimension $N_{\rm tr}$. The exact forms of $R$ and $\vec{y}$ are not necessary for the proof~\footnote{and left as an exercise for the readers}. What crucial here is that for a given estimate of $R$ (done by measuring associated QRP outputs with quantum computers), estimating the optimal parameters with Eq.~\eqref{eq:optimal-params} guarantees to minimize the empirical training loss. This happens regardless whether the QRP faces the concentration. This proves the second point of the corollary.

In the presence of concentration, Supplemental Theorem~\ref{sup-thm:stat-indis-qrp-output} indicates that a single attempt to estimate a QRP output with polynomial measurement shots is statistically indistinguishable from $\widehat{O}_{\rm rand}$ in Eq~\eqref{eq:rand-output} with high probability exponentially close to $1$. To prove the first and the last points of the corollary, it is sufficient to show that the statistical indistinguishability pertains with a polynomial number of attempts. This is because estimating optimal parameters and making prediction requires at most polynomial number of QRP outputs.

More concretely, consider estimating the total of $N_a$ QRP outputs (that exponentially concentrate) with $N_a \in \OC(\poly(n))$, and denote $o_k$ as a statistical estimate of the $k^{\rm th}$ output (with polynomial measurement shots). Let $E_k$ be the event that $o_k$ is statistically indistinguishable from $\widehat{O}_{\rm rand}$. From Supplemental Theorem~\ref{sup-thm:stat-indis-qrp-output}, we have
\begin{align}\label{eq:appx-proof-coro-1}
    {\rm Pr}[E_k] \geq 1 - \delta_\kappa \;\;; \; \forall o_k  \;,
\end{align}
with $\delta \in \OC(c^{-n})$ with some positive constant $c>1$ and the total qubits $n$. Now the probability that all $E_k$ occurs can be bounded using the union bound as
\begin{align}
    {\rm Pr}\left[ \bigcap_{k} E_k \right] & =  1 - {\rm Pr}\left[ \bigcup_k \bar{E}_k \right]  \\
    & \geq  1 - \sum_{k=1}^{N_a} {\rm Pr}\left[\bar{E}_k \right] \\
    & \geq 1 - \frac{N_a\delta}{2} \;,
\end{align}
where $\bar{E}_k$ is a conjugate event of $E_k$, we use the union bound in the second line and use ${\rm Pr}[\bar{E}_k] \leq \delta$ by reversing the final inequality in Eq.~\eqref{eq:appx-proof-coro-1}. Since $N_a \in \OC(\poly(n))$, we have that the probability that each of the estimates are statistically indistinguishable is $1 - \delta'$ with $\delta' := N_a\delta_k/2 \in \OC(\tilde{c}^{-n})$ for some $\tilde{c} >1$. 

Crucially, since the estimates are all indistinguishable, post-processing them also leads to the statistical indistinguishability. That is, for some map $\Phi(\cdot)$, we have that $\Phi \left(\{o_k\}_{k=1}^{N_a} \right)$ is also indistinguishable from $\Phi \left(\{o_{{\rm rand};k}\}_{k=1}^{N_a} \right)$ where $\{o_{{\rm rand};k}\}_{k=1}^{N_a}$ are different instances of $\widehat{O}_{\rm rand}$. Hence, by replacing QRP outputs with the instances of $\widehat{O}_{\rm rand}$ in Eq.~\eqref{eq:optimal-params} and the model prediction in Eq.~\eqref{eq:qrp-model-prediction}, with high probability we have that the optimal parameters are indistinguishable from
\begin{align}
        \vec{\eta}_{\rm rand} = \left(R_0^T R_0\right)^{-1}R_0^T \, \vec{y} \;,
\end{align}
where $R_0$ is an input-independent random matrix of the dimensions $N_{\rm tr}\times M$ whose elements are instances of $ \widehat{O}_{\text{rand}}$, and the model prediction on unseen data is indistinguishable from
\begin{align}
        \hat{f}_{\rm rand}  = \vec{O}_{\rm rand}^T \,\vec{\eta}_{\rm rand} \;,
\end{align}
where $\vec{O}_{\rm rand}$ is an input-independent vector of length $M$ whose elements are instances of $ \widehat{O}_{\text{rand}}$. This completes the proof.
\end{proof}
\clearpage

\newpage
\section{Role of noise in output concentration and memory retention in QRP}
In this appendix, we analyse two distinct noise sources in QRP:  
(i) {input-data} encoding noise and  
(ii) noise acting during reservoir evolution.  We derive recursive relations in the time step $t$, and obtain, respectively, an \emph{output-concentration bound} induced by noisy input-data encodings, and a \emph{memory-forgetting bound} that quantifies how noisy reservoir dynamics washes out information about early inputs and the initial state.

\subsection{Preliminaries on noise channels}
We consider \emph{local} noise, which can be modelled by single-qubit channels. In general, any single-qubit channel $\mathcal{N}$ admits the following decomposition~\cite{king2001minimal, ben2013quantum}
\begin{align}
    \mathcal{N}= \mathcal{U}\circ\mathcal{M} \circ \mathcal{V},
\end{align}
where $\mathcal{U}= U(\cdot)U^\dag$ and 
 $\mathcal{V}= V(\cdot)V^\dag$ are unitary channels and $\mathcal{M}$ is a channel in \emph{normal form}, i.e. a quantum operation that captures the non-unitary components of the noise. Given a single-qubit state $\sg \coloneqq \frac{1}{2}\left(I + \sum_{P\in\{X,Y,Z\}} r_P P\right)$, we have
\begin{align}
    \mathcal{M}\left(\sg\right) = \frac{1}{2}\left(I + \sum_{P\in\{X,Y,Z\}} (D_P r_P + t_P)P \right).
\end{align}
The vectors $\mathbf{D} = (D_X, D_Y, D_Z)$ and $\mathbf{t} = (t_X, t_Y, t_Z)$ are called \emph{normal form parameters}.
A quantum channel is \emph{unital} if it maps the maximally mixed state onto itself, and otherwise it is called \emph{non-unital}. Specifically, a single-qubit channel is unital if it satisfies $\mathbf{t} = (0,0,0)$. 

An important property of unital noise is that the associated normal form channel $\mathcal{M}$ can be expressed as a Pauli channel, i.e. as a probabilistic mixture of Pauli unitaries, i.e. 
\begin{align}
    \mathcal{M}(\cdot) = \sum_{P\in \{I,X,Y,Z \}} p_P P(\cdot)P \label{eq:pauli-noise},
\end{align}
where $p_P\geq 0$ and $\sum_{P\in \{I,X,Y,Z \}} p_P = 1$. A channel $\mathcal{M}$ expressed as in Eq.~\eqref{eq:pauli-noise} is commonly referred as \emph{Pauli noise} in the literature.

Unital channels can be further classified in \emph{depolarizing-like} and \emph{dephasing-like}. Depolarizing-like noise satisfies $\max_P \abs{D_P} < 1$, while dephasing-like noise satisfies $\min_P \abs{D_P} < \max_P \abs{D_P} = 1$.
Importantly, depolarizing-like noise drives any state towards the maximally mixed state, thereby increasing its entropy. In contrast, dephasing-like noise and non-unital noise do not necessarily increase the entropy of a quantum state. In order to show that generic local noise can cause QRP to lose memory of initial states or inputs, we follow two distinct strategies:
\begin{itemize}
    \item for \textbf{depolarizing-like} noise, we will employ entropy accumulation tools, such as those introduced in Refs.~\cite{muller2016relative, hirche2022contraction, wang2020noise},
    \item for \textbf{dephasing-like} and \textbf{non-unital} noise, we will study the effect of noise on \emph{typical} reservoirs, such as those composed by local Haar-random unitaries. Our analysis extends to quantum reservoir processing the techniques introduced in Refs.~\cite{quek2022exponentially, mele2024noise, angrisani2025simulating} in the context of random quantum circuits.
\end{itemize}

\subsection{Noisy input-data encoding induces concentration in the output}\label{sec:conc_noisy_enc}
In this section, we additionally study the scenario where classical data encoding strategy is noisy, generalizing the concentration upper bound induced by noisy embedding from QELM~\cite{xiong2025fundamental} to QRP with arbitrary number of steps. Similar to the noise model we discussed for QELMs~\cite{xiong2025fundamental}, here we also consider a $L$-layered encoding subject to local noise for each time step. Then the embedded state will be
\be\label{noisy_encoding}
\rhoa{\t} = \NC^{\otimes n} \circ \UC_{L}(\sv_{\t,L}) \circ \NC^{\otimes n}  \circ...\circ \NC^{\otimes n}  \circ \UC_1(\sv_{\t,1})\circ \NC^{\otimes n} (\rho_0) \,,
\ee
where $\UC_i(\sv_\t)(\cdot) = U_i(\sv_{\t,i})(\cdot)U_i(\sv_{\t,i})\ad$ and $\{\sv_{\t,i}\}_{i=1}^L$ are some functions of $\sv_\t$.
We assume the noise to be depolarizing-like, i.e. the associated normal form parameters $\mathbf{D}$ and $\mathbf{t}$ satisfy
\begin{align}
    &\max_P \abs{D_P} < 1
    \\& t_X = t_Y = t_Z = 0.
\end{align}
Finally let us define the characteristic noise parameter $q= \max \{|D_X|, |D_Y|, |D_Z|\}$. 
We will make use of the following lemma.
\begin{lemma}[Sandwiched 2-Rényi relative entropy contraction, Supplementary Lemma 6 in Ref.~\cite{wang2020noise} and Corollary 5.6 in Ref.\ \cite{hirche2020contraction}]
\label{lem:wang}
Let $\mathcal{N}$ be a single-qubit noise channel with normal form parameters $\mathbf{D}, \mathbf{t}$ as defined above. Assume that $\mathbf{t}=(0,0,0)$ (i.e. $\mathcal{N}$ is unital) 
\begin{align}
    S_2\left(\mathcal{N}^{\otimes n}(\rho) \bigg\| \frac{I^{\otimes n}}{2^n} \right) \leq q^{b} S_2\left( \rho \bigg\| \frac{I^{\otimes n}}{2^n} \right)
\end{align}
where $q\coloneqq \max_P \abs{D_P}$, $b= 1/\ln2$ and $S_2(\cdot\lVert\cdot)$ denotes the sandwiched 2-Rényi relative entropy.
\end{lemma}

Then, the following theorem holds true, which is numerically demonstrated in Fig.~\ref{fig:conc_noisy_enc}.
\begin{supplemental_theorem}[Noisy encoding induced concentration]
    Suppose a $L$-layered encoding as defined in Eq. \eqref{noisy_encoding}, with $q<1$. Then, for each $t \in \mathbb{N}$, the expectation value of the observable will concentrate towards an input-independent quantity $\mu_t = \Tr[\Omega_{t}O]$ where $\Omega_{t} = U\big(\IC/2^{n_a} \otimes \Tr_a[\Omega_{t-1}]\big)U\ad$ with $\Omega(0)=\rhot{0}$:
    \be
        |\expval{O}_t - \mu_t| \leq \norm{O}_\infty \sqrt{t} \biggr[2\ln{2}q^{b(L+1)}S_2\biggr(\rho_0\biggr\lVert\dfrac{\IC_{a}}{2^{n_a}}\biggr)\biggr]^{1/2}\,,
    \ee
    where $b= 1/\ln{2}$, and $S_2(\cdot\lVert\cdot)$ denotes the sandwiched 2-Rényi relative entropy.
\end{supplemental_theorem}
\begin{figure}[ht]
    \centering
    \includegraphics[width=0.48\textwidth]{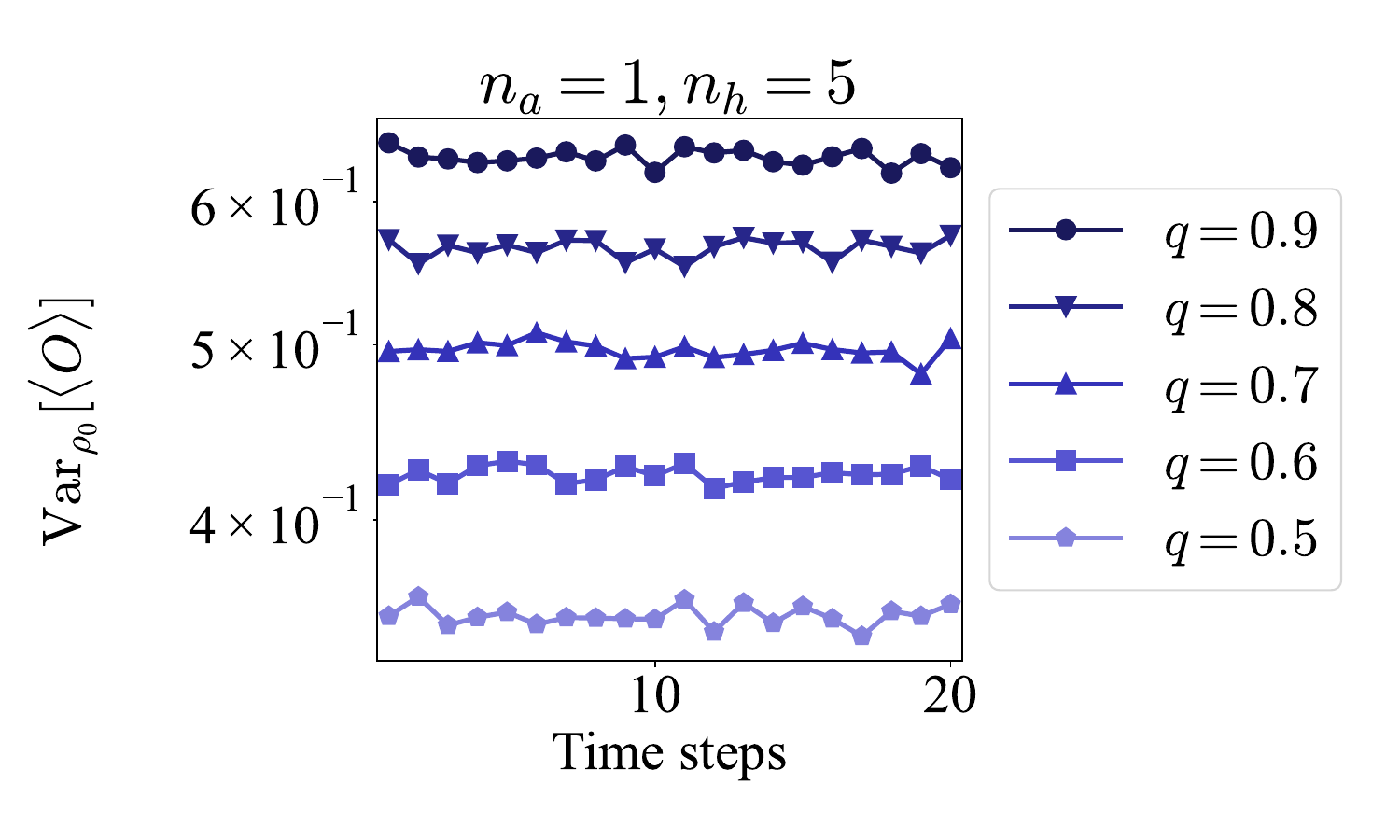}
    
    \caption{\textbf{Noisy encoding induced concentration.} The variance of outputs over randomly selected input series are plotted.
    }
    \label{fig:conc_noisy_enc}
\end{figure}
This upper bound indicates that if the number of encoding layers  $L \in \OC(n_{a})$, then the deviation of the observable from an input-independent value will concentrate exponentially with respect to the number of the accessible qubits.

\begin{specialproof}
We compute the deviation of the expectation value for time-step $t$,
\begin{align}
    |\expval{O}_{t} -\mu_{t}| &= \left|\Tr[\biggr(\rhot{t}-\Omega_{t}\biggr)O]\right| \\
    &\leq \norm{\rhot{t}-\Omega_{t}}_1\norm{O}_\infty \\
    &\leq  \norm{O}_\infty \sqrt{2\ln{2}} S_2\biggr(\rhot{t}\biggr\lVert\Omega_{t}\biggr)^{1/2}\,, \label{Eq:Noise_B13}
\end{align}
where the first inequality is achieved by using the Hölder's inequality, in the second inequality we combine Pinsker's inequality and $S(\cdot\lVert\cdot)\leq S_2(\cdot\lVert\cdot)$. Now, we find the recursive relation for the term $S_2\biggr(\rhot{t}\biggr\lVert\Omega_{t}\biggr)$,
\begin{align}
    S_2\biggr(\rhot{t}\biggr\lVert\Omega_{t}\biggr) &= S_2\biggr(\rhoa{t}\otimes\rhoh{t-1}\biggr\lVert\frac{\IC_{a}}{2^{n_a}}\otimes \Tr_a[\Omega_{t-1}]\biggr)\\
    &= S_2\biggr(\rhoa{t}\biggr\lVert\frac{\IC_{a}}{2^{n_a}}\biggr)+S_2\biggr(\rhoh{t-1}\biggr\lVert \Tr_a[\Omega_{t-1}]\biggr)\\
    &\leq S_2\biggr(\rhoa{t}\biggr\lVert\frac{\IC_{a}}{2^{n_a}}\biggr)+S_2\biggr(\rhot{t-1}\biggr\lVert \Omega_{t-1}\biggr)\,,
\end{align}
where we use the unitary invariance of relative entropy in the first equality and the additivity in the second equality. The last inequality is obtained using data processing inequality. Next, by applying this recursive relation $t$-times we obtain
\begin{align}
     S_2\biggr(\rhot{t}\biggr\lVert\Omega_{t}\biggr) &\leq \sum_{i=1}^t S_2\biggr(\rho_i\biggr\lVert\frac{\IC_{a}}{2^{n_a}}\biggr) \\
     &\leq tq^{b(L+1)} S_2\biggr(\rho_0\biggr\lVert\frac{\IC_{a}}{2^{n_a}}\biggr)\,,\label{Eq:Noise_B18}
\end{align}
where in the last inequality we apply the noise inequality in Lemma~\ref{lem:wang} for all $1 \leq i \leq t$. Finally, by combining Eq.~\eqref{Eq:Noise_B13} and \eqref{Eq:Noise_B18}, we obtain

\begin{align}
     |\expval{O}_t - \mu_t| \leq \norm{O}_\infty \sqrt{t} \biggr[2\ln{2}q^{b}S_2\biggr(\rho_0\biggr\lVert\dfrac{\IC_{a}}{2^{n_a}}\biggr)\biggr]^{1/2}\,.
\end{align}
\end{specialproof}
\subsection{Memory erasure induced by reservoir with depolarizing-like noise}\label{sec:conc_noisy_res_proof}
We now prove that adding noisy layers between each reservoir iteration can also erase memory of the initial condition and early inputs exponentially quickly.
\begin{customthm}{3}[Memory erasure induced by Pauli noise in the reservoir]
    Consider a QRP with noisy layers between reservoir evolutions as defined in Eq.~\eqref{noisy_encoding}, with $q<1$. Then, for each $t\in\mathbb{N}$, the expectation value of the observable will become independent of the initial condition and early inputs, exponentially in number of steps $t$. Formally, given two initial states of hidden qubits $\rhoh{0}$ and $\sgh{0}$, given two inputs $\rhoa{1}$ and $\sga{1}$ at the first time step, given the same series of input states $[\rhoa{k}]_{k=2}^t=[\sga{k}]_{k=2}^t$ from the second step, then the deviation of the corresponding outputs
    \begin{align}
        &\Delta O\leq \norm{O}_\infty \sqrt{2\ln{2}q^{b(t-1)}}\left(S_2\biggr(\rhoa{1}\lVert\sga{1}\biggr)+S_2\biggr(\rhoh{0}\biggr\lVert\sgh{0}\biggr)\right)^{1/2} \,,
    \end{align}
    where $\Delta O=|\expval{O}^{(\rhoh{0},\rhoa{1})}_t - \expval{O}^{(\sgh{0},\sga{1})}_t|$, $b= 1/\ln{2}$ and $S_2(\cdot\lVert\cdot)$ denotes the sandwiched 2-Rényi relative entropy.
\end{customthm}
\begin{specialproof}
We similarly upper bounds the deviation of the expectation values using 2-Rényi relative entropy,
\begin{align}
    \Delta O\ &= \left|\Tr[\biggr(\rhot{t}-\sgt{t}\biggr)O]\right| \\
    &\leq \norm{\rhot{t}-\sgt{t}}_1\norm{O}_\infty \\
    &\leq  \norm{O}_\infty \sqrt{2\ln{2}} S_2\biggr(\rhot{t}\biggr\lVert\sgt{t}\biggr)^{1/2}\,, \label{Eq:Noise_reservoir_evo}
\end{align}
where the first inequality is achieved by using the Hölder's inequality, in the second inequality we combine Pinsker's inequality and $S(\cdot\lVert\cdot)\leq S_2(\cdot\lVert\cdot)$. Next, we find the recursive relation for the term $S_2\biggr(\rhot{t}\biggr\lVert\sgt{t}\biggr)$,
\begin{align}
    S_2\biggr(\rhot{t}\biggr\lVert\sgt{t}\biggr) &= S_2\biggr(\rhoa{t}\otimes\NC\left(\rhoh{t-1}\right)\biggr\lVert\rhoa{t}\otimes \NC\left(\sgh{t-1}\right)\biggr)\\
    &=S_2\biggr(\rhoa{t}\lVert\rhoa{t}\biggr)+S_2\biggr(\NC\left(\rhoh{t-1}\right)\biggr\lVert\NC\left(\sgh{t-1}\right)\biggr)\\
    &=0+q^{b}S_2\biggr(\rhoh{t-1}\biggr\lVert\sgh{t-1}\biggr)\\
    &\leq q^{b}S_2\biggr(\rhot{t-1}\biggr\lVert\sgt{t-1}\biggr)\,.
\end{align}
By applying this recursive relation $(t-1)$ times, we obtain
\begin{align}
    S_2\biggr(\rhot{t}\biggr\lVert\sgt{t}\biggr)&\leq q^{b(t-1)}S_2\biggr(\rhot{1}\biggr\lVert\sgt{1}\biggr)\\
    &=q^{b(t-1)}S_2\biggr(\rhoa{1}\otimes\rhoh{0}\biggr\lVert\rhoa{1}\otimes \sgh{0}\biggr)\\
    &=q^{b(t-1)}\left(S_2\biggr(\rhoa{1}\lVert\sga{1}\biggr)+S_2\biggr(\rhoh{0}\biggr\lVert\sgh{0}\biggr)\right)
\end{align}
Finally, by combining the last inequality with Eq.~\eqref{Eq:Noise_reservoir_evo}, we obtain the statement.
\end{specialproof}

\subsection{Memory erasure induced by reservoir with arbitrary noise}
Let $\mathcal{N}$ be an arbitrary single-qubit noise channel. We assume that the reservoir $\mathcal{U}$ is composed of a sequence of random unitary channels interleaved by local noise of the form $\mathcal{N}^{\otimes n}$:
\begin{align}\label{Eq:non_unital_res_app}
    \mathcal{U} =  \mathcal{U}_L \circ \mathcal{N}^{\otimes n} \circ \mathcal{U}_{L-1}\circ \dots \circ \mathcal{N}^{\otimes n}  \circ \mathcal{U}_1 \circ \mathcal{N}^{\otimes n} .
\end{align}
We define the average contraction coefficient as
\begin{align}\label{Eq:avg_contr}
    \chi \coloneqq \sqrt{\frac{1}{3}\left(\norm{\mathbf{D}}_2^2 +\norm{\mathbf{t}}_2^2 \right)},
\end{align}
where $\mathbf{D}$ and $\mathbf{t}$ are the normal form parameters defined in Section II of Ref.~\cite{angrisani2025simulating}.
Notably, $\eta < 1$ if and only if $\mathcal{N}$ is non-unitary, as shown in Lemma 18 in Ref.~\cite{mele2024noise}.
While our analysis encompasses all non-unitary local noise channels, we will make some additional assumptions on the reservoir dynamics.
\begin{itemize}
    \item \textbf{Random circuit layers.} We assume that the unitary $\mathcal{U}_j$ are independently sampled from a locally scrambled distribution $\mathcal{D}$ up to the second moment, i.e. for any pair of Hermitian operators $O,O'$ we have
\begin{align}
    \mathbb{E}_{\mathcal{U}_j \sim \mathcal{D}} \left\{\mathcal{U}_j^{\otimes 2}\left(O \otimes O' \right)\right\}
    = \mathbb{E}_{\mathcal{U}_j \sim \mathcal{D}} \mathbb{E}_{V_1, V_2, \dots , V_{n_h} \sim \mathrm{Cl}(2)}
    \left\{\mathcal{U}_j^{\otimes 2}\left( \bigotimes_{j=1}^{n_h} V_j^{\otimes 2}\left[O \otimes O'\right]\bigotimes_{j=1}^{n_h} V_j^{\dag\otimes 2} \right)\right\},
\end{align}
where $V_1, V_2,\dots, V_{n_h}$ are single-qubit unitaries sampled independently and uniformly from the Clifford group.
In other terms, each layer $\mathcal{U}_j$ is invariant under post-processing by randomly sampled single-qubit Clifford gates, up to the second moment.
In particular, this assumptions is satisfied by Brickwork circuits with Haar-random gates.
\item \textbf{Linear-depth reservoir.} We assume that the depth $L$ in Eq.~\eqref{Eq:non_unital_res_app} is at least linear in system size, i.e. $L \in \Omega(n_h + n_a)$.
\end{itemize}

As a preliminary step, we restate a result from Ref.~\cite{angrisani2025simulating} within our notation.
\begin{lemma}[Adapted from Theorem 7 in Ref.~\cite{angrisani2025simulating}]
Let $O$ be a traceless Hermitian operator independent of $\mathcal{U}$. Then
    \begin{align}
        \mathbb{E}_{\mathcal{U}} \norm{\mathcal{U}^\emph{\dag}(O)}_2 \leq  \chi^L \norm{O}_2.
    \end{align}
\end{lemma}
The lemma above states that the Schatten 2-norm of an operator \( O \) is suppressed by the noisy reservoir in the adjoint picture, on average over the random unitary layers. However, a potential issue arises in the analysis of reservoir dynamics, where we often deal with observables \( O_{\mathcal{U}} \) that depend on the reservoir. As a result, the above statement is not immediately applicable. To address this, we derive a refined version that accounts for operators that explicitly depend on the reservoir.  

\begin{supplemental_corollary}
\label{cor:non-unital}
It holds that
    \begin{align}
        \mathbb{E}_{\mathcal{U}}\left\{\max_{H} \frac{\norm{ \mathcal{U}(H)}_1}{\norm{H}_2} \right\}\leq   (d_hd_a)^{\frac{3}{2}}\chi^{L},
    \end{align}
where the maximization is over traceless Hermitian operators.
\end{supplemental_corollary}
\begin{proof}
Given a Hermitian operator $H$ and a channel $\mathcal{U}$, we expand $\mathcal{U}(H)$ in the Pauli basis:
    \begin{align}
        \mathcal{U}(H) = \frac{1}{d_hd_a}\sum_{P \in \{I,X,Y,Z\}^{\otimes (n_h + n_a)}} \Tr[\mathcal{U}(H)P] P
        = \frac{1}{d_hd_a}\sum_{P \in \{I,X,Y,Z\}^{\otimes (n_h+n_a)}\setminus \{I^{\otimes (n_h+n_a)}\}} \Tr[H\mathcal{U}^\dag(P)] P,
    \end{align}
where we noted that $\Tr[I^{\otimes (n_h+n_a)} \mathcal{U}(H)] = \Tr[\mathcal{U}(H)] = \Tr[H] = 0$ since $\mathcal{U}$ is trace-preserving.
We can upper bound the trace norm $\norm{\mathcal{U}(H)}_1$ via Minkowski's inequality:
\begin{align}
    \norm{\mathcal{U}(H)}_1 \leq &\frac{1}{{d_hd_a}}\sum_{P \in \{I,X,Y,Z\}^{\otimes (n_h+n_a)}\setminus \{I^{\otimes (n_h+n_a)}\}}\abs{\Tr[H\mathcal{U}^\dag(P)]} \norm{P}_1 \\ \leq &\sum_{P \in \{I,X,Y,Z\}^{\otimes (n_h+n_a)}\setminus \{I^{\otimes (n_h+n_a)}\}}\abs{\Tr[H\mathcal{U}^\dag(P)]}
    \\ \leq & \sum_{P \in \{I,X,Y,Z\}^{\otimes (n_h+n_a)}\setminus \{I^{\otimes (n_h+n_a)}\}}\norm{\mathcal{U}^\dag(P)}_2\norm{H}_2
\end{align}
We obtain the desired result taking the maximum over $H$ and averaging over $\mathcal{U}$.
\begin{align}
    \mathbb{E}_{\mathcal{U}}\left\{\max_{H} \frac{\norm{ \mathcal{U}(H)}_1}{\norm{H}_2} \right\}
    &\leq \mathbb{E}_{\mathcal{U}}\max_{H}\frac{\norm{H}_2}{\norm{H}_2} \left\{ \sum_{P \in \{I,X,Y,Z\}^{\otimes (n_h + n_a)}\setminus \{I^{\otimes (n_h + n_a)}\}}\norm{\mathcal{U}^\dag(P)}_2\right\}\\ 
    &\leq (d_h d_a)^2 \mathbb{E}_{\mathcal{U}} \norm{\mathcal{U}^\dag(P)}_2 \leq (d_h d_a)^2 \chi^L \norm{P}_2 \\
    &= (d_h d_a)^{\frac{3}{2}} \chi^L.
\end{align}
\end{proof}
Repeatedly applying Supplementary Corollary~\ref{cor:non-unital}, we found that average-case noisy reservoirs lead to memory erasure that is exponentially fast in both the system size and the number of iterations. 
\begin{customthm}{4}[Memory erasure induced by non-unital noise in the reservoir]
Assume that the reservoir $\mathcal{U}$ consists in a random circuit interspersed by local noise as in Eq.~\eqref{Eq:non_unital_res_app}, and let $\chi$ with average contraction coefficient defined in Eq.~\eqref{Eq:avg_contr}.
Let the reservoir depth $L$ be at least
\begin{align}
    L \geq \frac{\frac{3}{2}(n_a + n_h) + \log_2(1/\epsilon) + \log(1/\delta)}{\log_2(1/\chi)}.
\end{align}
Given two initial states of hidden qubits $\rhoh{0}$ and $\sgh{0}$, given two inputs $\rhoa{1}$ and $\sga{1}$ at the first time step, given the same series of input states $[\rhoa{k}]_{k=2}^t=[\sga{k}]_{k=2}^t$ from the second step, the corresponding output states satisfies
\begin{align}
    \norm{\rhot{t} - \sgt{t}}_1 \leq \epsilon^{t-1} \norm{\rhoa{1} - \sga{1}}_1,
\end{align}
with probability at least $1-\delta$ over the randomness of the reservoir.
In particular, if $L \geq c\cdot n$ for a sufficiently large constant $c>0$, then with obtain
\begin{align}
    \norm{\rhot{t} - \sgt{t}}_1 \in \exp\left(-\Omega((n_a+n_h)t)\right),
\end{align}
with probability at least $1 - \exp(-\Omega(n_a+n_h))$.
\end{customthm}
\begin{proof}
First, we use Markov's inequality to derive a probabilistic statement from Corollary~\ref{cor:non-unital},
\begin{align}
    \Pr_{\mathcal{U}}\left\{ \forall H : {\norm{ \mathcal{U}(H)}_1}\leq \epsilon {\norm{H}_1} \right\} \geq 1-\frac{(d_hd_a)^{\frac{3}{2}}\chi^L}{\epsilon} \geq 1-\delta,
\end{align}
where last inequality follows from the fact that $L \geq \frac{\frac{3}{2}(n_a + n_h) + \log_2(1/\epsilon) + \log(1/\delta)}{\log_2(1/\chi)}.$
In the following of this proof, we condition on the event $E \coloneqq \left\{ \forall H: {\norm{ \mathcal{U}(H)}_1} \leq \epsilon {\norm{H}_1}\right\}$ happening.

Thus, we can find the recursive relation for the term $\norm{\rhot{t} - \sgt{t}}_1$,
\begin{align}
     \norm{\rhot{t} - \sgt{t}}_1 &= \mathbb{E}_{\mathcal{U}}\norm{\mathcal{U}(\rhoa{t}\otimes\rhoh{t-1} - \rhoa{t}\otimes \sgh{t-1})}_1\\
    &\leq \epsilon \norm{\rhoa{t}\otimes\rhoh{t-1} - \rhoa{t}\otimes \sgh{t-1}}_1\\
    & =  \epsilon \norm{\rhoh{t-1} - \sgh{t-1}}_1\\
    &\leq \epsilon \norm{\rhot{t-1} - \sgt{t-1}}_1.
\end{align}
By applying this recursive relation $t-1$ times, we obtain the desired result
\begin{align}
    \norm{\rhot{t} - \sgt{t}}_1  \leq \epsilon \norm{\rhot{t-1} - \sgt{t-1}}_1
     \leq \epsilon^2 \norm{\rhot{t-2} - \sgt{t-2}}_1 \leq \dots  \leq \epsilon^{t-1} \norm{\rhot{1} - \sgt{1}}_1.
\end{align}
\end{proof}

\section{Memory indicators, and relations to standard metrics for reservoir computing}\label{appendix:memory}
In this appendix, we provide a pedagogical discussion of our \emph{memory indicators} $\MC_a$ and $\MC_h$ for quantum reservoir processing (QRP), and relate these quantities to the common metrics used in reservoir computing literature, namely the \emph{Echo State Property} (ESP) and the \emph{Fading Memory Property} (FMP). However, we emphasize that our definitions capture these properties \emph{on average}, and \emph{not} in the deterministic or topological sense as in mathematical foundations of reservoir computing \cite{gonon2021universal, grigoryeva2018echo}. These quantities are thus better viewed as \emph{FMP-like} and \emph{ESP-like} characteristics, which nonetheless provide quantitative measures of how old inputs and initial reservoir states become irrelevant, due to \emph{exponential concentration}, in typical scenarios for QRP with finite measurement shots. While these concentration-induced memory characteristics resonate with the essence of FMP or ESP in mathematical RC, we do not claim a strict equivalence to the mathematical definitions used in formal universality results. Finally, we discuss the effects that exponential concentration have on reservoir universality.

\subsection{Background: echo state and fading memory properties in reservoir computing}

In \emph{classical} Reservoir Computing (RC) \cite{jaeger2001echo,maass2002real,lukosevicius2009reservoir, yidiz2012echo},
one uses a fixed dynamical system (the \emph{reservoir}) to process sequential data (time-series) $\{s_t\}$. At each time step $t$, the reservoir's hidden state $\mathbf{x}_t$ evolves under the current input $s_t$ and the previous state $\mathbf{x}_{t-1}$ through a fixed (nonlinear) map $f(s_t, \mathbf{x}_{t-1}) = \mathbf{x}_t $. Only a simple \emph{readout} (often, for simplicity, a linear map on $\mathbf{x}_t$) is trained to produce the output $y_t$. 

Two fundamental dynamical properties, informally stated, that are necessary for classical reservoir computing to be well-defined (though by themselves insufficient to guarantee strong predictive performance) are:

\begin{enumerate}

    \item \emph{Fading Memory Property (FMP):} A reservoir has FMP if the influence of inputs fed \emph{far in the past} on the present output \emph{fades away} in time. In other words, perturbations in the input $s_{t-\tau}$ for large $\tau$ cause only a negligible difference in the current output $y_t$. FMP ensures that the reservoir focuses on the relevant recent history of the input and is not overwhelmed by far past information.

    \item \emph{Echo State Property (ESP):} A reservoir has ESP if its internal state $\mathbf{x}_t$ at large $t$ depends only on the driving inputs $\{s_t\}$ and \emph{not} on the reservoir's initial state $\mathbf{x}_0$. ESP thus guarantees that, after a long washout (burn-in) period, the reservoir state becomes a function of only recent inputs, and not on how the reservoir was initially prepared. Without ESP, the output could spuriously depend on arbitrary past state preparations, making the system ill-defined for consistent time series processing.
\end{enumerate}
When both properties are satisfied, the reservoir map $f$ defines a \emph{causal time-invariant filter} on semi-infinite inputs\footnote{
In many treatments, inputs $\{s_t\}
$ are defined for all integer time $t \le 0$, called a \emph{left-infinite} or \emph{semi-infinite} input sequence. Intuitively, we let the reservoir receives the input data from $-\infty$ up to the present so that any dependence on an arbitrarily chosen initial reservoir state $\mathbf{x}_{-\infty}$ is washed out at large times. This ensures well-defined behavior of the map $f$.}, guaranteeing that outputs are primarily driven by recent inputs \cite{grigoryeva2018echo}. Moreover, by optimizing the readout layer (via linear regression) and by properly choosing a reservoir (such as a sufficiently large Echo State Network (ESN)), certain systems with such properties can achieve a \emph{universal approximation} of fading memory maps and possess short-term sequence-modeling capabilities \cite{grigoryeva2018echo, gonon2021universal}.

\smallskip
\noindent
\textbf{How $\MC_a$ and $\MC_h$ in our QRP framework relate to FMP and ESP.}
While mathematical RC defines FMP and ESP using rigorous deterministic statements~\cite{grigoryeva2018echo}, our memory indicators $\MC_a$ and $\MC_h$ are \emph{ensemble-averaged} quantities. We consider how quickly the variance of outputs decays under perturbations (in either past inputs or initial states). This approach captures an \emph{operational} sense of early input forgetting and early state forgetting in the \emph{typical behavior} of a quantum reservoir. Unlike statements that universally work for each input or each initial reservoir instance in mathematical RC, our probabilistic statements reflect the common use of ensemble averaging in theoretical models (e.g., Haar random unitaries or noise channels, as used in the main text) which enables a rigorous study of exponential concentration in QRP, and the essential role of dissipation and noise in open quantum systems \cite{cheamsawat2025dissipation, sannia2024dissipation, martinez2023quantum, yosifov2025dissipation}. Thus, our results do not necessarily imply the same universal approximation theorems proven by classical references \cite{grigoryeva2018echo, gonon2021universal}, but does capture the essence of FMP-like and ESP-like behaviors relevant to QRP with finite measurement shots. The formalization of ESP and FMP in QRP is an active area of research, with recent proposals adapting classical concepts to quantum contexts \cite{kobayashi2024echopre, kobayashi2024coherence, martinez2024input, martinez2023quantum, sannia2024dissipation, gonon2025feedback}.

\subsection{Exponential decay of $\MC_a(t)$ implies Fading Memory Property (FMP)-like characterics in QRP}\label{app:memory}

\noindent
\textbf{Setting.} 
Fix an initial time step $\tau$, such that an input state $\rhoa{\t}$ is injected into the reservoir’s {\it accessible} qubits. The reservoir then evolves for $t$ more steps (with subsequent inputs $\rhoa{\t+1},\rhoa{\t+2},\dots$) until time $(\tau + t)$. Denote the reservoir’s measurement outcome at that final time by
\begin{equation}\label{eq:FMP_setting}
   \langle O\rangle_{(\tau+t)}^{(\rhoa{\t})}
   =
   \mathrm{Tr}\bigl[
       O\rhot{\tau+t}\bigr],
\end{equation}
where $\rhot{\tau+t}$ is the resulting reservoir state (accessible and hidden qubits together).

As in Definition~\ref{def:input_dep} of the main text, we define
\begin{equation}\label{eq:MaDef}
  \MC_{a}(t;\,\SC_a,U)
  =
         \mathrm{Var}_{\rhoa{\t}\sim\SC_a}
         \bigl(\langle O\rangle_{\tau+t}\bigr)\,,
\end{equation}
where $\mathrm{Var}_{\rhoa{\t}\sim\SC_a}(\cdot)$ is the variance of $\langle O\rangle_{\tau+t}$ over the distribution of possible inputs $\rhoa{\t}$ at the \emph{initial time} $\tau$. If $ \MC_a(t)$ tends to be large, then changing the initial input can  significantly influence the final output at the subsequent time step $t$. Conversely, if $ \MC_a(t)$ is likely small, the reservoir effectively \emph{forgets} that earlier input.

\begin{supplemental_theorem}[FMP-like behavior from exponentially decaying $\MC_a(t)$]
\label{thm:FMP_main}
Assume there exist constants $c>0$ and $\alpha>0$ such that
\begin{equation}\label{eq:MaExpDecay}
   \Ebb_{U}[\MC_a(t;\,\SC_a,U)] 
   \le
   c\,\mathrm{e}^{-\alpha t},
   \quad
   \forall\,t\in\mathbb{N}\,,
\end{equation}
Then the reservoir exhibits FMP-like fading of old inputs on average over random pairs of inputs. Concretely, for two different initial states $\rhoa{\t}$ and $\sga{\t}$ drawn from $\SC_a$ (with all subsequent inputs in the time series fixed), the outputs at time $(\tau+t)$ satisfy
\begin{equation}\label{eq:FMPaverage}
   \mathbb{E}_{U,\rhoa{\t},\sga{\tau}}\Bigl[
     \bigl(
        \langle O\rangle_{(\tau+t)}^{(\rhoa{\t})}
        -
        \langle O\rangle_{(\tau+t)}^{(\sga{\tau})}
     \bigr)^2
   \Bigr]
   \le
   4\,c\,\mathrm{e}^{-\alpha t}.
\end{equation}
Hence, on average, older inputs $\rhoa{\t}$ for large $\tau$ fade exponentially from the reservoir’s subsequent outputs.

\smallskip
\noindent
\emph{Implications from exponential concentrations:} 
\begin{itemize}
\item \textbf{Reservoir's dimension-based decays:} For Haar random reservoirs,  $\Ebb_{U}[\MC_a(t)]\in\OC\left( \frac{1}{\,d_{h}\,d_{a}^t}\right)$, see Theorem~\ref{thm:mem_ind_upp_bound} in the main text, and then $\alpha \in\OC(n_{a})$. 
\item \textbf{Reservoir's noise-based decays:} For local Pauli noise with $0<q<1$, one typically gets $\MC_a(t)\propto \sqrt{q}^{bt}$, see Theorem~\ref{thm:unital_noise} in the main text, so $\alpha \in \OC(1)$.  More generally, non-unital noise can also induce local-scrambling exponentials.
\end{itemize}

\end{supplemental_theorem}

\begin{proof}[Proof]
\noindent  
By \eqref{eq:MaDef} and \eqref{eq:MaExpDecay}, we have
\begin{equation}\label{eq:MaAvgVar}
  \Ebb_{U}[\MC_a(t;\SC_a,U)]
  =
     \Ebb_{U}[\mathrm{Var}_{\rhoa{\t}\sim\SC_a}\bigl(\langle O\rangle_{\tau+t}\bigr)]
  \le
  c\,\mathrm{e}^{-\alpha t}.
\end{equation}
Thus, on average over the reservoir $U$, the variance of the output with respect to the random input $\rhoa{\t}$ is at most $c\,e^{-\alpha t}$.

We next draw two i.i.d. states $(\rhoa{\t},\sga{\tau})$ from the same distribution $\SC_a$, representing two possible early inputs at time $\tau$. Our goal is to bound
\begin{equation}\label{eq:FMP2var_objective}
   \bigl(\langle O\rangle_{\tau+t}^{(\rhoa{\t})}
         - \langle O\rangle_{\tau+t}^{(\sga{\tau})}\bigr)^2
\end{equation}
\emph{on average} over both $U$ and $(\rhoa{\t},\sga{\tau})$. 

First, fix the reservoir $U$. For each possible input $\rhoa{\t}\in\SC_a$, define
\begin{equation}\label{eq:X_centering}
   X(\rhoa{\t})
   =
   \langle O\rangle_{\tau+t}^{(\rhoa{\t})} 
    -
   \mathbb{E}_{\rhoa{\t}'\sim \SC_a}\bigl[
       \langle O\rangle_{\tau+t}^{(\rhoa{\t}')}\bigr],
\end{equation}
thus \emph{centering} the random variable $\langle O\rangle_{\tau+t}^{(\rhoa{\t})}$ by subtracting its mean value (averaged over $\rhoa{\t}'\!\sim \SC_a$). Given two distinct inputs $(\rhoa{\t},\sga{\tau})$, one can rewrite the difference
\begin{equation}
  \langle O\rangle_{\tau+t}^{(\rhoa{\t})}
  -
  \langle O\rangle_{\tau+t}^{(\sga{\tau})}
  =
  \bigl(X(\rhoa{\t}) - X(\sga{\tau})\bigr)
    + \Delta_{\mu},
\end{equation}
where 
\(
  \Delta_{\mu}
   =
  \mathbb{E}_{\rho'_{\tau}\sim\SC_a}
   [\langle O\rangle_{\tau+t}^{(\rho'_{\tau})}]
  -
  \mathbb{E}_{\sg'_{\tau}\sim\SC_a}
   [\langle O\rangle_{\tau+t}^{(\sg'_{\tau})}] 
   = 0
\)
since the two initial states are drawn from the same ensemble $\SC_a$ and the rest of the sequence are identical. We will use the usual inequality $(x\pm y)^2 \le 2x^2 + 2y^2$ (see the footnote \footnote{A quick way to see the inequality is to first expand $(x\pm y)^2 = x^2 \pm 2xy + y^2.$ Then, note that
   $\pm 2xy \le x^2 + y^2$, since $0 \le (x \mp y)^2 = x^2 \mp 2xy + y^2.$
Therefore, $(x \pm y)^2 \le x^2 + y^2 + (x^2 + y^2) = 2x^2 + 2y^2$.
}).
From $(x-y)^2 \le 2x^2 + 2y^2$, one obtains
\begin{equation}\label{eq:last_step_fmp}
  \bigl(\langle O\rangle_{\tau+t}^{(\rhoa{\t})}
        -\langle O\rangle_{\tau+t}^{(\sga{\tau})}\bigr)^2
  =
  \bigl(X(\rhoa{\t}) - X(\sga{\tau})\bigr)^2
  \le
    2X(\rhoa{\t})^2 + 2X(\sga{\tau})^2.
\end{equation}

We now take expectations \emph{first} over the random pair \((\rhoa{\t},\sga{\tau})\) drawn i.i.d.\ from \(\SC_a\), and \emph{then} over the reservoir \(U\). Observe that
\begin{equation}
   \mathbb{E}_{\rhoa{\t}}\bigl[X(\rhoa{\t})^2\bigr]
   =
   \mathrm{Var}_{\rhoa{\t}\sim\SC_a}\!\bigl(
       \langle O\rangle_{\tau+t}^{(\rhoa{\t})}
     \bigr),
\end{equation}
by definition of variance. The same argument applies to $\sga{\tau}$. Hence, all terms inside that expectation in the RHS of \eqref{eq:last_step_fmp} are precisely 
\begin{align}
  \mathbb{E}_U\Bigl[
    \mathrm{Var}_{\rhoa{\t}\sim\SC_a}\!\bigl(
      \langle O\rangle_{\tau+t}^{(\rhoa{\t})}
    \bigr)
  \Bigr],
\end{align}
which is indeed $ \Ebb_{U}[\MC_a(t)]$. Therefore, taking all the expectations together on both sides of \eqref{eq:last_step_fmp} yields
\begin{equation}\label{eq:FMPfinal}
  \mathbb{E}_{U,\rhoa{\t},\sga{\tau}}
   \Bigl[
     \bigl(
       \langle O\rangle_{\tau+t}^{(\rhoa{\t})}
       -\langle O\rangle_{\tau+t}^{(\sga{\tau})}
     \bigr)^2
   \Bigr]
  \le
  4 \Ebb_{U}[\MC_a(t)] \;.
\end{equation}

Thus, on average, old inputs cause at most an exponentially small deviation in the output, once $ \Ebb_{U}[\MC_a(t)]$ is exponentially small in $t$.
\end{proof}

Note that inequality \eqref{eq:FMPfinal} implies that once $ \Ebb_{U}[\MC_a(t)]$ decays as $c e^{-\alpha t}$, any small change in the input at time $\tau$ leads to an \emph{exponentially small} difference in the outputs at time  $(\tau+t)$, \emph{on average} over the reservoir ensemble and pairs of inputs. This reflects an \emph{FMP-like} characteristic, such that old inputs are effectively forgotten at a rate set by $\alpha$. By contrast, mathematical RC typically uses a more stringent, deterministic definition of FMP that requires \emph{all} possible input variations to vanish in time \cite{grigoryeva2018echo}. In our QRP framework, the exponential concentration phenomena leads to the difference between outputs shrinking \emph{exponentially} fast in $t$. 

\subsection{Exponential decay of $\MC_h(t)$ implies Echo State Property (ESP)-like characteristics in QRP}

\textbf{Setting.}  
Let $\rhoh{0}$ be the hidden qubits’ initial state,  while keeping all subsequent inputs $\rhoa{1},\rhoa{2},\dots$ in the accessible qubits fixed.  After $t$ steps of evolution, an observable $O$ is measured, giving
\begin{equation}
   \langle O\rangle_{t}^{(\rhoh{0})}
   =
   \mathrm{Tr}\bigl[O\,\rhot{t}\bigr],
\end{equation}
where $\rhot{t}$ is the resulting reservoir state (accessible and hidden qubits together).
As in Definition~\ref{def:init_state_dep} of the main text, we define
\begin{equation}
   \MC_h(t;\SC_h,U)
   =
     \mathrm{Var}_{\rhoh{0}\sim\SC_h}
       \bigl(\langle O\rangle_{t}\bigr)\,,
\end{equation}
  Similarly to the previous subsection, if $\MC_h(t)$ is small with high probability, then $\langle O\rangle_{t}^{(\rhoh{0})}$ barely depends on how the hidden qubits were initially prepared, implying the reservoir \emph{forgets} its own initial condition. If one starts the reservoir in two distinct hidden-qubit states $\rhoh{0}$ and $\sgh{0}$ at the outset, the reservoir's outputs should eventually converge, given the \emph{same} subsequent sequence of inputs. As discussed earlier, this intuitive ESP-like requirement guarantees the reservoir is well-defined on average, since one does not want the final outputs to spuriously depend on arbitrary preparation procedures from long in the past.

\begin{supplemental_theorem}[ESP-like behavior from exponentially decaying $\MC_h(t)$]
\label{thm:ESP-from-Mh}
Suppose there exist constants $c'>0$ and $\alpha'>0$ such that
\begin{equation}
\label{eq:MhExpDecay}
  \Ebb_{U}[\MC_h(t;\SC_h,U)]
  \le
  c' e^{-\alpha' t}
  \quad
  \forall\,t\in \mathbb{N}\,,
\end{equation}
Then the reservoir exhibits an ESP-like behavior: on average, the influence of the reservoir's initial hidden state decays exponentially. Specifically, consider two runs of the reservoir with the same time-series inputs $\{\rhoa{1},\dots,\rhoa{t},\dots\}$, but different reservoir's hidden-qubit initial states $\rhoh{0}$ and $\sgh{0}$. Denote their final outcomes at time $t$ by
\begin{equation}
\label{eq:twoOutputsHidden}
\langle O\rangle_{(t)}^{(\rhoh{0})}
\quad\text{and}\quad
\langle O\rangle_{(t)}^{(\sgh{0})}.
\end{equation}
Then, on average over random draws $\,\rhoh{0},\sgh{0}\sim\SC_h$ and $\,U\sim\UC$,
\begin{equation}
\label{eq:ESPResult}
\mathbb{E}_{U,\rhoh{0},\sgh{0}}
\Bigl[
   \Bigl(
     \langle O\rangle_{(t)}^{(\rhoh{0})}
     -
     \langle O\rangle_{(t)}^{(\sgh{0})}
   \Bigr)^2
\Bigr]
\le
4 c' e^{-\alpha' t}.
\end{equation}
Hence, the two output trajectories converge exponentially quickly, irrespective of distinct hidden-state preparations in the far past. In other words, old initial conditions of the reservoir state are washed away, and the final outputs does echo the recent input signals.
\end{supplemental_theorem}

\begin{proof}[Proof]
The argument fully parallels the proof for FMP in Theorem~\ref{thm:FMP_main}. First, the assumption \eqref{eq:MhExpDecay} tells us that
\begin{equation}
 \Ebb_{U}[\MC_h(t)]
=
\mathbb{E}_{U}
\Bigl[
  \mathrm{Var}_{\rhoh{0}\sim\SC_h}
    \bigl(\langle O\rangle_{(t)}\bigr)
\Bigr]
\le
c' e^{-\alpha' t}.
\end{equation}
For any \emph{fixed} reservoir evolution $U$, define the centered random variable
\begin{equation}
Y(\rhoh{0})
\coloneqq
\langle O\rangle_{(t)}^{(\rhoh{0})}
-
\mathbb{E}_{\rhoh{0}'\sim\SC_h}
   \bigl[
     \langle O\rangle_{(t)}^{(\rhoh{0}')}
   \bigr].
\end{equation}
One sees that
$\mathbb{E}_{\rhoh{0}}[Y(\rhoh{0})] = 0$
and
$\mathbb{E}_{\rhoh{0}}[Y(\rhoh{0})^2]
 = \mathrm{Var}_{\rhoh{0}\sim\SC_h}\bigl(\langle O\rangle_{(t)}\bigr).$
Next, let us compare two distinct initial hidden states, $\rhoh{0}$ and $\sgh{0}$. Their outputs differ by
\begin{equation}
\langle O\rangle_{(t)}^{(\rhoh{0})}
  -
\langle O\rangle_{(t)}^{(\sgh{0})}
=
\bigl(Y(\rhoh{0}) - Y(\sgh{0})\bigr)
   +
\Delta_{\mu},
\end{equation}
where 
$\Delta_{\mu} \equiv
   \mathbb{E}_{\rhoh{0}'\sim\SC_h}[\langle O\rangle_{(t)}^{(\rhoh{0}')}]
     - \mathbb{E}_{\sgh{0}'\sim\SC_h}[\langle O\rangle_{(t)}^{(\sgh{0}')}]].$
Since $\rhoh{0},\sgh{0}$ are \emph{both} drawn from the same distribution $\SC_h$, then $\Delta_{\mu} = 0.$ Applying the usual inequalities $(x - y)^2 \le 2(x^2 + y^2)$ to the squared of the previous equation yields
\begin{equation}
\bigl(\langle O\rangle_{(t)}^{(\rhoh{0})}
   -
   \langle O\rangle_{(t)}^{(\sgh{0})}
\bigr)^2
\le
2\bigl[Y(\rhoh{0})^2 + Y(\sgh{0})^2\bigr].
\end{equation}
Finally, taking the expectation over the i.i.d.\ draws $\rhoh{0},\sgh{0} \sim \SC_h$ and over $U \sim \UC$, we get
\begin{equation}
\begin{aligned}
\mathbb{E}_{U,\rhoh{0},\sgh{0}}
\Bigl[
  \bigl(
    \langle O\rangle_{(t)}^{(\rhoh{0})}
    -
    \langle O\rangle_{(t)}^{(\sgh{0})}
  \bigr)^2
\Bigr]
&\le
2\,
\mathbb{E}_{U,\;\rhoh{0},\;\sgh{0}}
\bigl[
  Y(\rhoh{0})^2 \;+\; Y(\sgh{0})^2
\bigr]
\\[6pt]
&=
4\mathbb{E}_{U}
\Bigl[
  \mathrm{Var}_{\rhoh{0}\sim\SC_h}\bigl(\langle O\rangle_{(t)}^{(\rhoh{0})}\bigr)
\Bigr]
=
4 \Ebb_{U}[\MC_h(t)]
\le
4c' e^{-\alpha' t}.
\end{aligned}
\end{equation}
This completes the proof of Eq.~\eqref{eq:ESPResult}.
\end{proof}

 Once $ \Ebb_{U}[\MC_h(t)]$ decays as $c' e^{-\alpha' t}$, any difference in the reservoir’s initial hidden state $\rhoh{0}$ results in an \emph{exponentially small} difference in outputs at time $t$ \emph{on average}, over both hidden-state pairs and the reservoir ensemble. This signifies an ESP-like characteristic: the system forgets initial conditions at a rate determined by $\alpha'$. Note that mathematical RC typically adopts a more demanding, deterministic form of ESP~\cite{grigoryeva2018echo}. It is important to emphasize again that in our QRP framework, the \emph{exponential concentration} phenomena typically induce this rapid forgetting, exponentially fast in $t$.

\subsection{Discussions on reservoir universality and the role of dissipation}

In mathematical RC, if a sufficiently rich reservoir deterministically satisfies formal definitions of ESP and FMP (often via topological continuity or uniform state convergence), one can prove \emph{universal approximation} of fading-memory filters \cite{grigoryeva2018echo,gonon2021universal}. In our QRP framework, however, the average of memory indicators i.e., $ \Ebb_{U}[\MC_a(t)]$ and $ \Ebb_{U}[\MC_h(t)]$ yield a \emph{probabilistic} statement on the exponential decay of old inputs and initial states. While this forgetting is essential for operative QRP, it does \emph{not} directly imply the deterministic criteria required by classical universality theorems. Whether there is a direct link to formal universality would require additional analysis beyond the scope of this work. Instead, our findings emphasize that, under finite measurement shots, practical QRP faces a tension between scaling up the reservoir to handle more complex memory tasks and avoiding exponential concentration of outputs.

Moreover, realistic QRP platforms harness dissipation and decoherence as \emph{resources} for RC scheme: they wash out old inputs or initial states to enforce FMP-like and ESP-like behaviors \cite{sannia2024dissipation,sannia2023engineered,innocenti2023potential,cheamsawat2025dissipation,martinez2023quantum}. {Such engineered dissipative or measurement processes can provide a task-appropriate washout timescale, whereas non-Markovian open-system dynamics can extend memory when longer temporal correlations are required. This does not conflict with our noise bounds, which identify regimes in which contraction is too rapid and erases relevant information before the task's memory horizon. However, neither statement establishes that properly designed dissipation can remove system-size-dependent concentration in temporal QRP.} 

{
\section{Physical reservoir dynamics}
{
\subsection{Summary of widely used reservoirss}\label{app:summary_reservoir}
\onecolumngrid

\begingroup
\scriptsize
\setlength{\tabcolsep}{3pt}
\renewcommand{\arraystretch}{1.45}
\begin{longtable}{
p{0.45\linewidth}
p{0.35\linewidth}
>{\centering\arraybackslash}p{0.10\linewidth}
}
\caption{Summary of reservoirs used in existing literatures.}
\label{tab:QRP_reservoirs_detail}
\\
\hline
Reservoirs & Representative Tasks & References \\
\hline
\endfirsthead

\hline
Reservoirs & Representative Tasks & References \\
\hline
\endhead

\hline
\multicolumn{3}{r}{Continued on next page.}
\\
\endfoot

\hline
\endlastfoot
\multicolumn{3}{c}{\textbf{Circuit / Gate Reservoirs}}
\\[0.4em]
\hline

\textbf{Global Haar-Random Unitary Reservoirs}

$U_{\rm res}\sim{\rm Haar}(\mathbb U(2^N))$.
&
Memory, Mackey--Glass prediction, quantum-Ising time-series prediction, feedback QRP.
&
\cite{kobayashi2024feedback,murauer2025feedback}
\\[0.7em]
\cline{1-3}

\textbf{Local Haar-Random / Block-Random Circuit Reservoirs}

$U_{\rm circ}
=
\prod_{\ell}
\prod_{(i,j)\in{\cal B}_{\ell}}
U_{ij}^{(\ell)}$,

with local two-qubit blocks, e.g.
$U_{ij}^{(\ell)}\sim{\rm Haar}(\mathbb U(4))$,
dual-unitary gates, or other fixed two-qubit blocks.
&
NARMA-family benchmarks, Mackey--Glass prediction, circuit-reservoir benchmarking.
&
\cite{malik2026evaluating}
\\[0.7em]
\cline{1-3}

\textbf{Parameterized Rotation--Entangler Circuit Reservoirs}

$U_{\rm circ}(u_t,\boldsymbol\theta)
=
\prod_{\ell}
\left[
\prod_i R_{\alpha_{\ell i}}\!\left(\theta_{\ell i}(u_t)\right)
\prod_{(i,j)\in E_{\ell}}G_{ij}
\right]$,

with
$R_{\alpha}\in\{R_x,R_y,R_z\}$ and
$G_{ij}\in\{{\rm CX},{\rm CZ},{\rm CR},{\rm iSWAP},\ldots\}$.
Equivalently, some papers use
$U_{\rm circ}(\boldsymbol\theta)=\prod_{\ell}e^{-i\theta_\ell P_\ell}$,
with Pauli strings $P_\ell$.
Noise, damping, reset, or measurement channels may be applied after the same unitary skeleton.
&
Temporal processing on noisy quantum processors, NARMA, Mackey--Glass prediction,
object classification, chaotic-system forecasting, financial time series,
channel equalization, memory enhancement, noise-as-resource studies.
&
\cite{chen2020temporal,suzuki2022natural,kubota2022quantum,domingo2023taking,fry2023optimizing,pfeffer2022hybrid,pfeffer2023reduced,wudarski2023hybrid,monzani2024non,ricci2026quantum,yasuda2023quantum,hu2024overcoming}
\\[0.7em]
\hline

\multicolumn{3}{c}{\textbf{Hamiltonian Spin / Qubit Reservoirs}}
\\[0.4em]
\hline

\textbf{Ising-like / Transverse-Field Spin Reservoirs}

$H=\sum_{(i,j)\in E}J_{ij}Z_iZ_j+\sum_i h_iX_i$.

&
Short-term memory, parity, NARMA, Mackey--Glass, spatial multiplexing,
dynamical phases, graph benchmarks, volatility/weather prediction, Hamiltonian encoding.
&
\cite{fujii2017harnessing,negoro2018machine,nakajima2019boosting,kutvonen2020optimizing,martinez2021dynamical,li2026quantum,hou2026high,ivaki2024quantum,settino2025topology,mccaul2025minimal,cindrak2024enhancing,delorenzis2025harnessing,fuchs2024quantum,xia2023configured}
\\[0.7em]
\cline{1-3}

\textbf{Heisenberg Exchange / XY Reservoirs}

$H=\sum_{(i,j)\in E}\!\left[J_{ij}^{xy}(X_iX_j+Y_iY_j)+J_{ij}^{z}Z_iZ_j\right]
+\sum_i\vec h_i\cdot\vec\sigma_i$.

&
QERC feature maps, image classification, spin-network learning,
time-series prediction, interpolation, reservoir-state representation.
&
\cite{hayashi2025effective,delorenzis2025harnessing,gili2024generating}
\\[0.7em]
\cline{1-3}

\textbf{Driven Rydberg-Atom Reservoirs}

$H(t)=\frac{\Omega(t)}{2}\sum_i\sigma_i^x-\sum_i\Delta_i(t)n_i
+\sum_{i<j}V_{ij}n_in_j$,

with $n_i=|r_i\rangle\langle r_i|$ and $V_{ij}=C_6/R_{ij}^6$.
&
Memory, decision making, image classification, Santa Fe laser prediction,
large-scale analog neutral-atom QRP.
&
\cite{bravo2022quantum,kornjaca2024large,solanki2025harnessing,liu2026practical}
\\[0.7em]

\cline{1-3}

\textbf{Fermionic Reservoirs}
SYK model 
&
Memory, NARMA
&
\cite{kobayashi2026edge}
\\[0.7em]
\hline
\multicolumn{3}{c}{\textbf{Bosonic / Hybrid Reservoirs}}
\\[0.4em]

\hline

\textbf{Quadratic / Gaussian Bosonic Reservoirs}

$H_{\rm quad}
=
\sum_i\omega_i a_i^\dagger a_i
+
\sum_{i<j}
\left(
g_{ij}a_i^\dagger a_j
+
\xi_{ij}a_i^\dagger a_j^\dagger
+
{\rm h.c.}
\right)$.

Equivalently, $\hat a_i^\dagger\mapsto\sum_jU_{ij}\hat a_j^\dagger$
for passive linear optics, or $S(t)=e^{\Omega Gt}$ for Gaussian oscillator networks.
&
Image classification, spoken-vowel classification, handwritten-digit classification,
memory capacity, information capacity, real-time QRP, nonlinear time-series prediction.
&
\cite{sakurai2025quantum,nerenberg2025photon,cimini2026large,garcia2023scalable,nokkala2021gaussian,dudas2023quantum}
\\[0.7em]
\cline{1-3}

\textbf{Nonlinear Kerr / Bose--Hubbard Reservoirs}

$H_{\rm nl}
=
\sum_i\omega_i n_i
+
\sum_{i<j}J_{ij}(a_i^\dagger a_j+{\rm h.c.})
+
\sum_i K_i n_i(n_i-1)
+
H_{\rm drive}(t)$,

with $n_i=a_i^\dagger a_i$.
This includes Kerr oscillators, JPO/Kerr networks, and Bose--Hubbard lattices.
&
Multiqubit readout, quantum-state classification, tomography, continuous parity monitoring,
sine-phase estimation, finite-sample QRP, atomic-lattice QRP, nonlinear memory.
&
\cite{angelatos2021reservoir,khan2021physical,govia2021quantum,cheamsawat2025dissipation,llodra2025quantum}
\\[0.7em]
\cline{1-3}

\textbf{Spin--Boson / Cavity-QED Reservoirs}

$H_{\rm sb}
=
\omega_c a^\dagger a
+
\sum_j\frac{\omega_j}{2}\sigma_z^{(j)}
+
\sum_jg_j(a\sigma_+^{(j)}+a^\dagger\sigma_-^{(j)})
+
H_{\rm drive}(t)$.

This includes Jaynes--Cummings, Tavis--Cummings, dispersive, and superconducting qubit--oscillator variants.
&
Analog microwave signal classification, weak-signal processing, nonlinear memory,
Mackey--Glass prediction, sine-square classification, cavity-QED QRP.
&
\cite{senanian2024microwave,burgess2022quantum,das2026quantum,zhu2025practical,zhu2025minimalistic}

\\[0.4em]
\hline
\end{longtable}
\endgroup
}

\subsection{SYK model}\label{app:syk}

\paragraph*{Model and dynamical phases.}
The Sachdev--Ye--Kitaev (SYK) model~\cite{sachdev1993gapless} is a paradigmatic model of many-body quantum chaos, saturating the bound on the growth rate of out-of-time-ordered correlators~\cite{maldacena2016bound}. It has also recently entered quantum reservoir processing: Ref.~\cite{kobayashi2026edge} operates the complex SYK model as a quantum reservoir and identifies the ``edge of many-body quantum chaos'' -- approached either in evolution time or by tuning an integrability-breaking coupling -- as the regime of optimal task performance, as part of a broader effort relating reservoir performance to dynamical phases and quantum chaos~\cite{martinez2021dynamical,llodra2025quantum,vetrano2025state}. This makes the SYK reservoir a natural testbed for the converse question studied in this work, namely how the \emph{output concentration} behaves across the same transition. We therefore adopt from Ref.~\cite{kobayashi2026edge} not only the model but also its operating point -- the coupling conventions and spectral normalization, comparable system sizes, and the same two evolution times -- so that the concentration results below can be compared with the performance results of that reference point by point. Following Eq.~(1) of Ref.~\cite{kobayashi2026edge}, the reservoir Hamiltonian on $N = n_a + n_h$ fermionic modes reads
\begin{align}\label{eq:syk_reservoir}
    H = \underbrace{\sum_{i,j,k,l=1}^{N} J_{ijkl}\, c_i^\dagger c_j^\dagger c_k c_l}_{H_{\rm SYK_4}} \;+\; \underbrace{\sum_{i,j=1}^{N} \kappa_{ij}\, c_i^\dagger c_j}_{H_{\rm SYK_2}}\;,
\end{align}
where $c_i$ are fermionic annihilation operators and the couplings are zero-mean complex Gaussian random arrays of variances $\overline{|J_{ijkl}|^2} = J_4^2/N^3$ and $\overline{|\kappa_{ij}|^2} = \kappa_2^2/(2N)$, subject to $J_{ijkl} = J_{klij}^* = -J_{jikl} = -J_{ijlk}$ and $\kappa_{ij} = \kappa_{ji}^*$, so that $H$ is Hermitian and conserves the total particle number. We fix $J_4 = 1$, normalize the Hamiltonian to unit spectral norm, $H \to H/\Vert H \Vert$ (so that the evolution time $\Delta t$ is measured in units of the inverse many-body bandwidth, as in the parametric scan of Ref.~\cite{kobayashi2026edge}), and map the model onto $n = N$ qubits via the Jordan--Wigner transformation. The ratio $\kappa_2/J_4$ serves as the integrability-breaking knob: for $\kappa_2 \to 0$ the quartic SYK$_4$ term dominates and the model is maximally chaotic, with level statistics of the Gaussian unitary ensemble (GUE), $\expval{r}_{\rm GUE} \approx 0.600$ (particle--hole symmetry is broken by construction, which selects the unitary class~\cite{kobayashi2026edge}); for large $\kappa_2$ the quadratic SYK$_2$ term reduces the model to free fermions with Poisson statistics, $\expval{r}_{\rm P} \approx 0.386$. Since $H$ conserves the total particle number, the level statistics are evaluated within the half-filling sector, over the central half of the spectrum. Panel~\textbf{a)} of Fig.~\ref{fig:syk_transition} charts the resulting transition at $n = 10$: the GUE plateau persists up to $\kappa_2/J_4 \approx 6$, and Poisson statistics are reached around $\kappa_2/J_4 \approx 300$. Following Ref.~\cite{kobayashi2026edge}, we mark the Wigner--Dyson (WD) onset, where $|\expval{r} - \expval{r}_{\rm GUE}|$ first exceeds $10^{-2}$ ($\kappa_2/J_4 \approx 6.9$) -- the parametric \emph{edge of chaos} of that reference -- and the Poisson onset, where $|\expval{r} - \expval{r}_{\rm P}|$ first drops below $10^{-2}$ ($\kappa_2/J_4 \approx 220$). Three representative reservoirs are used throughout: deep ergodic ($\kappa_2/J_4 = 0.94$), edge of chaos ($\kappa_2/J_4 = 8.8$, the first point of our coupling grid past the WD onset), and deep integrable ($\kappa_2/J_4 = 500$).

The QRP protocol is identical to that of the main text (Section~\ref{sec:conc_ising}), with a single accessible qubit ($n_a = 1$) re-initialized in $\ket{+}$, the exponentially scaled $Z$-rotation encoding of a fixed Mackey--Glass input window, and readout of Pauli-$X$ on the input qubit after $T = 20$ iterations -- deep on the stationary plateau, cf.\ Fig.~\ref{fig:syk_temporal}\textbf{a)}. Fresh couplings $\{J, \kappa\}$ are drawn for each reservoir instance, and the variance is taken over $N_{\rm res} = 300$ realizations, at the two evolution times used in Ref.~\cite{kobayashi2026edge}: $\Delta t = 50$, inside the universal random-matrix regime, and $\Delta t = 1$, short of it. We have repeated the analysis with the input--output pair of Ref.~\cite{kobayashi2026edge} itself, i.e., amplitude encoding $\sqrt{1-u}\ket{0} + \sqrt{u}\ket{1}$ on the input site with occupation readout $\expval{n_0}$, with the same qualitative outcome; we return to the quantitative difference below. (Neither encoding conserves the particle number -- the amplitude encoding coherently superposes charge sectors, as noted in Ref.~\cite{kobayashi2026edge} -- while the occupation \emph{readout} is charge-diagonal.)

\begin{figure*}[htbp]
    \centering
    \includegraphics[width=0.9\textwidth]{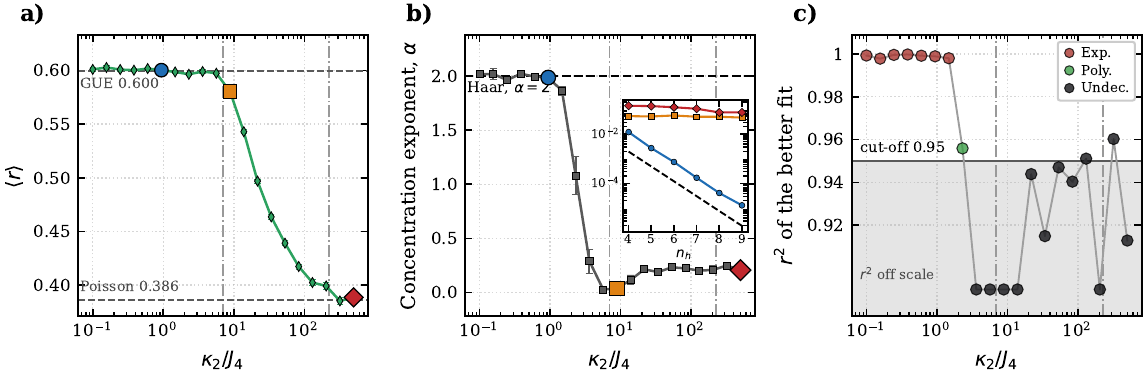}
    \caption{\textbf{Chaos transition and concentration of the SYK reservoir.} \textbf{a)} The mean level-spacing ratio $\expval{r}$ of the Hamiltonian of Eq.~\eqref{eq:syk_reservoir}, evaluated in the half-filling sector over the central half of the spectrum, as a function of the coupling ratio $\kappa_2/J_4$, for $n = 10$ fermionic modes ($300$ realizations per point). The statistics cross over from the GUE value $0.600$ to the Poisson value $0.386$ (horizontal dashed lines). The vertical dash-dotted lines mark the WD onset ($\kappa_2/J_4 \approx 6.9$; the parametric edge of chaos of Ref.~\cite{kobayashi2026edge}) and the Poisson onset ($\kappa_2/J_4 \approx 220$), defined by $|\expval{r} - r_{\rm th}| = 10^{-2}$. The colored symbols mark the three representative reservoirs used throughout: deep ergodic ($\kappa_2/J_4 = 0.94$, circle), edge of chaos ($8.8$, square), and deep integrable ($500$, diamond). \textbf{b)} The concentration exponent $\alpha$, obtained by fitting the readout variance to Eq.~\eqref{eq:alpha_fit} over $n_h \in \{4,\dots,9\}$, at $\Delta t = 50$ and $T = 20$ iterations ($n_a = 1$, $N_{\rm res} = 300$); error bars indicate the fit uncertainty, and the dashed line marks the extreme-scrambling limit $\alpha = 2$ of Theorem~\ref{Thm:HaaResVar}. The large symbols are the three reservoirs of panel a). Inset: the readout variance against $n_h$ for the same three reservoirs at $\Delta t = 50$, with the iterated Haar floor $1/(d_a d_h^2)$ (dashed): deep in the ergodic phase the variance decays at the Haar rate over three decades, while at the edge of chaos and in the integrable phase it is size-independent. \textbf{c)} Model selection for the fits of panel \textbf{b)}, by the criterion of Fig.~\ref{fig:ising_heatmap}\textbf{c)}: the $r^2$ of the better of the exponential and polynomial fits, with the marker coloured by the law selected and the region below the cut-off $r^2 = 0.95$ shaded. The exponential law is selected throughout the deep-ergodic plateau, $\kappa_2/J_4 \lesssim 1.5$; beyond $\kappa_2/J_4 \approx 4$ the variance is size-independent and no law is selected.}
    \label{fig:syk_transition}
\end{figure*}

\paragraph*{Concentration across the transition.}
Fig.~\ref{fig:syk_transition}~\textbf{b)} presents the central SYK result. Deep in the ergodic phase, the output concentrates at the extreme-scrambling rate: the fitted exponent $\alpha = 1.99 \pm 0.03$ reaches the Haar value $\alpha = 2$ of Theorem~\ref{Thm:HaaResVar} already at the moderate operating scales used here ($\Delta t = 50$, $T = 20$) -- strong concentration does not require the long-time idealization. As $\kappa_2$ grows, the exponent collapses over $\kappa_2/J_4 \approx 2$--$6$ and has essentially vanished \emph{at} the edge of chaos: $\alpha = 0.04 \pm 0.02$ at $\kappa_2/J_4 = 8.8$, where the variance is a size-independent constant ($\approx 5 \times 10^{-2}$ across the full range $n_h = 4$--$9$), just as deep in the integrable phase. The inset shows this explicitly: at $\Delta t = 50$ the deep-ergodic variance falls along the Haar slope, while at the edge of chaos and in the integrable phase it does not decay at all. From a design standpoint this is remarkably good news: the edge of chaos -- precisely the operating point that maximizes task performance in Ref.~\cite{kobayashi2026edge} -- does not suffer from output concentration at these operating scales, while the deeply chaotic phase concentrates maximally. Notably, the collapse of the exponent completes \emph{before} the spectral transition: $\alpha$ vanishes by $\kappa_2/J_4 \approx 6$ while $\expval{r}$ still sits on the GUE plateau, so at fixed $\Delta t$ the loss of concentration is not dictated by the level statistics alone. A natural reading is that the growing quadratic term absorbs an increasing share of the unit-normalized bandwidth and thereby stretches the effective scrambling time of the SYK$_4$ dynamics beyond the fixed evolution window: the parametric and temporal edges of chaos are intertwined in this reservoir. Consistent with this, at the short evolution time $\Delta t = 1$ the exponent is small and essentially coupling-independent ($\alpha \approx 0.5$--$0.6$ across the entire scan; not shown), mirroring the absence of a parametric performance feature at short times in Ref.~\cite{kobayashi2026edge}. Panel~\textbf{c)} confirms that the exponential form is selected only on the deep-ergodic side of the scan. Past $\kappa_2/J_4 \approx 4$ the variance is too flat in $n_h$ for any decay law to be identified, so the vanishing exponent at and beyond the edge of chaos records an absence of concentration rather than a poorly resolved rate.

The saturation of the Haar exponent deep in the chaotic phase also clarifies the role of the $U(1)$ symmetry. The Pauli-$X$ readout is off-diagonal in the conserved charge, and for it charge conservation evidently does not cap the concentration rate. For the charge-diagonal occupation readout $\expval{n_0}$ of Ref.~\cite{kobayashi2026edge}, by contrast, the deep-chaotic exponent saturates at $\alpha \approx 1.3$ (not shown): the symmetry restriction acts through the \emph{observable}, suppressing full Haar concentration only for readouts that commute with the conserved charge, in line with the symmetry discussion of Section~\ref{sec:discussion}.

\begin{figure*}[htbp]
    \centering
    \includegraphics[width=0.9\textwidth]{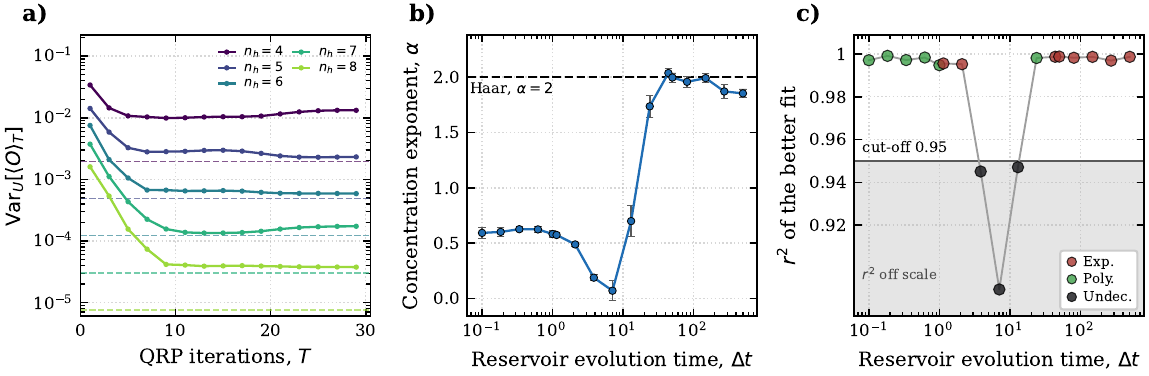}
    \caption{\textbf{Temporal dependence of the SYK concentration.} \textbf{a)} The variance of the QRP output against the number of iterations $T$ (shown at $T = 1, 3, \dots, 29$), deep in the ergodic phase ($\kappa_2/J_4 = 0.94$) at $\Delta t = 50$, for hidden-register sizes $n_h = 4,\dots,8$ ($n_a = 1$, $N_{\rm res} = 300$). The dashed line of each color marks the extreme-scrambling floor $1/(d_a d_h^2)$ for that $n_h$. After a decay over the first $\approx 10$ iterations the variance saturates a constant factor $\approx 5$--$7$ above the floor: the decay attains the Haar exponent, but not the Haar value. \textbf{b)} The concentration exponent $\alpha$ [fitted as in Fig.~\ref{fig:syk_transition}\textbf{b)}, over $n_h \in \{4,\dots,8\}$] as a function of the evolution time $\Delta t$, at the same coupling and $T = 20$. The exponent is small at short times ($\alpha \approx 0.6$ for $\Delta t \lesssim 1$), exhibits a non-monotonic dip around $\Delta t \approx 7$, and rises sharply to the Haar plateau $\alpha \approx 2$ for $\Delta t \gtrsim 40$; the evolution time $\Delta t = 50$ used in Fig.~\ref{fig:syk_transition} sits at the onset of this plateau. \textbf{c)} Model selection for the fits of panel \textbf{b)}, as in Fig.~\ref{fig:ising_heatmap}\textbf{c)}. The exponential law is selected on the Haar plateau, $\Delta t \gtrsim 40$. At short times both models fit well (each with $r^2 > 0.99$) and the polynomial is marginally preferred, while within the dip neither law is resolved.}
    \label{fig:syk_temporal}
\end{figure*}

\paragraph*{Dependence on the reservoir iterations and the evolution time.}
Fig.~\ref{fig:syk_temporal} resolves the two temporal knobs of the SYK reservoir at the deep-ergodic operating point. Panel~\textbf{a)} shows the iteration dependence: the variance decays over the first $\approx 10$ iterations and then saturates into a stable plateau, accumulating no further concentration up to the largest time simulated ($T = 50$). The $T = 1$ values agree with the single-iteration Haar variance $1/(d_a d_h + 1)$ of Appendix~\ref{app:var_small_t} to within $\approx 20\%$, and the subsequent iterations then drive the output down by a further factor $\sim d_h$: as for the other models of this study, the temporal structure of QRP makes the concentration strictly \emph{worse} than a single static reservoir pass. The saturated plateau decays with the full Haar exponent but sits a constant factor $\approx 5$--$7$ \emph{above} the extreme-scrambling floor $1/(d_a d_h^2)$: the variance level does not reach the Haar value -- in contrast with the kicked Ising reservoir of Appendix~\ref{app:kicked_ising}, whose plateau lands on the floor itself -- yet its exponential decay with $n_h$ is fully developed. Panel~\textbf{b)} shows the dependence on the evolution time. The concentration is weak at short times ($\alpha \approx 0.6$ for $\Delta t \lesssim 1$) and, after a non-monotonic dip around $\Delta t \approx 7$ whose origin we leave open, the exponent rises sharply and reaches the Haar plateau $\alpha \approx 2$ for $\Delta t \gtrsim 40$; the operating point $\Delta t = 50$ of Fig.~\ref{fig:syk_transition} sits at the onset of this plateau. As for the TFIM of the main text -- and in direct analogy with the kick number $K$ of Appendix~\ref{app:kicked_ising} -- restricting the evolution time below the scrambling time bounds the concentration rate even for a maximally chaotic Hamiltonian, identifying $\Delta t$ as a practical design knob. Panel~\textbf{c)} adds that the exponential form is cleanly selected on the plateau, whereas around the dip the two laws fit comparably well, so the detailed shape of that feature should not be over-interpreted.

\subsection{Kicked Ising model}\label{app:kicked_ising}

\paragraph*{Model and dynamical phases.}
As a Floquet counterpart of the Hamiltonian reservoirs above, we consider a periodically kicked, disordered Ising chain on a ring of $n = n_a + n_h$ qubits. One Floquet period reads
\begin{align}\label{eq:kicked_reservoir}
    U_F = \exp\!\Big(-i b \sum_i X_i\Big)\, \exp\!\Big(-i \Big[ J \sum_i Z_i Z_{i+1} + \sum_i D_i Z_i \Big]\Big)\;,
\end{align}
with fixed Ising coupling $J = 0.7$ and kick strength $b = 0.9$ -- a generic strongly kicked point bracketing the self-dual maximally chaotic line $J = b = \pi/4$ of the kicked Ising model~\cite{bertini2018exact} -- and longitudinal random fields $D_i$ drawn uniformly from $[-w, w]$, \emph{independently at every kick}. The reservoir applied at each QRP iteration is the depth-$K$ product 
\begin{align}
    U = U_F^{(K)}\, \cdots \; U_F^{(2)} U_F^{(1)} \;\;,
\end{align}
where the number of kicks $K$ plays the role of the evolution time in a discrete manner. This disorder model coincides with the on-site disorder implemented in digital Floquet experiments on superconducting processors~\cite{mi2022time}, and all parameters are $\OC(1)$ rotation angles, making this reservoir directly implementable on digital quantum hardware. Kicked Ising dynamics is, in fact, an actively developed reservoir substrate in its own right: closely related Floquet Ising circuits underlie recent quantum extreme learning experiments on superconducting processors~\cite{dao2026breaking} as well as reservoir-computing proposals based on discrete-time-crystal dynamics~\cite{zhang2026robust}. We caution that, while the circuit family is the same, the operating parameters of those works differ from ours (in the coupling and kick strengths as well as in the disorder protocol), so the results below characterize the strongly kicked, per-kick-randomized point of Eq.~\eqref{eq:kicked_reservoir} rather than those experimental settings. The QRP protocol is identical to that of the main text (Section~\ref{sec:conc_ising}), with a single accessible qubit ($n_a = 1$) re-initialized in $\ket{+}$ and encoded by $e^{-i s_\tau Z}$, a fixed Mackey--Glass input window, and readout of Pauli-$X$ on the accessible qubit; the disorder sequence $\{D_i^{(k)}\}$ is redrawn for each reservoir instance and the variance is taken over $N_{\rm res} = 300$ realizations.

The random-field width $w$ is both the integrability-breaking knob and the ensemble source. At $w = 0$ the model reduces to the integrable kicked transverse-field Ising chain (and the reservoir becomes deterministic, so the variance over the ensemble vanishes trivially); increasing $w$ drives the model into the quantum-chaotic phase. Since a Floquet system has no static Hamiltonian, the level statistics are evaluated on the quasi-energies of the single-period unitary $U_F$, whose gap ratios cross over from Poisson to the circular orthogonal ensemble (COE) value $\expval{r}_{\rm COE} \approx 0.531$.

\begin{figure*}[htbp]
    \centering
    \includegraphics[width=0.9\textwidth]{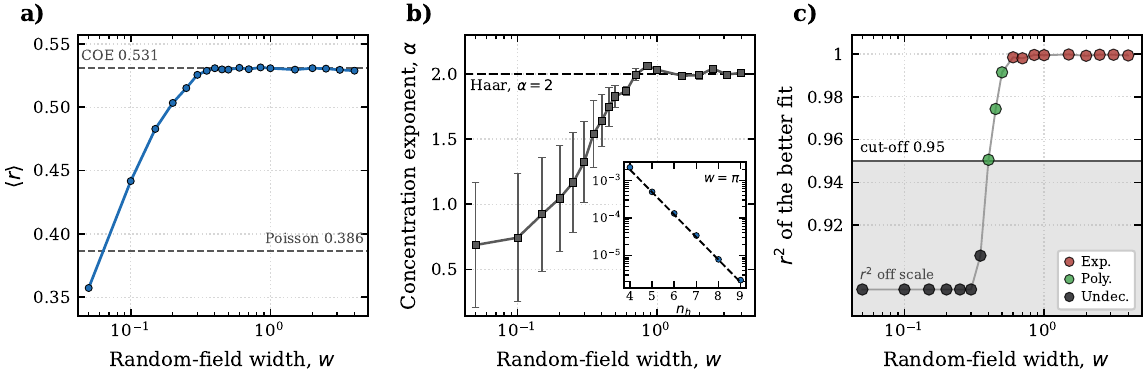}
    \caption{\textbf{Chaos transition and concentration onset of the kicked Ising reservoir.} \textbf{a)} The mean gap ratio $\expval{r}$ of the single-period Floquet unitary $U_F$ [Eq.~\eqref{eq:kicked_reservoir}] as a function of the random-field width $w$, for a system of $n = 10$ qubits ($150$ realizations per point; logarithmic axis -- the deterministic point $w = 0$, at which the reservoir ensemble collapses to a single unitary, is omitted). The statistics cross over from Poisson (with a sub-Poisson dip at $w \approx 0.05$ reflecting the clustered quasi-energies of the weakly perturbed integrable model) to the COE value $0.531$, saturating by $w \approx 0.4$. \textbf{b)} The concentration exponent $\alpha$, obtained by fitting the variance of the QRP output to Eq.~\eqref{eq:alpha_fit} over $n_h \in \{4,\dots,9\}$, at depth $K = 12$ and $t = 20$ iterations ($n_a = 1$, $N_{\rm res} = 300$); error bars indicate the fit uncertainty. The exponent rises across the chaos transition and saturates at the extreme-scrambling value $\alpha = 2$ of Theorem~\ref{Thm:HaaResVar} for $w \gtrsim 0.7$. Inset: the variance against $n_h$ deep in the chaotic phase ($w = \pi$) lies on the Haar prediction $1/(d_a d_h^2)$ (dashed line) over three decades -- the prediction is attained in value, not only in scaling. \textbf{c)} Model selection for the fits of panel \textbf{b)}, as in Fig.~\ref{fig:ising_heatmap}\textbf{c)}: the exponential law is selected for $w \gtrsim 0.6$, once the spectral crossover is complete, while for $w \lesssim 0.35$ the variance is not yet a decaying function of $n_h$ and no law is selected.}
    \label{fig:kicked_transition}
\end{figure*}

\paragraph*{Chaos transition and concentration results.}
Fig.~\ref{fig:kicked_transition} shows the phase diagnostic and the concentration across the transition. The gap ratio crosses over from (sub-)Poisson to the COE value as $w$ increases, completing the crossover by $w \approx 0.4$ at $n = 10$ [panel \textbf{a)}]. The crossover point moves toward smaller $w$ as the system grows (not shown), consistent with the ergodicity of the kicked Ising model at arbitrarily weak longitudinal fields in the thermodynamic limit~\cite{bertini2018exact}.

The concentration exponent tracks this transition [panel \textbf{b)}]. Near the integrable side the fits give $\alpha \approx 0.7$, with large uncertainties: in this weakly random regime the variance is not yet a clean exponential in $n_h$. Panel~\textbf{c)} makes this explicit: the exponential law is selected only for $w \gtrsim 0.6$, with a narrow band at $w \approx 0.4$--$0.5$ in which the polynomial decay is the better description. Even there the reservoir concentrates partially, because the per-kick redrawing of the disorder accumulates randomness over the $K = 12$ kicks; the truly frozen limit is $w \to 0$. Across the chaotic crossover the exponent rises steeply and, for $w \gtrsim 0.7$, saturates at the extreme-scrambling value $\alpha = 2$. Note that the exponent saturates somewhat later than the spectral statistics: as for the TFIM of the main text, spectral ergodicity is necessary but not by itself sufficient for full Haar concentration at fixed depth and iteration count. In contrast with the TFIM and SYK reservoirs -- which approach but do not reach the design prediction -- the strongly kicked Ising reservoir thus \emph{saturates} the Haar limit throughout its deeply chaotic phase. We attribute this to the disorder being quenched anew at every kick within each reservoir's dynamics: every reservoir instance is itself a genuine local random circuit, and local random circuits converge to approximate unitary $2$-designs with depth~\cite{harrow2009random,brandao2016local}, with no conserved charge impeding the convergence here (unlike the $U(1)$ symmetry that caps the SYK exponent in Appendix~\ref{app:syk}). The inset of panel \textbf{b)} shows this explicitly: at $w = \pi$ the measured variance falls on the Haar value $1/(d_a d_h^2)$ itself, across all simulated sizes.

\begin{figure*}[htbp]
    \centering
    \includegraphics[width=0.9\textwidth]{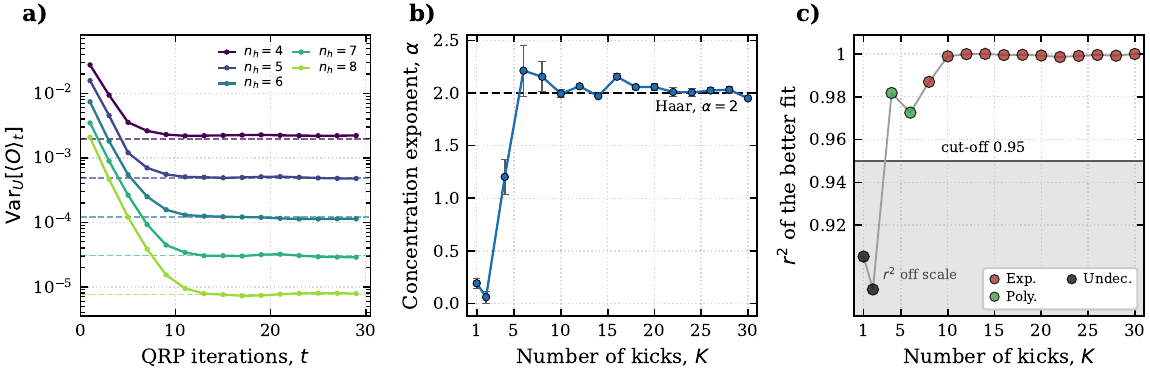}
    \caption{\textbf{Temporal dependence of the kicked Ising concentration.} \textbf{a)} The variance of the QRP output against the number of QRP iterations $t$ (shown at $t = 1, 3, \dots$), deep in the chaotic phase ($w = \pi$) at depth $K = 12$, for hidden-register sizes $n_h = 4, \dots, 8$ ($n_a = 1$, $N_{\rm res} = 300$). The dashed line of each color marks the extreme-scrambling floor $1/(d_a d_h^2)$ for that $n_h$. The variance decays over the first $\approx 10$ iterations and then saturates on the floor; the $t = 1$ values agree with the single-iteration Haar variance $1/(d_a d_h + 1)$ (cf.\ the small-$t$ analysis in Appendix~\ref{app:var_small_t}). \textbf{b)} The concentration exponent $\alpha$ [fitted as in Fig.~\ref{fig:kicked_transition}\textbf{b)}, over $n_h \in \{4,\dots,8\}$] as a function of the number of kicks $K$, at $w = \pi$ and $t = 20$ iterations; for the kicked reservoir, $K$ plays the role of the evolution time $\Delta t$ of the Hamiltonian models. For $K \lesssim 3$ the readout light cone does not yet cover the system, so the variance is a size-independent local quantity and $\alpha \approx 0$ (under-scrambled regime, not to be mistaken for concentration). The exponent reaches the extreme-scrambling value $\alpha = 2$ by $K \approx 6$ and remains there up to the largest depth studied, $K = 30$. \textbf{c)} Model selection for the fits of panel \textbf{b)}, as in Fig.~\ref{fig:ising_heatmap}\textbf{c)}: no law is selected for $K \lesssim 2$, where the variance is independent of $n_h$, and the exponential law is selected from $K \approx 8$ onwards.}
    \label{fig:kicked_temporal}
\end{figure*}

\paragraph*{Dependence on the reservoir iterations and the number of kicks.}
Fig.~\ref{fig:kicked_temporal} resolves the two temporal knobs of the kicked reservoir: the number of QRP iterations $t$, and the number of kicks $K$ within each iteration. The latter is this model's analogue of the evolution time -- $K$ stands to the Floquet reservoir as $\Delta t$ stands to the Hamiltonian reservoirs of the main text and of Appendices~\ref{app:rydberg} and~\ref{app:syk}.

Panel \textbf{a)} shows the iteration dependence. As for every other model in this study, the variance relaxes within roughly ten iterations to a stable plateau and accumulates no further concentration beyond it; at the converged depth $K = 12$ the plateau coincides with the extreme-scrambling floor for every register size. The temporal structure of QRP makes the concentration \emph{worse} than in the static, single-pass setting: at $t = 1$ the measured variance agrees with the single-iteration Haar value $1/(d_a d_h + 1) \sim 1/(d_a d_h)$ (Appendix~\ref{app:var_small_t}), while the repeated reservoir iterations drive the output down to the floor $1/(d_a d_h^2)$ -- smaller by a further factor of $d_h$. Recurrent operation is thus strictly more concentration-prone than a single reservoir pass.

Panel \textbf{b)} resolves the onset of concentration along the depth axis. For $K \lesssim 3$ the fitted exponent is near zero: after $K$ kicks the readout observable is supported only on the $\OC(K)$-site light cone of the readout qubit, so the variance is a local quantity, independent of $n_h$ -- a trivially frozen regime rather than genuine concentration (a useful practical diagnostic is precisely this $n_h$-independence). The exponent then rises steeply, reaches the Haar value $\alpha = 2$ by $K \approx 6$ and, after a weak transient associated with the scrambling front wrapping the ring, settles there from $K \approx 10$ up to the largest depth studied ($K = 30$) -- confirming $K = 12$ as a converged operating point. The kicked model thus needs only a modest, size-comparable number of kicks to realize extreme-scrambling behavior. Panel~\textbf{c)} separates the two regimes: for $K \lesssim 2$ no decay law is selected at all, as the light-cone argument requires, whereas from $K \approx 8$ the exponential law is selected at every depth. Together, Figs.~\ref{fig:kicked_transition} and~\ref{fig:kicked_temporal} establish the kicked Ising chain as an explicit, hardware-realizable route by which a physical reservoir attains -- and here saturates -- the extreme-scrambling concentration of Theorem~\ref{Thm:HaaResVar}, and, by the same token, as a regime that scalable reservoir designs should avoid.

\subsection{Further analyses of the Ising reservoir}

\subsubsection{Memory decay of chaotic Ising reservoirs.}
Beyond output concentration, the memory properties of Ising reservoirs -- a workhorse platform for QRP~\cite{mujal2021opportunities,martinez2020information,domingo2023optimal,palacios2024role,domingo2022optimal,kobayashi2023quantum,kobayashi2024quantum,kutvonen2020optimizing} -- display an analogous dependence on the degree of scrambling. As a complementary probe for a slightly different model, we consider a chaotic tilted-field Ising chain, chaotic Ising dynamics having been found to be more expressive than their integrable counterparts~\cite{xiong2025fundamental}, and examine how the memory indicator $\MC_a$ of Section~\ref{sec:exp_mem_indicators} decays with the number of time steps.

\begin{figure}[ht]
    \centering
    \includegraphics[width=0.48\textwidth]{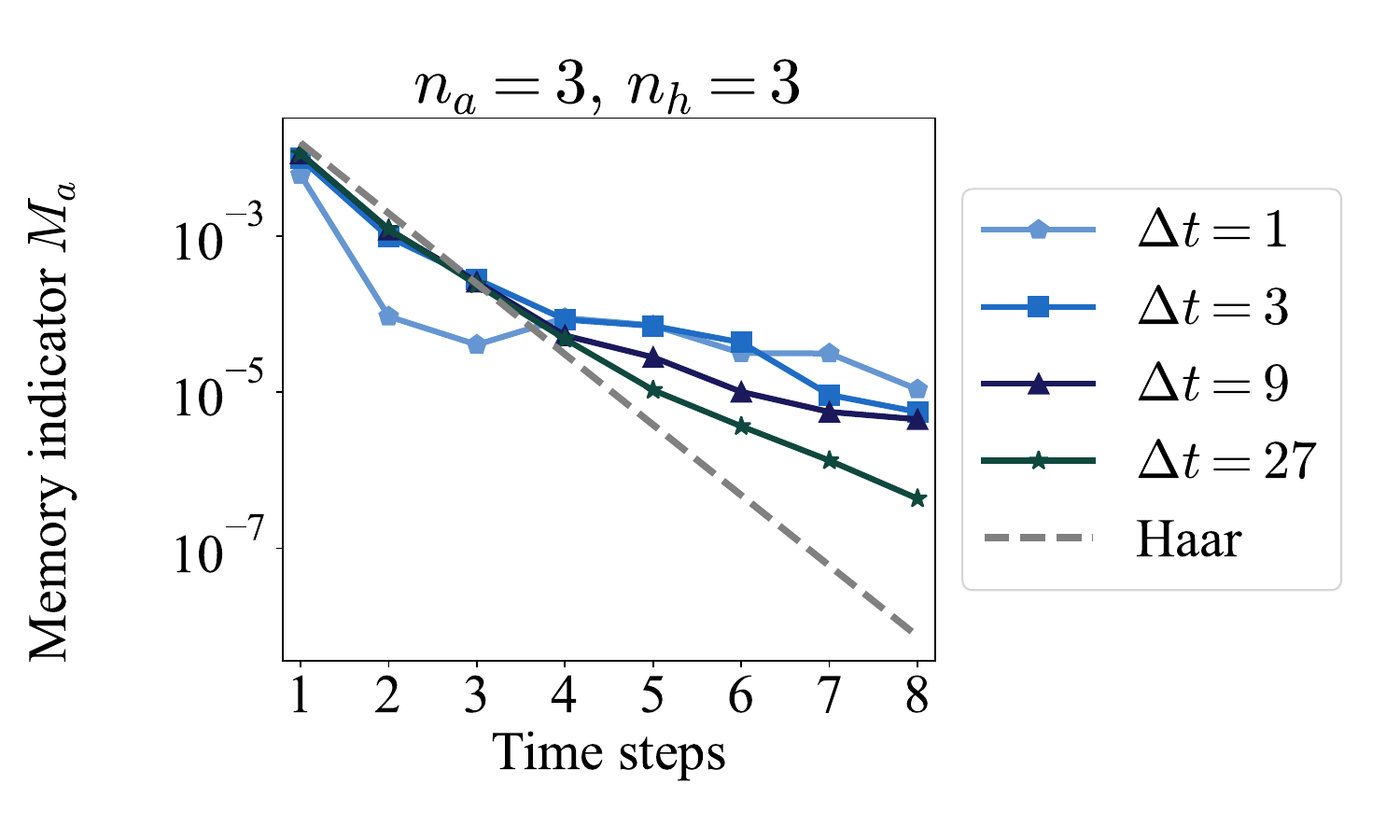}
    
    \caption{{\textbf{Memory indicator of Ising reservoirs.} The dependence $\MC_a$ of outputs on inputs are plotted against number of time steps for different reservoirs. We consider reservoirs with different dynamics --sampled from Haar random unitaries and Ising models with different evolution times $\Delta t$, where the Hamiltonian is $H_{\rm{Ising}}=J\sum_{i=1}^{n-1}Z_i{}Z_{i+1}+B_z\sum_i^{n}Z_i+B_x\sum_i^n{}X_i\,,$ with $J=-1,\ B_x=0.7,\ B_z=1.5\,$.
    }}
    \label{fig:inp_dep_comp}
\end{figure}
As shown in Fig.~\ref{fig:inp_dep_comp}, increasing the evolution time $\Delta t$ makes the memory indicator decay faster with the number of time steps, approaching the exponential decay of a Haar-random reservoir established in Section~\ref{sec:exp_mem_indicators}. The evolution time thus tunes the reservoir continuously between a slowly forgetting, weakly scrambling regime and the rapidly forgetting extreme-scrambling limit, consistent with the concentration results above.

\subsubsection{Role of decaying interaction: Rydberg atoms}\label{app:rydberg}
Another representative physical system having Ising-like dynamics is the Rydberg-atom arrays. Rydberg-atom arrays provide a programmable many-body reservoir through their strong, geometry-dependent interactions and global or local laser controls.
They were initially proposed as quantum recurrent and reservoir systems for
multitasking, decision-making, and long-term-memory tasks
\cite{bravo2022quantum}. More recent work has demonstrated large-scale
reservoir learning on analog neutral-atom hardware, including classification
and time-series prediction with systems of up to \(108\) atoms
\cite{kornjaca2024large}. Rydberg reservoirs have also been studied using
tensor-network simulations of one-dimensional chains and under experimentally
relevant constraints such as decoherence, finite sampling, and measurement
overhead \cite{solanki2025harnessing,liu2026practical}. Their dynamics are
commonly described by
\begin{equation}
    H = \frac{\Omega}{2} \sum_i \sigma_x^{(i)}
    - \Delta \sum_i n_i
    + \sum_{i<j} V_{ij} n_i n_j,
\end{equation}
where \(n_i=\lvert r_i\rangle\langle r_i\rvert\) projects atom \(i\) onto its
Rydberg state and \(V_{ij}=C_6/R_{ij}^{6}\) is the van der Waals interaction between atoms separated by \(R_{ij}\). Inputs can be encoded through the Rabi drive \(\Omega\), the detuning \(\Delta\), or the atomic geometry, while the resulting interacting dynamics generate the reservoir features used by the classical readout.

We consider a neutral atom array driven by a laser field, where the atoms interact via van der Waals forces. Treating each atom as a two-level system with a ground state $|g\rangle$ and a highly excited Rydberg state $|r\rangle$, the system's dynamics are governed by the standard Rydberg Hamiltonian,
\begin{equation}
    H = \frac{\Omega}{2} \sum_i \sigma_x^{(i)} - \Delta \sum_i n_i + \sum_{i<j} V_{ij} n_i n_j,
\end{equation}
where $\Omega$ is the Rabi frequency, $\Delta$ is the laser detuning, $V_{ij}$ is the interaction strength between atoms $i$ and $j$, and $n_i = |r\rangle\langle r|_i$ is the atomic number operator. 

To establish a mapping to a transverse-field Ising model, we treat each atom as a pseudo-spin-$1/2$ system by identifying $|g\rangle \equiv |\downarrow\rangle$ and $|r\rangle \equiv |\uparrow\rangle$. In this basis, the number operator can be expressed in terms of the Pauli-$z$ matrix as $n_i = (1 + \sigma_z^{(i)})/2$. Substituting this into the interaction term gives
\begin{equation}
    n_i n_j = \frac{1}{4} \left( 1 + \sigma_z^{(i)} + \sigma_z^{(j)} + \sigma_z^{(i)} \sigma_z^{(j)} \right).
\end{equation}

By replacing $n_i$ in the full Hamiltonian and regrouping the terms, we recover the familiar transverse-field Ising Hamiltonian,
\begin{equation}
    H_{\text{Ising}} = \sum_{i<j} J_{ij} \sigma_z^{(i)} \sigma_z^{(j)} + \sum_i h_i \sigma_z^{(i)} + \frac{\Omega}{2} \sum_i \sigma_x^{(i)} + E_0.
\end{equation}
Here, we can explicitly identify the effective Ising parameters: the spin-spin coupling $J_{ij} = V_{ij}/4$, the local longitudinal magnetic field $h_i = -\Delta/2 + \frac{1}{4} \sum_{j \neq i} V_{ij}$, and a constant energy offset $E_0 = -N\Delta/2 + \frac{1}{4} \sum_{i<j} V_{ij}$. Crucially, the native van der Waals interaction $V_{ij} = C_6/r_{ij}^6$ translates directly into a long-range, power-law decaying spin-spin coupling ($J_{ij} \propto 1/r_{ij}^6$).

\begin{figure}
    \centering
    \includegraphics[width=0.5\linewidth]{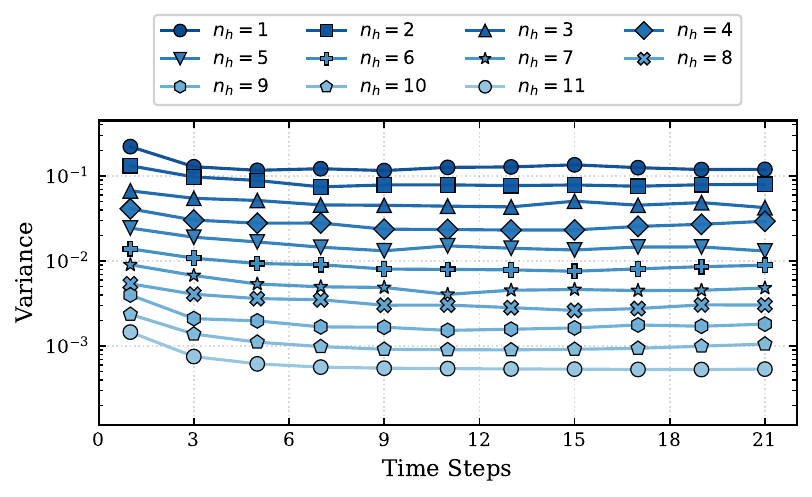}
    \caption{Concentration of QRP with 1D Rydberg-atoms as reservoir. The output variances of reservoirs with different sizes are plotted against the number of reservoir steps. Here $\Omega=1.0$, $\Delta=1$, $V_{ij}=3$, and $\Delta t=10^6$.}
    \label{fig:conc_rydberg_t_nh}
\end{figure}

As exemplified in Fig.~\ref{fig:conc_rydberg_t_nh}, the Rydberg-atoms-reservoirs may also exhibit exponential concentrations in some parameters regimes. Similar to exponential concentrations identified in extreme scrambling reservoirs and TFIM reservoirs in ergodic phases, here the variance of output also decays exponentially with system size. However, compared to TFIM reservoirs, here the evolution time required for strong concentration is substantially longer.

}

\end{document}